\documentclass[a4paper,10pt,UKenglish]{article}

\usepackage{ifthen}
\usepackage{comment}
\makeatletter
\usepackage[sort,comma,square,numbers]{natbib}
\renewcommand*{\NAT@spacechar}{~} %
\renewcommand\bibsection %
{
  \section*{\refname
    \@mkboth{\MakeUppercase{\refname}}{\MakeUppercase{\refname}}}
}

\usepackage{color,tikz,environ}
\usetikzlibrary{decorations.markings}

\usepackage{etoolbox}


\newcommand{\boundellipse}[3]%
{(#1) ellipse (#2 and #3)
}

\usepackage{fullpage}

\ifthenelse{\isundefined{\lipics}}
{
\usepackage{graphicx}
\usepackage[colorlinks=true,bookmarks=false]{hyperref}
\usepackage[small,bf]{caption}
\usepackage{subfigure}
\usepackage[british]{babel}
}
{}
\usepackage{xcolor}
\definecolor{darkblue}{rgb}{0,0,0.45}
\definecolor{darkred}{rgb}{0.6,0,0}
\definecolor{darkgreen}{rgb}{0.13,0.5,0}
\hypersetup{colorlinks, linkcolor=darkblue, citecolor=darkgreen,
urlcolor=darkblue}

\ifthenelse{\isundefined{\llncs}}{
  \usepackage{amsthm}
  \usepackage{authblk}
  
}

\usepackage{amsmath,amsfonts,amssymb}
\usepackage{mathtools}
\usepackage{mathrsfs}
\usepackage{wrapfig}
\usepackage{xspace}
\usepackage[shortlabels]{enumitem}
\setlist[enumerate]{nosep} %
\setlist[itemize]{nosep} %

\usepackage{multirow}

\usepackage[protrusion=true]{microtype}

\usepackage{aliascnt}
\ifthenelse{\isundefined{\llncs}}{
  \theoremstyle{plain}
}{}
\newtheorem{thm}{Theorem}[section]
\newaliascnt{lem}{thm}
\newaliascnt{crl}{thm}
\newaliascnt{dfn}{thm}
\newaliascnt{claim}{thm}
\newaliascnt{prop}{thm}
\newaliascnt{remark}{thm}
\newaliascnt{hyp}{thm}
\newtheorem{lem}[lem]{Lemma}
\newtheorem{claim}[claim]{Claim}

\newtheorem{crl}[crl]{Corollary}
\newtheorem{remark}[remark]{Remark}
\newtheorem{hyp}[hyp]{Hypothesis}
\ifthenelse{\isundefined{\llncs}}{
  \theoremstyle{definition}
}{}
\newtheorem{dfn}[dfn]{Definition}
\aliascntresetthe{lem}
\aliascntresetthe{remark}
\aliascntresetthe{crl}
\aliascntresetthe{prop}
\aliascntresetthe{dfn}
\aliascntresetthe{claim}
\aliascntresetthe{hyp}

\newcommand{\mc}[1]{{\mathcal{#1}}}

\newcommand{\eps}{{\varepsilon}}
\newcommand{\cqed}{\renewcommand{\qed}{\hfill$\lrcorner$}}
\DeclareMathOperator{\cost}{cost}
\DeclareMathOperator*{\val}{val}

\newcommand{\hy}{\hbox{-}\nobreak\hskip0pt}

\newcommand{\ignore}[1]{}
\usepackage{environ,xstring}
\newif\iflabel
\newif\ifdbs
\newif\ifamp
\NewEnviron{doitall}{%
  \noexpandarg
  \expandafter\IfSubStr\expandafter{\BODY}{\label}{\labeltrue}{\labelfalse}%
  \expandafter\IfSubStr\expandafter{\BODY}{\\}{\dbstrue}{\dbsfalse}%
  \expandafter\IfSubStr\expandafter{\BODY}{&}{\amptrue}{\ampfalse}%
  \iflabel\def\doitallstar{}\else\def\doitallstar{*}\fi
  \ifdbs
    \ifamp
      \def\doitallname{align}%
    \else
      \def\doitallname{multline}%
    \fi
  \else
    \def\doitallname{equation}
  \fi
  \begingroup\edef\x{\endgroup
    \noexpand\begin{\doitallname\doitallstar}%
    \noexpand\BODY
    \noexpand\end{\doitallname\doitallstar}%
  }\x
}
\def\[#1\]{\begin{doitall}#1\end{doitall}}

\newcommand{\pname}[1]{\textsc{#1}}

\usepackage{thm-restate}

\usepackage{framed}

\newcommand{\ud}[1]{\overline{#1}}
\newcommand{\polyn}{\cdot n^{O(1)}}
\newcommand{\dsn}{\textsc{DSN}\xspace}
\newcommand{\dsnK}{\textsc{DSN$_\mc{K}$}\xspace}
\newcommand{\dsnP}{\textsc{DSN$_\textsc{Planar}$}\xspace}
\newcommand{\bidsn}{\textsc{bi\hy{}DSN}\xspace}
\newcommand{\altbidsn}{\mdseries{\bidsn}\xspace}

\newcommand{\bidsnP}{\textsc{bi\hy{}DSN$_\textsc{Planar}$}\xspace}
\newcommand{\altbidsnP}{\mdseries{\bidsnP}\xspace}
\newcommand{\scss}{\textsc{SCSS}\xspace}
\newcommand{\altscss}{\mdseries{\scss}\xspace}

\newcommand{\scssP}{\textsc{SCSS$_\textsc{Planar}$}\xspace}
\newcommand{\biscss}{\textsc{bi\hy{}SCSS}\xspace}

\newcommand{\altbiscss}{\mdseries{\biscss}\xspace}
\newcommand{\biscssP}{\textsc{bi\hy{}SCSS$_\textsc{Planar}$}\xspace}
\newcommand{\altbiscssP}{\mdseries{\biscssP}\xspace}
\newcommand{\biscssK}{\textsc{bi\hy{}SCSS$_\mc{K}$}\xspace}
\newcommand{\csi}{\textsc{Colored Subgraph Isomorphism}\xspace}
\newcommand{\mcsi}{\textsc{Maximum Colored Subgraph Isomorphism}\xspace}

\newcommand{\gt}{\textsc{Grid Tiling}\xspace}

\newcommand{\poly}{\text{poly}\xspace}

\newcommand{\hamc}{\textsc{Hamiltonian Cycle}\xspace}
\newcommand{\nextone}{\textbf{next}\xspace}
\newcommand{\prevone}{\textbf{prev}\xspace}
\newcommand{\HS}{\text{VS}\xspace}
\newcommand{\VS}{\text{HS}\xspace}
\newcommand{\M}{\text{M}\xspace}
\newcommand{\mcc}{\textsc{Multicolored Clique}\xspace}

\date{}
\title{Parameterized Approximation Algorithms for\\
Bidirected Steiner Network Problems}

\author[1]{Rajesh~Chitnis\thanks{Supported by ERC grant 2014-CoG 647557. Part of 
this work was done while at Weizmann Institute of Science, Israel (and supported 
by Israel Science Foundation grant
\#897/13) and visiting Charles University in
Prague, Czechia}}
\author[2]{Andreas~Emil~Feldmann\thanks{Supported by the Czech Science
Foundation GA{\v C}R (grant \#19-27871X), and by the Center for Foundations of 
Modern Computer Science (Charles Univ.\ project UNCE/SCI/004).}}
\author[3]{Pasin~Manurangsi\thanks{This work was done while the author was
visiting Weizmann Institute of Science. Currently at Google Research.}}
\affil[1]{University of Birmingham, UK.
\texttt{rajeshchitnis@gmail.com}}
\affil[2]{Charles University, Czechia. 
\texttt{feldmann.a.e@gmail.com}}
\affil[3]{University of California, Berkeley, USA.
\texttt{pasin@berkeley.edu.}}

\begin{document}
\renewcommand*{\sectionautorefname}{Section}
\renewcommand*{\subsectionautorefname}{Section}
\renewcommand*{\subsubsectionautorefname}{Section}

\maketitle

\begin{abstract}
The \pname{Directed Steiner Network (DSN)} problem takes as input a directed
graph $G=(V,E)$ with non-negative edge-weights and a set $\mc{D}\subseteq
V\times V$ of $k$ demand pairs. The aim is to compute the cheapest network
$N\subseteq G$ for which there is an $s\to t$ path for each $(s,t)\in\mc{D}$. It
is known that this problem is notoriously hard as there is no
$k^{1/4-o(1)}$\hy{}approximation algorithm under Gap-ETH, even when
parametrizing the runtime by $k$ [Dinur \& Manurangsi, ITCS 2018]. In light of
this, we systematically study several special cases of \dsn and determine their
parameterized approximability for the parameter $k$.

For the \bidsnP problem, the aim is to compute a solution $N\subseteq G$ whose 
cost is at most that of an optimum planar solution in a bidirected graph $G$, 
i.e., for every edge $uv$ of $G$ the reverse edge $vu$ exists and has the same 
weight. This problem is a generalization of several well-studied special cases. 
Our main result is that this problem admits a parameterized approximation scheme 
(PAS) for~$k$. We also prove that our result is tight in the sense that (a) the 
runtime of our PAS cannot be significantly improved, and (b) no PAS exists for 
any generalization of \bidsnP, under standard complexity assumptions. The 
techniques we use also imply a polynomial-sized approximate kernelization scheme 
(PSAKS). Additionally, we study several generalizations of \bidsnP and obtain 
upper and lower bounds on obtainable runtimes parameterized by~$k$.

One important special case of \dsn is the \pname{Strongly Connected Steiner
Subgraph (SCSS)} problem, for which the solution network $N\subseteq G$ needs to
strongly connect a given set of $k$ terminals. It has been observed before that
for \scss a parameterized $2$-approximation exists for parameter~$k$
[Chitnis et~al., IPEC 2013]. We give a tight inapproximability result by showing
that for $k$ no parameterized $(2-\eps)$\hy{}approximation algorithm exists
under Gap-ETH. Additionally, we show that when restricting the input of \scss
to bidirected graphs, the problem remains NP-hard but becomes FPT for $k$.
\end{abstract}

\newpage

\section{Introduction}

In this work we study the \pname{Directed Steiner Network (DSN)}
problem,\footnote{Also sometimes called \pname{Directed Steiner Forest}. Note
however that in contrast to the undirected \pname{Steiner Forest} problem, an
optimum solution to \dsn is not necessarily a forest.}
in which a directed graph~$G=(V,E)$ with non-negative edge weights is given
together with a set of $k$ \emph{demands} $\mc{D}=\{(s_i,t_i)\}_{i=1}^k\subseteq
V\times V$. The aim is to compute a minimum cost (in terms of edge weights)
network $N\subseteq G$ containing a directed~$s_i\to t_i$ path for each
$i\in\{1,\ldots,k\}$. This problem has applications in network
design~\cite{kerivin2005design}, and for instance models the setting where nodes
in a radio or ad-hoc wireless network connect to each other
unidirectionally~\cite{chen1989strongly,ramanathan2000topology}.

The \dsn problem is notoriously hard. First of all, it is NP-hard, and one
popular way to handle NP-hard problems is to efficiently compute an
\emph{$\alpha$-approximation}, i.e., a solution that is guaranteed to be at most
a factor~$\alpha$ worse than the optimum. For this paradigm we typically demand
that the algorithm computing such a solution runs in polynomial time in the
input size~$n=|V|$. However for \dsn it is known that even computing  an
$O(2^{\log^{1-\eps}n})$\hy{}approximation is not possible~\cite{dodis1999design}
in polynomial time, unless NP~$\subseteq$~DTIME$(n^{\text{polylog}(n)})$. It is
possible to obtain approximation factors $O(n^{2/3+\eps})$ and $O(k^{1/2+\eps})$
though~\cite{berman2013approximation,DBLP:journals/talg/ChekuriEGS11,FKN12}. For
settings where the number~$k$ of demands is fairly small, one may aim for
algorithms that only have a mild exponential runtime blow-up in $k$, i.e., a
runtime of the form $f(k)\polyn$, where~$f(k)$ is some function independent
of~$n$. If an algorithm computing the optimum solution with such a runtime
exists for a computable function~$f(k)$, then the problem is called
\emph{fixed-parameter tractable (FPT)} for parameter~$k$. However it is unlikely
that \dsn is FPT for this well-studied parameter, as it is known to be
W[1]-hard~\cite{DBLP:journals/siamdm/GuoNS11} when parameterized by $k$. In fact
one can show~\cite{DBLP:conf/soda/ChitnisHM14} that under the \emph{Exponential
Time Hypothesis (ETH)} there is no algorithm computing the optimum in time
$f(k)\cdot n^{o(k)}$ for any function $f(k)$ independent of $n$. ETH assumes
that there is no $2^{o(n)}$ time algorithm to solve
\pname{3SAT}~\cite{ImpagliazzoP01,ImpagliazzoPZ01}. The best we can hope for is
therefore a so-called \emph{XP-algorithm} computing the optimum in time
$n^{O(k)}$, and this was also shown to exist
by~\citet{DBLP:journals/siamcomp/FeldmanR06}.

None of the above algorithms for \dsn seem satisfying though, either due to
slow runtimes or large approximation factors. To circumvent the hardness
of the problem, one may aim for \emph{parameterized approximations}, which have
recently received increased attention for various problems (cf.\ the recent
survey in~\cite{DBLP:journals/algorithms/FeldmannSLM20}).
In this paradigm an $\alpha$\hy{}approximation is computed in time $f(k)\polyn$
for parameter~$k$, where $f(k)$ again is a computable function independent
of~$n$. Unfortunately, a recent result by~\citet{DM18} excludes significant
improvements over the known polynomial time approximation
algorithms~\cite{berman2013approximation,DBLP:journals/talg/ChekuriEGS11,FKN12},
even if allowing a runtime parameterized in~$k$. More specifically, no
$k^{1/4-o(1)}$-approximation is possible\footnote{In a previous version~\cite{DBLP:conf/esa/ChitnisFM18} of this
work, we showed that no $k^{o(1)}$-approximation is possible for DSN in time
$f(k)\polyn$. This result in now subsumed by~\cite{DM18}; see \autoref{app:dsn}
for more details.} in time $f(k)\polyn$ for any function $f(k)$ under the
\emph{Gap Exponential Time Hypothesis (Gap-ETH)}, which postulates that there
exists a constant $\eps>0$ such that no (possibly randomized) algorithm running
in $2^{o(n)}$ time can distinguish whether all or at most a $(1-\eps)$-fraction
of clauses of any given \pname{3SAT} formula can be satisfied\footnote{Gap-ETH
follows from ETH given other standard conjectures, such as the existence of
linear sized PCPs or exponentially-hard locally-computable one-way functions.
See~\cite{param-inapprox,Applebaum17} for more details.}~\cite{Dinur16,MR17}.

Given these hardness results, the main question we explore is: what
approximation factors and runtimes are possible for special cases of \dsn when
parameterizing by $k$? There are two types of standard special cases that are
considered in the literature:
\begin{itemize}
\item Restricting the input graph $G$ to some special graph class. A typical
assumption for instance is that $G$ is planar (where a directed graph is
planar if the underlying undirected graph is).
\item Restricting the pattern of the demands in $\mc{D}$. For example, one
standard restriction is to have a set $R\subseteq V$ of \emph{terminals}, a
fixed \emph{root} $r\in R$, and demand set $\mc{D}=\{(r,t)\mid t\in R\}$, which
is the well-known \pname{Directed Steiner Tree (DST)} problem.
\end{itemize}

In fact, an optimum solution to the \pname{DST} problem is an arborescence
(hence the name), i.e., it is planar. Thus if an algorithm is able to compute a
solution that costs at most as much as the cheapest planar \dsn solution in an
otherwise unrestricted graph, it can be used for both the above types of
restrictions: it can of course be used if the input graph is planar as well, and
it can also be used if the demand pattern implies that the optimum must be
planar. Taking the structure of the optimum solution into account has been a
fruitful approach leading to several results on related problems, both for
approximation and fixed-parameter tractability, from which we also draw some of
the inspiration for our results (cf.~\autoref{sec:techniques}). A main focus of
our work is to systematically explore the influence of the structure of
solutions on the complexity of the \dsn problem. Formally, fixing a class
$\mc{K}$ of graphs, we define the \dsnK problem, which asks for a solution
network $N\subseteq G$ for $k$ given demands such that the cost of $N$ is at
most that of an optimum solution in $G$ belonging to the class $\mc{K}$, i.e.,
we compare against a feasible solution from $\mc{K}$ of minimum cost. Note that the
solution $N$ does not have to belong to the class $\mc{K}$. As explained for the
class of planar graphs above, \dsnK can be thought of as the special case that
lies between restricting the input to the class $\mc{K}$ and the general
unrestricted case.

The \dsnK problem has been implicitly studied in several results before for
various classes~$\mc{K}$ (cf.~\autoref{table:results}), in particular when
$\mc{K}$ contains either planar graphs, or graphs of bounded treewidth (here the
\emph{undirected} treewidth is meant, i.e., the treewidth of the underlying
undirected graph). For these results, typically an algorithm is given that
computes a solution for an input of a class~$\mc{K}$, but the algorithm is in
fact more general and can also be applied to the corresponding \dsnK problem.
Our algorithms presented in this paper for the class $\mc{K}$ of planar graphs
are also of this type. The reader may therefore want to think of the case when
the input is planar for our algorithms. On the other hand, our corresponding
hardness results are for the more general \dsnK problem, which means that they
rule out algorithms of this general type. In particular, they can be interpreted
as saying that if there are algorithms for input graphs from $\mc{K}$ that beat
our lower bounds for the more general \dsnK problem, then they cannot be of the
general type that seems prevalent in the study of the \dsn problem for special
input graphs.

Another special case we consider is the \bidsn problem, where the input graph
$G$ is \emph{bidirected}, i.e., for every edge $uv$ of $G$ the reverse edge $vu$
exists in $G$ as well and has the same weight as $uv$. This in turn can be
understood as the case lying between undirected and directed graphs, since
bidirected graphs are directed, but, similar to undirected graphs, a path can be
traversed in either direction at the same cost. Bidirected graphs model the
realistic setting~\cite{chen1989strongly, ramanathan2000topology,
wang2008approximate, lam2015dual} when the cost of transmitting from a node $u$
to a node $v$ in a wireless network is the same in both directions, which for
instance happens if the nodes all have the same transmitter model.

We systematically study several special cases of \dsn resulting from the above
restrictions, and prove several matching upper and lower bounds on runtimes
parameterized by $k$. We now give a brief overview of the studied problems, and
refer to \autoref{sec:results} for a detailed exposition of our results.

\begin{description}
\item[\altbidsnP,] i.e., the \dsnK problem on bidirected inputs, where $\mc{K}$
is the class of planar graphs: For this problem we present our main result,
which is that \bidsnP admits a \emph{parameterized approximation scheme (PAS)},
i.e., an algorithm that for any $\eps>0$ computes a $(1+\eps)$-approximation in
$f(\eps,k)\cdot n^{g(\eps)}$ time for some computable functions $f$ and $g$. We
also prove that, unless FPT=W[1], no \emph{efficient parameterized approximation
scheme (EPAS)} exists, i.e., there is no algorithm computing a
$(1+\eps)$-approximation in $f(\eps,k)\polyn$ time for any computable function
$f$. Thus the degree of the polynomial runtime dependence on $n$ has to depend
on $\eps$.

\item[\altbidsn,] i.e., the \dsn problem on bidirected inputs: The above PAS for
the rather restricted \bidsnP problem begs the question of whether a PAS also
exists for any more general problems, such as \bidsn. In particular, one may at
first think that~\bidsn closely resembles the undirected variant of \dsn, i.e.,
the well-known \pname{Steiner Forest~(SF)} problem, which is
FPT~\cite{DBLP:conf/icalp/FeldmannM16,DBLP:journals/networks/DreyfusW71} for
parameter~$k$. Surprisingly however, we can show that \bidsn is almost as hard
as \dsn (with almost-matching runtime lower bound under ETH), and moreover, no
PAS exists under Gap-ETH.
\end{description}
\vspace{-1mm}

Apart from the \pname{DST} problem, another well-studied special case of \dsn
with restricted demands is when the demand pairs form a cycle, i.e., we are
given a set $R=\{t_1,\ldots,t_{k}\}$ of $k$ terminals and the set of demands is
$\mc{D}=\{(t_i,t_{i+1})\}_{i=1}^{k}$ where $t_{k+1}=t_1$. Since this implies
that any optimum solution is strongly connected, this problem is accordingly
known as the \pname{Strongly Connected Steiner Subgraph (SCSS)} problem. In
contrast to \pname{DST}, it is implicit from~\cite{DBLP:journals/siamdm/GuoNS11}
(by a reduction from the \pname{Clique} problem) that optimum solutions to \scss
do not belong to any minor-closed graph class. Thus \scss is not easily captured
by some \dsnK problem for a restricted class~$\mc{K}$. Nevertheless it is still
possible to exploit the structure of the optimum solution to \scss, which
results in the following findings.

\begin{description}
\item[\altscss:] It is known that a $2$-approximation is
obtainable~\cite{DBLP:conf/iwpec/ChitnisHK13} when parameterizing by~$k$. We
prove that the factor of $2$ is best possible under Gap-ETH. To the best of our
knowledge, this is the first example of a problem with a tight parameterized
approximation result with non-trivial approximation factor (in this case
$2$), which also beats any approximation computable in polynomial time.

\item[\altbiscss,] i.e., the \scss problem on bidirected inputs: As for \bidsn,
one might think that \biscss is easily solvable via its undirected version,
i.e., the well-known \pname{Steiner Tree~(ST)} problem. In particular, the
\pname{ST} problem is FPT~\cite{DBLP:journals/networks/DreyfusW71} for
parameter~$k$. However, it is not the case that simply taking an optimum
undirected solution twice in a bidirected graph will produce a (near-)optimum
solution to \biscss (see \autoref{fig:savings}). Nevertheless we prove that
\biscss is FPT for parameter $k$ as well, while also being NP-hard. Our
algorithm is non-trivial and does not apply any methods used for undirected
graphs. To the best of our knowledge, bidirected inputs are the first
example where \scss remains NP-hard but turns out to be FPT parameterized by
$k$.
\end{description}

\begin{figure}
\centering
\includegraphics[scale=0.7]{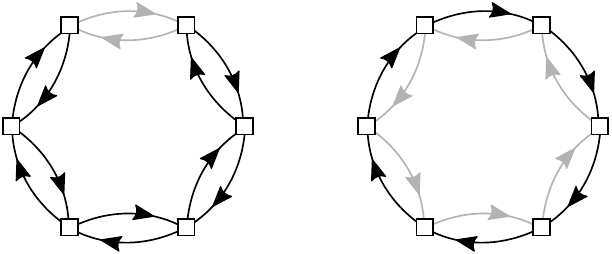}
\caption{A \biscss instance where all vertices are terminals.
Left: Black edges show a solution which takes an undirected optimum twice.
Right: The actual optimum solution is shown in black.}
\label{fig:savings}
\end{figure}

\subsection{Our results}\label{sec:results}

\paragraph*{Bidirected inputs with planar solutions.}
Our main theorem implies the existence of a PAS for \bidsnP, where the
parameter is the number $k$ of demands.

\begin{restatable}{thm}{thmscheme}
\label{thm:scheme}
There is a $2^{O(k^2)}n^{2^{O(1/\eps)}}$ time algorithm for \bidsnP, that for
any $\eps>0$ computes a \mbox{$(1+\eps)$-approx}\-imation.
\end{restatable}

This result begs the question of whether the considered special case is not too
restrictive. Should it not be possible to obtain better runtimes and/or should
it not be possible to even compute the optimum solution when parameterizing by
$k$ for this very restricted problem? And could it not be that a similar result
is true in more general settings, when for instance the input is bidirected but
the optimum is not restricted to a planar graph? We prove that both questions
can be answered in the negative.

First off, it is not hard to prove that a \emph{polynomial time approximation
scheme (PTAS)} is not possible for \bidsnP, i.e., it is necessary to
parameterize by $k$ in \autoref{thm:scheme}. This is implied by the following
result, since (as mentioned before) a PTAS for \bidsnP would also imply a PTAS
for \pname{bi-DST}, i.e., the \pname{DST} problem on bidirected input graphs.

\begin{restatable}{thm}{thmlbbiDST}
\label{thm:lb-biDST}
The \pname{bi-DST} problem (and by extension also the \bidsnP problem) is
\textup{APX-hard}.
\end{restatable}

One may wonder however, whether parameterizing by $k$ does not make the \bidsnP
problem FPT, so that approximating the planar optimum as in \autoref{thm:scheme}
would in fact be unnecessary. Furthermore, even if it is necessary to
approximate, one may ask whether the runtime given in \autoref{thm:scheme} can
be improved. In particular, note that the runtime we obtain in
\autoref{thm:scheme} is similar to that of a PTAS, i.e., the exponent of $n$ in
the running time depends on~$\eps$. Ideally we would like an EPAS, which has a
runtime of the form~$f(k,\eps)\polyn$, i.e., we would like to treat $\eps$ as a
parameter as well. The following theorem\footnote{We note that the W[1]-hardness
in \autoref{thm:lb-scheme-biDSN} for \bidsnP and also in \autoref{thm:lb-biDSN}
for \bidsn carries over to the parameterization by the solution size and also to
the solution cost when restricting to integer edge weights.} shows that both
approximation and runtime dependence on $\eps$ are in fact necessary in
\autoref{thm:scheme}.

\begin{restatable}{thm}{thmlbschemebiDSN}
\label{thm:lb-scheme-biDSN}
The \bidsnP problem is \textup{W[1]}-hard parameterized by $k$. Moreover,
under ETH, for any computable functions $f(k)$ and $f(k,\eps)$, the \bidsnP
problem
\begin{itemize}
 \item has no $f(k)\cdot n^{o(\sqrt k)}$ time algorithm to compute an optimum
solution, i.e., a solution with cost at most that of the cheapest planar one,
and
 \item has no $f(k,\eps)\cdot n^{o(\sqrt k)}$ time algorithm to compute a
solution with cost at most $(1+\eps)$ times that of the cheapest planar
one, if $\eps>0$ is part of the input.
\end{itemize}
\end{restatable}

It stands out that to compute optimum solutions, this theorem rules out
runtimes for which the dependence of the exponent of $n$ is substantially
better than $\sqrt{k}$, while for the general \dsn problem, as mentioned above,
the both necessary and sufficient dependence of the exponent is linear
in~$k$~\cite{DBLP:journals/siamcomp/FeldmanR06, DBLP:conf/soda/ChitnisHM14}.
Could it be that \bidsnP is just as hard as \dsn when computing optimum
solutions? The answer is no, as the next theorem shows.

\begin{restatable}{thm}{thmalgplanarbiDSN}
\label{thm:alg-planar-biDSN}
There is a $2^{O(k^{3/2}\log k)}\cdot n^{O(\sqrt k)}$ time algorithm to compute
the optimum solution for \bidsnP, i.e., a solution with cost at most that of
the cheapest planar one.
\end{restatable}

This result is an example of the so-called ``square-root phenomenon'': planarity
often allows runtimes that improve the exponent by a square root factor in terms
of the parameter when compared to the general
case~\cite{planar-1,planar-2,planar-3,planar-4,planar-5,planar-6,planar-7,
planar-8,new-planar}. Interestingly though, \citet{DBLP:conf/soda/ChitnisHM14}
show that under ETH, no $f(k)\cdot n^{o(k)}$ time algorithm can compute the
optimum solution to~\textsc{DSN}$_{\textsc{planar}}$. Thus assuming a bidirected
input graph in \autoref{thm:alg-planar-biDSN} is necessary (under ETH) to obtain
a factor of $O(\sqrt{k})$ in the exponent of $n$.

\paragraph*{Bidirected inputs.}
Since in contrast to \bidsnP, the \bidsn problem does not restrict the optimum
solutions, one may wonder whether a parameterized approximation scheme
as in \autoref{thm:scheme} is possible for this more general case as well. We
answer this in the negative by proving the following result, which implies
that restricting the optima to planar graphs was necessary for
\autoref{thm:scheme}.

\begin{restatable}{thm}{thmlbapproxbiDSN}
\label{thm:lb-approx-biDSN}
Under Gap-ETH, there exists a constant $\alpha>1$ such that for any computable
function~$f(k)$ there is no $f(k)\cdot n^{O(1)}$ time algorithm that computes
an $\alpha$-approximation for~\bidsn.
\end{restatable}

Also for the other obvious generalization of \bidsnP, in which the input graph
is unrestricted but we need to compute the planar optimum (i.e., the \dsnP
problem), no parameterized approximation scheme exists. This follows from a
recent result~\cite{Chitnis19-directed}, which shows that no
$(2-\eps)$-approximation can be computed for \dsnP in $f(k)\cdot n^{O(1)}$ time
for any $\eps>0$ and computable function $f$, under Gap-ETH.

What approximation factors can be obtained for \bidsn when parameterizing by
$k$, given the lower bound of \autoref{thm:lb-approx-biDSN} on one hand, and the
before-mentioned result~\cite{DM18} that rules out a
$k^{1/4-o(1)}$-approximation for \dsn in time parameterized by $k$ on the other?
It turns out that it is not too hard to obtain a constant approximation for
\bidsn, given the similarity of bidirected graphs to undirected graphs. In
particular, relying on the fact that for the undirected version of \dsn, i.e.\
the \pname{SF} problem, there is a polynomial time $2$-approximation algorithm
by \citet{steiner-forest}, and an FPT algorithm based on
\citet{DBLP:journals/networks/DreyfusW71}, we obtain the following theorem,
which is also in contrast to \autoref{thm:lb-biDST}.

\begin{restatable}{thm}{thmalgapproxbiDSN}
\label{thm:alg-approx-biDSN}
The \bidsn problem admits a $4$-approximation in polynomial time, and a
$2$\hy{}approximation in $2^{O(k)}\polyn$ time.
\end{restatable}

Even if \autoref{thm:lb-approx-biDSN} in particular shows that \bidsn cannot be
FPT under Gap-ETH, it does not give a strong lower bound on the runtime
dependence in the exponent of $n$. However using the weaker ETH assumption we
can obtain such a lower bound, as the next theorem shows. Interestingly, the
obtained lower bound implies that when aiming for optimum solutions, the
restriction to bidirected inputs does not make \dsn easier than the general
case, as also for \bidsn the $n^{O(k)}$ time algorithm by
\citet{DBLP:journals/siamcomp/FeldmanR06} is essentially best possible. This is
in contrast to the \bidsnP problem where the square-root phenomenon takes effect
as shown by \autoref{thm:alg-planar-biDSN}.

\begin{restatable}{thm}{thmlbbiDSN}
\label{thm:lb-biDSN}
The \bidsn problem is \textup{W[1]}-hard parameterized by $k$. Moreover, under
ETH there is no $f(k)\cdot n^{o(k/\log k)}$ time algorithm for \bidsn, for any
computable function~$f(k)$.
\end{restatable}

\paragraph*{Strongly connected solutions.}
Just like the more general \dsn problem, the \scss problem is
W[1]-hard~\cite{DBLP:journals/siamdm/GuoNS11} parameterized by $k$, and
is also hard to approximate as no polynomial time
$O(\log^{2-\eps}n)$\hy{}approximation is possible~\cite{approx-hardness},
unless NP~$\subseteq$~ZTIME$(n^{\text{polylog}(n)})$. However it is possible to
exploit the structure of the optimum to \scss to obtain a $2$-approximation
algorithm parameterized by $k$, as observed by
\citet{DBLP:conf/iwpec/ChitnisHK13}. This is because any strongly connected
graph is the union of two arborescences, and these form solutions to
\pname{DST}. The $2$-approximation follows, since \pname{DST} is FPT by the
classic result of \citet{DBLP:journals/networks/DreyfusW71}. Thus in contrast to
\dsn, for \scss it is possible to beat any approximation factor obtainable in
polynomial time when parameterizing by~$k$.

\begin{thm}[\cite{DBLP:conf/iwpec/ChitnisHK13}] The \scss problem admits a
$2$-approximation in $3^k\polyn$ time.
\end{thm}

An obvious question now is whether the approximation ratio of this rather simple
algorithm can be improved. Interestingly we are able to show that this is not
the case. To the best of our knowledge, this is the first example of a problem
with a tight parameterized approximation result with non-trivial approximation
factor (in this case $2$), which also beats any approximation computable in
polynomial time.

\begin{restatable}{thm}{thmlbapproxSCSS}
\label{thm:lb-approx-SCSS}
Under Gap-ETH, for any $\eps > 0$ and any computable function $f(k)$, there is
no $f(k)\cdot n^{O(1)}$ time algorithm that computes a $(2-\eps)$-approximation
for \scss.
\end{restatable}

We remark that our reduction for \autoref{thm:lb-approx-SCSS} uses edge weights,
which however can be polynomially bounded. As a consequence an instance can be
further reduced in polynomial time by first scaling the edge weights to
polynomially bounded integers, and then subdividing each edge $w$ times if its
weight is $w$. This results in an equivalent unweighted instance, and thus the
lower bound of \autoref{thm:lb-approx-SCSS} is also valid for unweighted
instances of \scss.

\paragraph*{Bidirected inputs with strongly connected solutions.}
In light of the above results for restricted cases of \dsn, what can be said
about restricted cases of \scss? It is implicit in the work of
\citet{DBLP:conf/soda/ChitnisHM14} that \scssP, i.e., the problem of computing a
solution of cost at most that of the cheapest strongly connected planar
solution, can be solved in $2^{O(k\log k)}\cdot n^{O(\sqrt{k})}$ time, while
under ETH no $f(k)\cdot n^{o(\sqrt{k})}$ time algorithm is possible. Hence
\scssP is slightly easier than \textsc{DSN}$_{\textsc{planar}}$ where the
exponent of $n$ needs to be linear in $k$, as mentioned before. On the other
hand, the \biscss problem turns out to be a lot easier to solve than \bidsn.
This is implied by the next theorem, which stands in contrast to
\autoref{thm:lb-approx-biDSN} and \autoref{thm:lb-biDSN}.

\begin{restatable}{thm}{thmalgbiSCSS}
\label{thm:alg-biSCSS}
There is a $4^{k^2+O(k)}\polyn$ time algorithm for \biscss, i.e., it is FPT
for parameter~$k$.
\end{restatable}

Could it be that \biscss is even solvable in polynomial time? We prove that this
is not the case, unless P~=~NP. To the best of our knowledge, the class of
bidirected graphs is the first example where \scss remains NP-hard but
turns out to be FPT parameterized by~$k$. Moreover, note that the above
algorithm has an exponential runtime in~$k^2$. We conjecture that a single
exponential runtime should suffice, and we also obtain a lower bound result of
this form.

\begin{restatable}{thm}{thmlbbiSCSS}
\label{thm:lb-biSCSS}
The \biscss problem is NP-hard. Moreover, under ETH there is no
$2^{o(k)}\polyn$ time algorithm for \biscss.
\end{restatable}

\paragraph*{Remark.}
For ease of notation, throughout this paper we chose to use the number of
demands~$k$ uniformly as the parameter. Alternatively one might also consider
the smaller parameter $|R|$, where $R=\bigcup_{i=1}^{k} \{s_i, t_i\}$ is the set
of terminals (as also done in~\cite{DBLP:conf/stacs/EibenKPS19}). Note for
instance that in case of the \scss problem, $k=|R|$, while for \dsn, $k$ can be
as large as $\Theta(|R|^2)$. However we always have $k\geq |R|/2$, since the
demands can form a matching in the worst case. It is interesting to note that
all our algorithms for \dsn have the same running time for parameter~$|R|$ as
for parameter~$k$. That is, we may set $k=|R|$ in \autoref{thm:scheme},
\autoref{thm:alg-planar-biDSN}, and~\autoref{thm:alg-approx-biDSN}.

\begin{table}%
\centering
\begin{tabular}{ l || c | c | c | c | c | c}
      & \multicolumn{3}{c|}{algorithms} & \multicolumn{3}{c}{lower bounds}\\
      problem & approx. & runtime & ref. & approx. & runtime & ref. \\
      \hline\hline &&&&&\\[-1em]

      \dsn
        & --
        & $n^{O(k)}$
        & \cite{DBLP:journals/siamcomp/FeldmanR06}
        & --
	& $f(k)\cdot n^{o(k)}$
	& \cite{DBLP:conf/stacs/EibenKPS19,DBLP:journals/siamdm/GuoNS11} \\
      \hline &&&&&\\[-1em]

      \dsn
        & $O(k^{\frac{1}{2}+\eps})$ %
        & $n^{O(1)}$
        & \cite{DBLP:journals/talg/ChekuriEGS11}
	& $k^{\frac{1}{4}-o(1)}$
	& $f(k)\polyn$
	& \cite{DM18}\\
      \hline &&&&&\\[-1em]

      \textsc{DSN}$_{\textsc{TW: }\omega}$
        & --
	& $2^{O(k\omega\log\omega)}\cdot n^{O(\omega)}$
	& \cite{DBLP:conf/icalp/FeldmannM16}
	& --
	& $f(k,\omega)\cdot n^{o(\omega)}$
	& \cite{DBLP:conf/icalp/FeldmannM16}\\
      \hline &&&&&\\[-1em]

      \bidsnP
        & --
        & $2^{O(k^{3/2}\log k)}\cdot n^{O(\sqrt{k})}$
        & Thm~\ref{thm:alg-planar-biDSN}
        & --
	& $f(k)\cdot n^{o(\sqrt{k})}$
	& Thm~\ref{thm:lb-scheme-biDSN}\\
      \hline &&&&&\\[-1em]

      \bidsnP
        & $1+\eps$
        & $2^{O(k^2)}n^{2^{O(1/\eps)}}$
	& Thm~\ref{thm:scheme}
        & $1+\eps$
        & $f(\varepsilon,k)\cdot n^{o(\sqrt{k})}$
        & Thm~\ref{thm:lb-scheme-biDSN}\\
      \hline &&&&&\\[-1em]

      \dsnP
        & --
	& $n^{O(k)}$
	& \cite{DBLP:conf/stacs/EibenKPS19,DBLP:journals/siamcomp/FeldmanR06}
	& --
	& $f(k)\cdot n^{o(k)}$
	& \cite{DBLP:conf/soda/ChitnisHM14}\\
      \hline &&&&&\\[-1em]

      \dsnP
        & $\alpha\geq 2$
	& (open)
	&
	& $(2-\eps)$
	& $f(k)\polyn$
	& \cite{Chitnis19-directed}\\
      \hline &&&&&\\[-1em]

      \bidsn
        & --
        & $n^{O(k)}$
        & \cite{DBLP:journals/siamcomp/FeldmanR06}
	& --
	& $f(k)\cdot n^{o(k/\log k)}$
	& Thm~\ref{thm:lb-biDSN}\\
      \hline &&&&&\\[-1em]

      \bidsn
        & 2
        & $2^{O(k)}\polyn$
        & Thm~\ref{thm:alg-approx-biDSN}
	& $\alpha\in \Theta(1)$
	& $f(k)\polyn$
	& Thm~\ref{thm:lb-approx-biDSN}\\
      \hline &&&&&\\[-1em]

      \bidsn
        & $4$
        & $n^{O(1)}$
        & Thm~\ref{thm:alg-approx-biDSN}
        & $\alpha\in\Theta(1)$
        & $n^{O(1)}$
        & Thm~\ref{thm:lb-biDST}\\
      \hline &&&&&\\[-1em]

      \scss
        & --
        & $n^{O(k)}$
	& \cite{DBLP:journals/siamcomp/FeldmanR06}
	& --
	& $f(k)\cdot n^{o(k/\log k)}$
	& \cite{DBLP:journals/siamdm/GuoNS11,DBLP:conf/soda/ChitnisHM14} \\
      \hline &&&&&\\[-1em]

      \scss
        & $2$
        & $3^k\polyn$
        & \cite{DBLP:conf/iwpec/ChitnisHK13}
        & $2-\eps$
        & $f(k)\polyn$
        & Thm~\ref{thm:lb-approx-SCSS}\\
      \hline &&&&&\\[-1em]

      \scssP
        & --
	& $2^{O(k\log k)}\cdot n^{O(\sqrt{k})}$
	& \cite{DBLP:conf/soda/ChitnisHM14}
	& --
	& $f(k)\cdot n^{o(\sqrt{k})}$
	& \cite{DBLP:conf/soda/ChitnisHM14}\\
      \hline &&&&&\\[-1em]

      \biscss
        & --
        & $4^{k^2+O(k)}\polyn$
        & Thm~\ref{thm:alg-biSCSS}
        & --
        & $2^{o(k)}\polyn$
        & Thm~\ref{thm:lb-biSCSS}
\end{tabular}
\caption{Summary of achievable runtimes for \dsn and \scss when parameterizing
by $k$. A dash (--) refers to computing optimum solutions. Some of the previous
results are implicit and in the papers are rather stated for the case when the
input graphs are restricted to the same class as the optimum solutions.
}
\label{table:results}
\end{table}

\subsection{Our techniques}\label{sec:techniques}

It is already apparent from the above exposition of our results, that
understanding the structure of the optimum solution is a powerful tool when
studying \dsn and its related problems (cf.~\autoref{table:results}). This is
also apparent when reading the literature on these problems, and we draw some
of our inspiration from these known results, as described below.

\paragraph*{Approximation scheme for {\altbidsnP}.}
We generalize the insights on the structure of optimum solutions to the
classical \pname{Steiner Tree~(ST)} problem for our main result in
\autoref{thm:scheme}. For the \pname{ST}  problem, an \emph{undirected}
edge-weighted graph is given together with a terminal set $R$, and the task is
to compute the cheapest tree connecting all $k$ terminals. For this
problem only polynomial-time $2$-approximations were known~\cite{GP68,V01},
until it was taken into account~\cite{KZ97, PS00, zelikovsky_1993_11_over_6_apx,
RZ05} that any optimum Steiner tree can be decomposed into so-called \emph{full
components}, i.e., subtrees for which exactly the leaves are terminals. If a
full component contains only a small subset of size $k'$ of the terminals, it is
the solution to an \pname{ST} instance, for which the optimum can be computed
efficiently in time $3^{k'}\polyn$ using the algorithm of
\citet{DBLP:journals/networks/DreyfusW71}. A fundamental observation proved by
\citet{borchers-du} is that for any~$k'$ there exists a solution to \pname{ST}
of cost at most $1+\frac{1}{\lfloor \log_2 k'\rfloor}$ times the optimum, in
which every full component contains at most $k'$ terminals. Thus setting
$k'=2^{1/\eps}$ for some constant~$\eps>0$, all full-components with at most
\smash{$2^{1/\eps}$} terminals can be computed in polynomial time, and among
them exists a collection forming a $(1+\eps)$-approximation. The key to obtain
approximation ratios smaller than $2$ for \pname{ST} is to cleverly select a
good subset of all computed full-components. This is for instance done
in~\cite{DBLP:journals/jacm/ByrkaGRS13} via an iterative rounding procedure,
resulting in an approximation ratio of $\ln(4)+\eps< 1.39$, which currently is
the best one known.

Our main technical contribution is to generalize the \citeauthor{borchers-du}
Theorem to \bidsnP. In particular, to obtain our approximation scheme of
\autoref{thm:scheme}, we employ a similar approach by decomposing a \bidsnP
solution into sub-instances, each containing a small number of terminals. As
\bidsnP is \mbox{W[1]-hard} by \autoref{thm:lb-scheme-biDSN}, we cannot hope to
compute optimum solutions to each sub-instance as efficiently as for \pname{ST}.
However, we provide an XP-algorithm with runtime $2^{O(k^{3/2}\log k)}\cdot
n^{O(\sqrt{k})}$ for \bidsnP in \autoref{thm:alg-planar-biDSN}. Thus if every
sub-instance contains at most $2^{1/\eps}$ terminals, each can be solved in
\smash{$n^{2^{O(1/\eps)}}$} time, and this accounts for the ``non-efficient''
runtime of our approximation scheme. Since we allow runtimes parameterized
by~$k$, we can then search for a good subset of precomputed small optimum
solutions to obtain a solution to the given demand set $\mc{D}$. For the latter
solution to be a $(1+\eps)$-approximation however, we need to generalize the
\citeauthor{borchers-du} Theorem for \pname{ST} to \bidsnP (see
\autoref{lem:full-comp} for the formal statement). This constitutes the bulk of
the work to prove \autoref{thm:scheme}.

\paragraph*{Exact algorithms for {\altbidsnP} and {\altbiscss}.}
Also from a parameterized point of view, understanding the structure of the
optimum solution to \dsn has lead to useful insights in the past. We will
leverage one such recent result by \citet{DBLP:conf/icalp/FeldmannM16}, where
the above mentioned standard special case of restricting the patterns of the
demands in $\mc{D}$ is studied in depth. The result is a complete dichotomy over
which classes of restricted patterns define special cases of \dsn that are FPT
and which are W[1]-hard for parameter~$k$. The high-level idea is that whenever
the demand patterns imply optimum solutions of constant treewidth, there is an
FPT algorithm computing such an optimum. In contrast, the problem is W[1]-hard
whenever the demand patterns imply the existence of optimum solutions of
arbitrarily large treewidth. The FPT algorithm
from~\cite{DBLP:conf/icalp/FeldmannM16} lies at the heart of all our positive
results, and therefore shows that the techniques developed
in~\cite{DBLP:conf/icalp/FeldmannM16} to optimally solve special cases of \dsn
can be extended to find (near-)optimum solutions for other W[1]-hard special
cases as well. It is important to note that the algorithm of
\cite{DBLP:conf/icalp/FeldmannM16} can also be used to compute the cheapest
solution of treewidth at most $\omega$, even if there is an even better solution
of treewidth larger than~$\omega$ (which might be hard to compute). Formally,
the result leveraged in this paper is the following.

\begin{thm}[implicit in Theorem~5 of \cite{DBLP:conf/icalp/FeldmannM16}]
\label{thm:alg-FMarx}
If $\mc{K}$ is the class of graphs with treewidth at most~$\omega$, then the
\dsnK problem can be solved in $2^{O(k\omega\log\omega)}\cdot n^{O(\omega)}$
time.
\end{thm}

We exploit the algorithm given by \autoref{thm:alg-FMarx} to prove our
algorithmic
results of \autoref{thm:alg-planar-biDSN} and \autoref{thm:alg-biSCSS}. In
particular, we prove that any \bidsnP solution has treewidth $O(\sqrt k)$, from
which \autoref{thm:alg-planar-biDSN} follows immediately. For \biscss however,
we give an example of an optimum solution of treewidth~$\Omega(k)$. Hence we
cannot exploit the algorithm of \autoref{thm:alg-FMarx} directly to obtain
\autoref{thm:alg-biSCSS}. In fact on general input graphs, a treewidth of
$\Omega(k)$ would imply that the problem is W[1]-hard by the hardness results
in~\cite{DBLP:conf/icalp/FeldmannM16} (which was indeed originally shown
by~\citet{DBLP:journals/siamdm/GuoNS11}). As this stands in stark contrast to
\autoref{thm:alg-biSCSS}, it is particularly interesting that the \scss problem
on bidirected input graphs is FPT. We prove this result by decomposing an
optimum solution to \biscss into sub-instances of \biscssK, where~$\mc{K}$ is a
class of directed graphs of treewidth~$1$ (so-called \emph{poly-trees}). For
each such sub-instance we can compute a solution in $2^{O(k)}\polyn$ time by
using \autoref{thm:alg-FMarx} (for $\omega=1$), and then combine them into an
optimum solution to \biscss.

\paragraph*{W[1]-hardness and runtime lower bounds.}
Our hardness proofs for \bidsn are based on reductions from the \gt
problem~\cite{pc-book}. This problem is particularly well-suited to prove
hardness for problems on planar graphs, due to its grid-like structure. We first
develop a specific gadget that can be exploited to show hardness for bidirected
graphs. This gadget however is not planar. We only exploit the structure of \gt
to show that the optimum solution is planar for \autoref{thm:lb-scheme-biDSN}.
For \autoref{thm:lb-biDSN} we modify this reduction to obtain a stronger runtime
lower bound, but in the process we lose the property that the optimum is planar.

\paragraph*{Parameterized inapproximability.}
Our hardness result for \scss is proved by combining a variant of a known
reduction by \citet{DBLP:journals/siamdm/GuoNS11} with a recent parameterized
hardness of approximation result for \pname{Densest
$k$-Subgraph}~\cite{param-inapprox}. Our inapproximability result for \bidsn is
shown by combining our W[1]-hardness reduction with the same hardness of
approximation result of \pname{Densest $k$-Subgraph}.

\subsection{Approximate kernelization}

A topic closely related to parameterized algorithms is kernelization, which
concerns efficient pre-processing algorithms. As formalized by
\citet{lokshtanov2017lossy}, an \emph{$\alpha$-approximate kernel} for an
optimization problem consists of a \emph{reduction} and a \emph{lifting}
algorithm, both running in polynomial time. The reduction algorithm takes an
instance $I$ with parameter $k$ and computes a new instance $I'$ and parameter
$k'$, such that the size $|I'|+k'$ of the new instance is bounded by some
function $f(k)$ of the input parameter. This new instance $I'$ is also called a
\emph{kernel} of~$I$. The lifting algorithm takes as input any
$\beta$-approximation for the kernel $I'$ and computes an
$\alpha\beta$-approximate solution for $I$.

It has long been known that a problem is FPT if and only if it admits an
\emph{exact} kernel, i.e., a $1$-approximate kernel. \citet{lokshtanov2017lossy}
prove that this is also the case in general: a problem has a parameterized
$\alpha$-approximation algorithm if and only if it admits an
$\alpha$-approximate kernel. Note that the size of the kernel might in general
be very large, even if it is bounded in the input parameter. Therefore a
well-studied interesting question is whether a polynomial-sized kernel exists,
which can be taken as evidence that a problem admits a very efficient
pre-processing algorithm. If a problem admits a polynomial-sized
$(1+\eps)$\hy{}approximate kernel for every $\eps>0$, then we say that it
admits a \emph{polynomial-sized approximate kernelization scheme~(PSAKS)}.

The \pname{ST} problem is known to be
FPT~\cite{DBLP:journals/networks/DreyfusW71}, while not admitting any
polynomial-sized exact kernel~\cite{dom} for parameter~$k$, unless
NP~$\subseteq$~coNP/Poly. However, as noted by \citet{lokshtanov2017lossy}, the
\citeauthor{borchers-du} Theorem implies the existence of a PSAKS for the
\pname{ST} problem. As a consequence of our generalization of the
\citeauthor{borchers-du} Theorem we also obtain a PSAKS for \bidsnP (see
\autoref{crl:lossy} for a formal statement). This is despite the fact that this
problem does not admit \emph{any} exact kernel for parameter $k$ according to
\autoref{thm:lb-scheme-biDSN}. Furthermore, we observe that the same kernel is
in fact a polynomial-sized $(2+\eps)$-approximate kernel for the \bidsn problem.
This nicely complements the existence of a parameterized $2$-approximation
algorithm according to \autoref{thm:alg-approx-biDSN}.

\subsection{Related work}

The \pname{ST} problem is one of the 21 NP-hard problems listed in the seminal
paper of \citet{MR51:14644}. \citet{DBLP:journals/networks/DreyfusW71} showed
that the problem is solvable in time $3^k\polyn$, which was later
improved~\cite{fuchs2007dynamic} to $(2+\eps)^k\polyn$ for any constant
$\eps>0$. For values $k>2\log n\log^3\log n$ an even faster algorithm
exists~\cite{vygen2011faster}. For unweighted graphs an algorithm with runtime
$2^k\polyn$ can be
obtained~\cite{BjorklundHKK07,DBLP:journals/algorithmica/Nederlof13}. An early
LP-based $2$-approximation algorithm for \pname{ST} uses the so-called
\emph{bidirected cut relaxation (BCR)}~\cite{wong1984dual,E67,
feldmann2016equivalence}, which formulates the problem by bidirecting the
undirected input graph. Thus bidirected instances have implicitly been used even
for the classical \pname{ST} problem since the 1960s. For \pname{ST} and
\pname{SF} there are PTASes on planar and bounded genus
graphs~\cite{bateni2011approximation,eisenstat2012efficient}.

A recent result~\cite{DBLP:conf/stacs/EibenKPS19} investigates the complexity of
\dsn with respect to the stronger parameter~$|R|$ (instead of the number of
demands $k$; see remark above). It is shown that for bounded genus graphs the
\dsn problem can be solved in $f(|R|)\cdot n^{O(|R|)}$ time, while in general
no $f(|R|)\cdot n^{o(|R|^2/\log |R|)}$ time algorithm exists, under ETH.
The \pname{DST} problem has an $O(k^\eps)$-approximation in polynomial
time~\cite{DBLP:journals/jal/CharikarCCDGGL99}, and an $O(\log^2 k
/\log\log k)$-approximation in quasi-polynomial
time~\cite{DBLP:conf/stoc/0001LL19}. Moreover, no better approximation is
possible in quasi-polynomial time~\cite{DBLP:conf/stoc/0001LL19}. A long
standing open problem is whether a polynomial-time algorithm with
poly-logarithmic approximation guarantee exists for \pname{DST}. The \scss
problem has also been studied in the special case when $R=V$. This case is
commonly known as \pname{Minimum Strongly Connected Spanning Subgraph}, and the
best approximation factor known is~$2$, which is also given by computing two
spanning arborescences~\cite{frederickson1981approximation}, and for $R=V$ can
be done in polynomial time. For the unweighted case however, a
$3/2$-approximation is obtainable~\cite{vetta2001approximating}, which is
contrast to unweighted \scss, where the lower bound of
\autoref{thm:lb-approx-SCSS} is also valid.

Bidirected input graphs have been studied in the context of radio and ad-hoc
wireless networks~\cite{chen1989strongly,ramanathan2000topology,
wang2008approximate,lam2015dual}. In the \pname{Power Assignment} problem,
nodes of a given bidirected network need to be activated in order to induce a
network satisfying some connectivity condition. For instance
in~\cite{chen1989strongly}, the problem of finding a strongly connected network
is considered, but also other settings such as
$2$-(vertex/edge)-connectivity~\cite{wang2008approximate} or
$k$-(vertex/edge)-connectivity~\cite{lam2015dual} have been studied.

\subsection{Organization of the paper}

We give some preliminaries and basic observations on the structure of optimum
solutions to \bidsn in bidirected input graphs in \autoref{sec:bidir}. These are
used throughout \autoref{sec:scheme}, where we present our approximation scheme
for \bidsnP of \autoref{thm:scheme}, and \autoref{sec:alg-opt}, where we show
how to compute optimum solutions to \bidsnP for \autoref{thm:alg-planar-biDSN}
and \biscss for \autoref{thm:alg-biSCSS}. Before presenting our main result of
\autoref{sec:scheme} however, we first need to develop the approximation
algorithms for \bidsn of \autoref{thm:alg-approx-biDSN}, which we do in
\autoref{sec:dir-vs-undir} together with the hardness result for \bidsnP of
\autoref{thm:lb-biDST}. The inapproximability results for \bidsn of
\autoref{thm:lb-approx-biDSN} and \scss of \autoref{thm:lb-approx-SCSS} are
given in \autoref{sec:inapprox}, and the remaining runtime lower bounds for
\bidsnP of \autoref{thm:lb-scheme-biDSN}, \bidsn of \autoref{thm:lb-biDSN}, and
\biscss of \autoref{thm:lb-biSCSS} can be found in \autoref{sec:lb}. In
\autoref{app:dsn} we present the reduction that was used later by \citet{DM18}
to prove the $k^{1/4 - o(1)}$-approximation hardness for \dsn. Finally, in
\autoref{sec:questions} we list some open questions.

\section{Structural properties of optimum solutions to \altbidsn}
\label{sec:bidir}

In this section we give some definitions relevant to directed and bidirected
graphs, and some fundamental observations on solutions to \pname{bi-DSN} that
we will use throughout the paper.

Due to the similarity of bidirected graphs to undirected graphs, we will often
exploit the structure of the underlying undirected graph of a given bidirected
graph. More generally, for any directed graph $G$ we denote the underlying 
undirected graph by~$\ud{G}$. A \emph{poly-graph} is obtained by
directing the edges of an undirected graph, and analogously we obtain
\emph{poly-cycles}, \emph{poly-paths}, and \emph{poly-trees}. A strongly
connected poly-cycle is a \emph{directed cycle}, and a poly-tree for which all
vertices can reach (or are reachable from) a designated \emph{root} vertex $r$
is called an \emph{out-arborescence} (or \emph{in-arborescence}). Note that for 
any edge $uv$ of a poly-graph, the reverse edge $vu$ does not exist, and
so a poly-graph is in a sense the opposite of a bidirected graph. In between
poly-graphs and bidirected graphs are general directed graphs.

\subsection{Cycles of optimum solutions in bidirected graphs}

We need the following observation, which has far reaching consequences for
\pname{bi-DSN} algorithms. An optimum \bidsn solution may contain poly-cycles, 
which are not directed cycles: consider for instance a bidirected graph for 
which the underlying undirected graph is a cycle on four vertices and every 
edge has unit weight. If the vertices are numbered $1,2,3,4$ along the cycle and 
the demands are $\{(1,2),(1,4),(3,2),(3,4)\}$, then it is not hard to see that 
an optimum solution is given by the poly-cycle with edges corresponding to the 
demands, i.e., the edge set $\{(1,2),(1,4),(3,2),(3,4)\}$. As the following 
lemma shows however, any such poly-cycle can be replaced by a directed cycle.

\begin{lem}\label{lem:cycles}
Let $O$ be a poly-cycle of a subgraph $N\subseteq G$ in a bidirected graph
$G$. Replacing~$O$ with a directed cycle on $V(O)$ in $N$ results in a subgraph
$M$ of $G$ with cost at most that of $N$, such that a $u\to v$ path exists in
$M$ for every vertex pair $u,v$ for which $N$ contained a $u\to v$ path.
\end{lem}
\begin{proof}
Removing all edges of $O$ in $N$ and replacing them with a directed cycle cannot
increase the cost, as $G$ is bidirected (the cost may decrease if an edge $uv$
of $O$ is replaced by an edge $vu$, which is already contained in $N$). Any
$u\to v$ path that leads through $O$ in $N$ can be rerouted through the strongly
connected directed cycle in $M$.
\end{proof}

From this we can deduce the following useful observation, which we will exploit
for all of our algorithms. The intuitive meaning of it is that any poly-cycle
of an optimum \pname{bi-DSN} solution splits the solution into parts of which
each contains at least one terminal.

\begin{lem}\label{lem:cycle-split}
Let $N\subseteq G$ be an optimum \pname{bi-DSN} solution in a bidirected
graph $G$, such that $N$ contains a poly-cycle $O\subseteq N$. Every edge of
$N$ that is incident to two vertices of $O$ is also part of~$O$. Moreover,
every connected component of the graph resulting from removing $V(O)$ from $N$
contains at least one terminal.
\end{lem}
\begin{proof}
By \autoref{lem:cycles} we may exchange $O$ with a directed cycle $O'$ without
increasing the cost and maintaining all connections for the demands given by
the \pname{bi-DSN} instance. Since $N$ has minimum cost, this means that the
resulting network $N'$ is also an optimum solution. Assume that $N$ contained
some edge $e$ incident to two vertices of $O$ but $e\notin E(O)$. The edge $e$
cannot be a reverse edge of some edge $f$ of~$O$, as we could replace $O$ with a
cycle directed in the same direction as $e$. This would decrease the cost as
$N'$ only contains $e$, while $N$ contains both $e$ and~$f$. We are left with
the case that $e$ is a \emph{chord} of $O$, i.e., it connects two non-adjacent
vertices of $O$. However in this case, the endpoints of $e$ are strongly
connected through~$O'$ in $N'$ even after removing $e$. Thus we would be able to
safely remove $e$ and decrease the cost of $N'$.

Now assume that some connected component $C$ of the graph obtained from $N$ by
removing $V(O)$ contains no terminal. Note that $C$ also exists in the graph
obtained from $N'$ by removing~$O'$ and its vertices. As $C$ contains no 
terminals, any $s\to t$ path in $N'$ for a demand $(s,t)$ that contains a vertex 
of $C$ must contain a $u\to v$ subpath for some vertices $u,v\in V(O')$ with 
internal vertices from $C$. However the vertices $u,v$ are strongly connected 
through $O'$ and hence the $u\to v$ subpath can be rerouted via~$O'$. This means 
we may safely remove $C$ from $N'$ without loosing any connections for the 
required demands. However this contradicts the optimality of~$N'$, and in turn 
also our assumption that $N$ is an optimum solution.
\end{proof}

\subsection{Reducing the vertex degrees}\label{app:degrees}

For our proofs, it will be convenient to assume that the degrees of the vertices 
in some given graph are bounded. More specifically, consider a (possibly planar) 
graph $N$ connecting a terminal set $R$ according to some set of demands. We use 
the standard procedure below, which assures that every terminal has only one 
neighbour in~$N$, and every \emph{Steiner vertex} of $N$, i.e.\ every 
non-terminal in $V(N)\setminus R$, has exactly three neighbours in~$N$.

We execute the following steps on $N$ in the given order. It is easy to see that 
these operations preserve planarity (if $N$ is planar), the cost of $N$, and the 
connectivity according to the demands.
\begin{enumerate}
 \item For every terminal $t\in R$ that has more than one neighbour in $N$, we
introduce a new Steiner vertex $v$ and add the edges $vt$ and $tv$ with cost $0$
each. Thereafter every neighbour $w$ of $t$ different from $v$ is made a
neighbour of $v$ instead. That is, the edges $wt$ and $tw$ are replaced by the
edges $wv$ and $vw$ of the same cost. After this, every terminal in $N$ has one
neighbour only.
 \item Then for every Steiner vertex $v$ with more than $3$ neighbours in $N$,
we split $v$ into two vertices as follows. In case $N$ is planar we first fix a 
drawing of $N$. Then we introduce a new Steiner vertex~$u$ and edges $uv$ and 
$vu$ with cost~$0$ each. As the new vertex only has one neighbour, we may draw 
$u$ in an arbitrary face of $N$ that is incident to $v$. Let $F$ be this face 
containing $u$, and let $w_1$ and $w_2$ be the two neighbours of~$v$ incident to 
$F$ that are different from $u$. For each $j\in\{1,2\}$, we replace the edges 
$vw_j$ and $w_jv$ with edges $uw_j$ and $w_ju$, respectively. We maintain the 
edge costs in each of these replacement steps. Note that $N$ remains planar 
under this operation. In case $N$ is not planar, we proceed in the same way but 
simply pick arbitrary neighbours $w_1,w_2$ of $v$ that are different from $u$. 
After repeating this for every Steiner vertex with more than $3$ neighbours, 
all Steiner vertices of $N$ have at most three neighbours.
 \item Next we consider each Steiner vertex $v$ that has exactly two neighbours
$u$ and $w$. If $N$ contains the path $uvw$, we add an edge $uw$ with the same 
cost as the path. Similarly, if $N$ has a $wvu$ path, we introduce the edge 
$wu$ with the same cost. We then remove the vertex~$v$. After this all Steiner 
vertices of $N$ have exactly three neighbours.
\end{enumerate}

\section{Hardness and algorithms for {\altbidsn} via undirected graphs}
\label{sec:dir-vs-undir}

In this section we present two results for problems on bidirected graphs that
follow from corresponding results on undirected graphs. We first prove
\autoref{thm:lb-biDST}, which we restate below. In particular, it implies that
\bidsnP has no PTAS, unless P=NP.

\thmlbbiDST*
\begin{proof}
Given a \pname{Steiner Tree (ST)} instance on an undirected graph $\ud{G}$, we
simply bidirect each edge to obtain the bidirected graph $G$. We then choose any
of the terminals in $G$ as the root to get an instance of \pname{bi-DST}. It is
easy to see that any solution to \pname{ST} in $\ud{G}$ corresponds to a
solution to \pname{bi-DST} in $G$ of the same cost, and vice versa. As the
\pname{ST} problem is APX-hard~\cite{chlebik2002approximation}, the hardness
carries over to \pname{bi-DST}.
\end{proof}

Note that as the definition of the \pname{ST} problem does not restrict the
feasible solutions to trees, this hardness result does not restrict the
approximate solutions to \pname{bi-DST} to arborescences either. That is, it is
also hard to compute an approximation $N$ to the optimum \bidsnP solution,
even if we allow~$N$ to be a non-planar graph.

Next we turn to the positive result of \autoref{thm:alg-approx-biDSN}, which we
also restate below. Note that this theorem is in contrast to
\autoref{thm:lb-biDST} and \autoref{thm:lb-approx-biDSN}.

\thmalgapproxbiDSN*
\begin{proof}
Given a bidirected graph $G$ and a demand set $\mc{D}=\{(s_i, t_i)\mid 1\leq
i\leq k\}$ of an instance to \bidsn, we reduce it to an instance of the
\pname{Steiner Forest (SF)} problem in the underlying undirected graph $\ud{G}$
with the corresponding unordered demand set $\ud{\mc{D}}=\{\{s_i,t_i\}\mid
(s_i,t_i)\in\mc{D}\}$. The returned \bidsn solution is the network $N\subseteq
G$ that contains both edges $uv$ and $vu$ for any undirected edge between $u$
and $v$ of the \pname{SF} solution computed for $\ud{G}$. Thus the cost of $N$
is at most twice the cost of the \pname{SF} solution. At the same time, the
optimum \pname{SF} solution in $\ud{G}$ has cost at most that of the optimum
\bidsn solution in $G$, since taking the underlying undirected graph of the
latter is an \pname{SF} solution in $\ud{G}$.

The first part of theorem now follows by using the polynomial time
$2$-approximation algorithm \citet{steiner-forest} for the \pname{SF} problem.
We now show how to solve the \pname{SF} problem in $2^{O(k)}\cdot n^{O(1)}$ time
using the FPT algorithm of  \citet{DBLP:journals/networks/DreyfusW71} for the
\pname{ST} problem which runs  in $3^p\polyn$ time where $p$ is number of
terminals. However, as observed in~\cite{DBLP:conf/icalp/FeldmannM16}, it can
also easily be used for \pname{SF} as well: the optimum solution to \pname{SF}
is a forest and therefore the terminal set can be partitioned so that each part
is a tree in the optimum.
We now use dynamic programming: for each $X\subseteq [k]$ let $\text{OPT}[X]$
denote the minimum cost solution for the instance of \pname{SF} with the
terminal pairs $\ud{\mc{D}}_{X}:=\{\{s_i,t_i\}\mid i\in X\}$. We define
$\text{OPT}[\emptyset]=0$. Then, we have the following recurrence $$
\text{OPT}[X] = \min_{Y\subseteq X, Y\neq \emptyset} \Big\{ \text{DF}[Y] +
\text{OPT}[X\setminus Y] \Big\} $$ where $\text{DF}[Y]$ is the cost obtained by
running the Dreyfus-Wagner algorithm on the instance of \pname{ST} whose
terminal set is $Y$. The correctness of the recurrence follows since the forest
which forms the optimal solution of the instance of \pname{SF} with terminal
pairs $\ud{\mc{D}}_{X}$ must contain some non-empty subset of the terminal pairs
in one of its trees. The final answer that we output is
$\text{OPT}\big[[k]\big]$. Since the running time of the Dreyfus-Wagner
algorithm is $3^p\polyn$, the running time of the dynamic program is
\[
 \sum_{i=1}^{k} \binom{k}{i}\cdot \Big(
\sum_{j=1}^{i} 3^{2j}\polyn \Big) \leq 9\polyn\cdot \Big(\sum_{i=1}^{k}
\binom{k}{i}\cdot 9^{i} \Big) = 2^{O(k)}\polyn \qedhere
\]
\end{proof}

In \autoref{sec:scheme} we will prove that a PSAKS exists for the \bidsnP
problem. One ingredient for this will be the existence of a polynomial time
algorithm that gives a good estimate of the cost of a planar optimum solution.
The $4$-approximation algorithm of \autoref{thm:alg-approx-biDSN} provides a
good lower bound on the cost of the overall optimum, and thus also for the
planar optimum. However, a priori this does not provide a good upper bound,
since the planar optimum might cost a lot more than the overall optimum, and
thus than the approximation. Note though that both algorithms of
\autoref{thm:alg-approx-biDSN} compute planar solutions, and hence they also
upper bound the optimum planar solution. In particular, the existence of a
planar $2$-approximation of the optimum solution to \bidsn implies that the
planar optimum cannot cost more than twice the overall optimum, as summarized
below.

\begin{crl}\label{crl:estimate-biDSN_P}
For any \bidsn instance with optimum solution $N$, there exists a planar
solution $N'$ such that $\cost(N')\leq 2\cost(N)$.
\end{crl}

\section{An approximation scheme for {\altbidsnP}}
\label{sec:scheme}

In this section we prove \autoref{thm:scheme}, which is restated below. Note
that since we have $k$ demand pairs, it follows that the number $|R|$ of
terminals is at most $2k$, where $R= \bigcup_{i=1}^{k} \{s_i, t_i\}$. Henceforth
in this section, we use the upper bound $2k$ on the number of terminals $|R|$
for ease of presentation (when instead we could replace $k$ by $|R|$ in the
running time of \autoref{thm:scheme}).

\thmscheme*

The bulk of the proof is captured by the following result, which generalizes the
corresponding theorem by \citet{borchers-du} for the \pname{ST} problem, and
which is our main technical contribution. In order to facilitate the definition
of a sub-instance to \dsn, we encode the demands of a \dsn instance using a
\emph{pattern graph}~$H$, as also done in~\cite{DBLP:conf/icalp/FeldmannM16}:
the vertex set of $H$ is the terminal set $R$, and $H$ contains the directed
edge $st$ if and only if $(s,t)$ is a demand. Hence the \dsn problem asks for a
minimum cost network $N\subseteq G$ having an $s\to t$ path for each edge $st$
of $H$. Here $\cost(N)$ denotes the cost of a graph (solution) $N$, i.e., the
sum of its edge weights.

\begin{thm}\label{lem:full-comp}
Let $G$ be a bidirected graph, and $H$ a pattern graph on $R\subseteq V(G)$. Let
$N\subseteq G$ be the cheapest planar solution to pattern $H$. For any $\eps>0$,
there exists a set of patterns $\mc{H}$ such that
\begin{enumerate}
 \item $V(H')\subseteq R$ with $|V(H')|\leq 2^{1+1/\eps}$ for each
$H'\in\mc{H}$,
 \item given any feasible solutions $N_{H'}\subseteq G$ for all $H'\in\mc{H}$,
the union~$\bigcup_{H'\in\mc{H}} N_{H'}$ of the these solutions forms a feasible
solution to $H$, and
 \item there exist feasible planar solutions $N^*_{H'}\subseteq G$ for all
$H'\in\mc{H}$ such that\\ $\sum_{H'\in\mc{H}}\cost(N^*_{H'})\leq
(1+\eps)\cdot\cost(N)$.
\end{enumerate}
\end{thm}

Note that the pattern graphs $H'$ of the set $\mc{H}$ in this theorem do not
have to be subgraphs of the given pattern $H$. In fact, as the proof of
\autoref{lem:full-comp} below shows, in general they are not. Before we give
the proof, we describe the consequences of this theorem.

\paragraph*{Consequences of \autoref{lem:full-comp}.}
Our approximation scheme of \autoref{thm:scheme} will compute optimum planar
solutions to all patterns with at most $2^{1+1/\eps}$ terminals using the XP
algorithm of \autoref{thm:alg-planar-biDSN}, and then essentially find the set
$\mc{H}$ of \autoref{lem:full-comp} via a dynamic program. This is captured in
the following proof.

\begin{proof}[Proof of \autoref{thm:scheme}.]
The first step of the algorithm is to consider every possible pattern graph on
at most $g(\eps)=2^{1+1/\eps}$ terminals from $R$. For each such pattern $H'$
the algorithm computes an optimum \bidsnP solution~$N_{H'}$ (if any) using the
algorithm of \autoref{thm:alg-planar-biDSN}. Since any considered pattern graph
has at most $2{g(\eps)\choose 2}< g(\eps)^2$ edges, and (regardless of the input
pattern $H$) there is a total of $2{2k\choose 2}< 4k^2$ possible demands between
the at most~$2k$ terminals of~$R$, the total number of considered pattern graphs
is less than $\sum_{i=0}^{g(\eps)^2}{4k^2 \choose i}=k^{2^{O(1/\eps)}}$. An
optimum planar solution (if it exists) to each of these patterns is computed in
$2^{g(\eps)^{3/2}\log g(\eps)}\cdot n^{O(\sqrt{g(\eps)})}=n^{2^{O(1/\eps)}}$
time via \autoref{thm:alg-planar-biDSN}. Hence up to now the algorithm takes
$n^{2^{O(1/\eps)}}$ time,  as $k\in O(n^2)$.

The next step is to use a dynamic program to compute a solution to the input
pattern $H$ by putting together these pre-computed planar solutions. More
concretely, let $N_1,\ldots,N_p$ be all the solutions computed in the first step
(given in any arbitrary order). For any subset $\mc{N}$ of these planar graphs,
in the following we denote by $\cost(\mc{N}):=\sum_{N\in\mc{N}}\cost(N)$ their
total cost and by $\bigcup\mc{N}:=\bigcup_{N\in\mc{N}}N$ their union. For $1\leq
i\leq p$ and any pattern graph~$H'$, we define
\[\label{eq:sigma1}
\sigma(H',i)=\min\left\{\cost(\mc{N})~\Big\vert~\mc{N}\subseteq\{N_1,
\ldots,N_i\} \text{ and } \bigcup\mc{N} \text{ feasible for } H'\right\}
\]
to be the minimum total cost of a subset of the first $i$ planar graphs
$N_1,\ldots,N_i$ that forms a feasible solution to~$H'$. Note that the total
cost of a set $\mc{N}$ counts edges appearing in more than one graphs of
$\mc{N}$ several times. If no feasible solution to $H'$ can be obtained from
any subset of $N_1,\ldots,N_i$, then we define $\sigma(H',i)$ to be~$\infty$.

Let $\mc{H}$ be the set of patterns given by \autoref{lem:full-comp} for the
optimum planar solution $N$ to $H$, and let $N^*_{H'}$ be the planar solution to
each $H'\in\mc{H}$ given by the theorem. The existence of $N^*_{H'}$ implies
that for each $H'\in\mc{H}$ there is also a feasible planar solution $N_{H'}$
among $N_1,\ldots,N_p$, and by \autoref{lem:full-comp} their union
$\bigcup_{H'\in\mc{H}} N_{H'}$ is a feasible solution to $H$. Thus
\[
\sigma(H,p)\leq \sum_{H'\in\mc{H}} \cost(N_{H'}) \leq \sum_{H'\in\mc{H}}
\cost(N^*_{H'}) = (1+\eps)\cost(N),
\]
where the second inequality follows since each computed planar solution $N_i$
(and thus each $N_{H'}$ where~$H'\in\mc{H}$) is an optimum solution due to
\autoref{thm:alg-planar-biDSN}. We conclude that $\sigma(H,p)$ is the value of a
$(1+\eps)$-approx\-imation to the optimum \bidsnP solution.

To recursively compute $\sigma(H',i)$ for any pattern graph $H'$ on $R$ and any
$1\leq i\leq p$, we keep track of the subset
$\mc{N}^i_{H'}\subseteq\{N_1,\ldots,N_i\}$ of planar graphs that obtain the cost
stored by the following dynamic program in $\sigma(H',i)$. For $i=1$ we just
check whether~$N_1$ is a feasible solution to $H'$. If so, we set
$\sigma(H',1)=\cost(N_1)$ and~$\mc{N}^1_{H'}=\{N_1\}$, while otherwise we set
$\sigma(H',1)=\infty$ and $\mc{N}^1_{H'}=\emptyset$. This obviously computes
$\sigma(H',1)$ correctly. To compute $\sigma(H',i)$ for any~$i\geq 2$, we check
for every pattern graph $H''$ whether $(\bigcup\mc{N}^{i-1}_{H''}) \cup N_i$ is
a feasible solution to~$H'$. Among all such solutions and the graph
$\bigcup\mc{N}^{i-1}_{H'}$ we store the cost of the cheapest option. More
formally, we claim that for $i\geq 2$
\begin{multline}
\label{eq:dp1}
 \sigma(H',i)=\min\Big\{\sigma(H',i-1), \sigma(H'',i-1)+\cost(N_i)~\Big\vert\\
H''\text{ is a pattern with } \big(\bigcup\mc{N}^{i-1}_{H''}\big) \cup N_i\text{ 
feasible for } H'\Big\}.
\end{multline}
If the right-hand side of~\eqref{eq:dp1} is some finite value, we set
$\mc{N}^i_{H'}$ to the subset obtaining the minimum (i.e.,~either
$\mc{N}^{i-1}_{H'}$ or~$\mc{N}^{i-1}_{H''}\cup\{N_i\}$ for some $H''$).
Otherwise, we let $\mc{N}^i_{H'}=\emptyset$.

To show that the recursion given by~\eqref{eq:dp1} is correct, fix $H'$ and
$i\geq 2$, and let $\mc{N}^*\subseteq\{N_1,\ldots,N_i\}$ be the subset of
planar graphs defining $\sigma(H',i)$, i.e., $\mc{N}^*$ minimizes the
right-hand side of~\eqref{eq:sigma1}. We need to show that
$\cost(\mc{N}^i_{H'})=\cost(\mc{N}^*)$. First note that by~\eqref{eq:dp1},
$\bigcup\mc{N}^i_{H'}$ is a feasible solution to $H'$ and is the union of some
subset of $N_1,\ldots,N_i$, so that $\cost(\mc{N}^i_{H'})\geq\cost(\mc{N}^*)$ by
definition of $\mc{N}^*$. In case $N_i\notin\mc{N}^*$, by induction we have
$\cost(\mc{N}^{i-1}_{H'})=\cost(\mc{N}^*)$, and so
$\cost(\mc{N}^i_{H'})\leq\cost(\mc{N}^*)$, since
$\sigma(H',i-1)=\cost(\mc{N}^{i-1}_{H'})$ is considered as one of the values
over which~\eqref{eq:dp1} minimizes. In the other case when~$N_i\in\mc{N}^*$,
consider the graph $\bigcup(\mc{N}^*\setminus\{N_i\})$ obtained by taking the
union of all planar graphs in $\mc{N}^*$ except~$N_i$ (note that it may still
contain edges of $N_i$). Now let $H''$ be the pattern graph on~$R$, which
contains an edge $st$ if and only if $\bigcup(\mc{N}^*\setminus\{N_i\})$
contains an $s\to t$ path. By induction we have
$\cost(\mc{N}^{i-1}_{H''})\leq\cost(\mc{N}^*\setminus\{N_i\})$, and adding
$\cost(N_i)$ to both sides of this inequality we get
$\cost(\mc{N}^{i-1}_{H''})+\cost(N_i)\leq\cost(\mc{N}^*)$, since~$\mc{N}^*$
contains~$N_i$. Moreover, $(\bigcup\mc{N}^{i-1}_{H''})\cup N_i$ is a feasible
solution to $H'$, since $\bigcup\mc{N}^{i-1}_{H''}$ is a feasible solution to
$H''$ and adding~$N_i$ we obtain an $s\to t$ path between terminals $s,t\in R$
if and only if $\bigcup\mc{N}^*$ contains some $s\to t$ path as well. Hence
$\cost(\mc{N}^i_{H'})\leq\cost(\mc{N}^{i-1}_{H''})+\cost(N_i)$, as the latter
term is equal to $\sigma(H'',i-1)+\cost(N_i)$ and is considered as one of the
values over which~\eqref{eq:dp1} minimizes. In conclusion, also if
$N_i\in\mc{N}^*$ we have $\cost(\mc{N}^i_{H'})\leq\cost(\mc{N}^*)$ and so
$\cost(\mc{N}^i_{H'})=\cost(\mc{N}^*)$. Thus the recursion given
in~\eqref{eq:dp1} correctly computes the value of $\sigma(H',i)$ according to
its definition in~\eqref{eq:sigma1}.

To bound the runtime of the dynamic program, recall that there are $2{2k\choose
2}< 4k^2$ possible demands between the at most~$2k$ terminals of
$R$. Hence the number of considered pattern graphs $H'$ is less than~$2^{4k^2}$. 
Recall also that the first step of the algorithm computes (at most)
one planar solution to each pattern on at most $g(\eps)$ terminals of which
there are~$n^{2^{O(1/\eps)}}$ as $k\in O(n^2)$. Thus the asymptotic size of the
table given by all entries~$\sigma(H',i)$ (with $1\leq i\leq p\leq
n^{2^{O(1/\eps)}}$) is $2^{O(k^2)}n^{2^{O(1/\eps)}}$. To compute one entry of
the table via~\eqref{eq:dp1}, we need to consider every pattern $H''$ for each
of which we perform a feasibility check, which can be done in polynomial time.
Thus the runtime per entry is $2^{O(k^2)}\polyn$, and the total runtime of the
algorithm (including the first step) is bounded by
$2^{O(k^2)}n^{2^{O(1/\eps)}}$.
\end{proof}

Note that even though the output of the algorithm is a $(1+\eps)$-approximation
to the cheapest planar solution, the computed solution may not be planar if the
input graph is not. \autoref{lem:full-comp} shows though that the
\citeauthor{borchers-du} Theorem can be extended to a much more general case,
while the inapproximability results of \autoref{thm:lb-approx-biDSN} for \bidsn
and of~\cite{Chitnis19-directed} for \dsnP show that no further generalizations
in this direction are possible. Another consequence of the
\citeauthor{borchers-du} Theorem for the \pname{ST} problem is the existence of
a PSAKS for \pname{ST}, as recently shown by \citet{lokshtanov2017lossy}. By
similar arguments this is also true for \bidsnP, due to \autoref{lem:full-comp}.
A simple observation is that the same kernel is also a $(2+\eps)$-approximate
kernel for \bidsn due to \autoref{crl:estimate-biDSN_P}, which complements the
parameterized $2$-approximation of \autoref{thm:alg-approx-biDSN}. We defer the
proof of the following corollary to the end of this section, since we will
utilize some of the insights gained to prove \autoref{lem:full-comp} (in
particular those from \autoref{lem:paths} below).

\begin{crl}\label{crl:lossy}
The \bidsnP problem admits a polynomial-size approximate kernelization scheme
(PSAKS) of size~$k^{2^{O(1/\eps)}}$. The same kernel is also a polynomial-sized
$(2+\eps)$-approximate kernel for~\bidsn.
\end{crl}

\paragraph*{Proving \autoref{lem:full-comp}.}
We first use the transformations of \autoref{app:degrees} on the cheapest planar
solution $N\subseteq G$, so that each terminal has only~$1$ neighbour, and each
Steiner vertex has exactly $3$ neighbours. Furthermore, let $G_N$ be the graph
spanned by the edge set $\{uv,vu\mid uv\in E(N)\}$, i.e., it is the underlying
bidirected graph of $N$ after performing the transformations of
\autoref{app:degrees} on~$N$. In particular, also in $G_N$ each terminal has
only~$1$ neighbour, and each Steiner vertex has exactly $3$ neighbours. It is
not hard to see that proving \autoref{lem:full-comp} for the solution $N$ in
$G_N$ implies the same result for the original solution in $G$, by reversing all
transformations given in \autoref{app:degrees}.

The proof consists of two parts, of which the first exploits the bidirectedness
of $G_N$, while the second exploits the planarity of $N$. The first part will
identify paths connecting each Steiner vertex to some terminal in such
a way that the paths do not overlap much. This will enable us to select a
subset of these paths in the second part, so that the total weight of the
selected paths is an $\eps$-fraction of the cost of the solution $N$. This
subset of paths will be used to connect terminals to the boundary vertices of
small regions into which we divide $N$. These regions extended by the paths then
form solutions to sub-instances, which together have a cost of $1+\eps$ times
the optimum. The first part is captured by the next lemma.

\begin{lem}\label{lem:paths}
Let $N\subseteq G_N$ be the cheapest planar solution to a pattern graph $H$ on
$R\subseteq V(G_N)$. For every Steiner vertex $v\in V(N)\setminus R$ of $N$
there is a path $P_v$ in $G_N$, such that $P_v$ is a $v\to t$ path to some
terminal~$t\in R$, and the total cost $\sum_{v\in V(N)\setminus R}\cost(P_v)$ of
these paths is~$O(\cost(N))$.
\end{lem}

For the second part we give each vertex $v$ of $N$ a weight $c(v)$, which is
zero for terminals and equal to~$\cost(P_v)$ for each Steiner vertex $v\in
V(N)\setminus R$ and corresponding path $P_v$ given by \autoref{lem:paths}. We
now divide the optimum solution $N$ into regions of small size, such that the
boundaries of the regions have small total weight.

\begin{dfn}
A \emph{region} is a subgraph of~$N$, and given a set of regions, a
\emph{boundary vertex} is a vertex that lies in at least two regions. Given a
planar graph $N$ and a value $r$, a \emph{weighted weak $r$-division} is a set
of regions of $N$ inducing a partition on the edges of $N$, such that each
region has at most~$r$ vertices, and the total weight of all boundary vertices
is an $O(1/\log r)$-fraction of the total weight~$\sum_{v\in V(N)} c(v)$.
\end{dfn}

Unweighted weak $r$-divisions of planar graphs have found many applications in
approximation algorithms, and are for instance defined
in~\cite{italiano2011improved}. For these, the total number of boundary vertices
is at most $O(n/\sqrt r)$ (they are called ``weak'' since they do not bound the
boundary vertices of each region individually). Additionally, the number of
regions is bounded by
$O(n/r)$ in this case~\cite{italiano2011improved,frederickson1981approximation}.
To prove the existence of such weak $r$-divisions for planar graphs, a
separator theorem is applied recursively until each resulting region is small
enough. The bound on the number of boundary vertices follows from the well-known
fact that any planar graph has a small separator of size~$O(\sqrt{n})$.

We however need to bound the total weight of the boundary vertices to obtain
weighted weak $r$-divisions. Unfortunately, separator theorems are not helpful
here, since they only bound the number of vertices in the separator
but cannot bound their weight. Instead we leverage techniques developed for the
Klein-Plotkin-Rao (KPR) Theorem~\cite{KPR-Lee,KPR,fakcharoenphol2003improved}.
Even though the obtained $O(1/\log r)$-fraction for the weighted case is
exponentially worse than the $O(1/\sqrt r)$-fraction for unweighted graphs
obtained in~\cite{italiano2011improved,frederickson1981approximation}, it
follows from a lower bound result of \citet{borchers-du} that for weighted
graphs this is best possible, even if the graph is a tree. In contrast to the
unweighted case, we also do not guarantee any bound on the number of regions,
and we do not need such a bound either. Our proof follows the outlines of the
proof given by Lee~\cite{KPR-Lee} for the KPR Theorem.

\begin{lem}\label{lem:division}
Let $N$ be a directed planar graph for which each vertex has at most $3$
neighbours, and let each vertex $v$ of $N$ have a weight $c(v)\in\mathbb{R}$.
For any $r\in\mathbb{N}$ there is a weighted weak $r$-division.
\end{lem}

We first show how to put \autoref{lem:paths} and \autoref{lem:division}
together in order to prove \autoref{lem:full-comp}, before proving the lemmas.

\begin{proof}[Proof of \autoref{lem:full-comp}]
Given the cheapest planar solution $N$, recall that if $v$ is a Steiner vertex 
then we set the weight $c(v)$ to $\cost(P_v)$, i.e., the path costs of the paths 
$P_v$ of \autoref{lem:paths}, and otherwise we set the weight to $0$. 
Let~$\mc{E}$ be the partition of the edges of $N$ induced by the weighted weak 
$r$-division of $N$ given by \autoref{lem:division} using the weights~$c(v)$. To 
identify the pattern set $\mc{H}$, we first construct a graph $N_E\subseteq G_N$ 
from every edge set $E\in\mc{E}$ and the paths given by \autoref{lem:paths}, 
after which we extract a pattern from it.  We let $r=2^{1/\eps}$ in 
\autoref{lem:division}, so that each region has at most $2^{1/\eps}$ vertices 
and the total weight of the boundary vertices is an $O(\eps)$-fraction of the 
total weight.

We first include the graph spanned by $E$ in $N_E$. For every Steiner vertex $v$
that is a boundary vertex of the $r$-division inducing $\mc{E}$ and is incident 
to some edge of $E$, we also include the $v\to t$ path~$P_v$ given by 
\autoref{lem:paths} in~$N_E$. As $G_N$ is bidirected, the reverse $t\to v$ path 
of $P_v$ also exists in~$G_N$, and we include this path in $N_E$ as well. Let 
$H_E$ be the pattern that has the terminal set of $N_E$ as its vertices, and an 
edge $st$ if and only if there is an $s\to t$ path in $N_E$. The pattern set 
$\mc{H}$ contains all patterns $H_E$ constructed in this way for the edge sets 
$E\in\mc{E}$. We need to show (1) that each pattern $H_E$ contains a bounded 
number of terminals, (2) that the union of any solutions to these patterns is 
feasible for the input pattern $H$, and (3) that there are solutions to the 
patterns with total cost at most $(1+O(\eps))\cdot\cost(N)$. Making $\eps$ 
sufficiently small, this implies \autoref{lem:full-comp}.

For the first part, the bound on the terminals in a pattern $H_E$ follows from
the bound on the vertices spanned by the edges of $E$, as given in
\autoref{lem:division}: the graph $N_E$ contains all terminals spanned by the
edges of $E$, and one terminal for each boundary vertex that is a Steiner vertex
spanned by $E$. Thus the total number of terminals of $N_E$, and therefore also
of $H_E$, is at most $2r=2^{1+1/\eps}$.

For the second part, consider any solutions $N'_E$ to the patterns
$H_E\in\mc{H}$. We need to show that for every edge $st$ of $H$ there is an
$s\to t$ path in the union $\bigcup_{H_E\in\mc{H}} N'_E$. As $N$ is a feasible
solution to $H$, it contains an $s\to t$ path $P\subseteq N$. Consider the
sequence $P_1,P_2,\ldots,P_\ell$ of subpaths of~$P$, such that the edges of each
subpath belong to the same edge set of $\mc{E}$ and the subpaths are of maximal
length under this condition. We construct a sequence $t_0,t_1,\ldots, t_\ell$ of
terminals from these subpaths as follows. As it has maximal length, the
endpoints of each subpath $P_i$ is either a Steiner vertex that is also a
boundary vertex of $\mc{E}$, or a terminal (e.g.\ $s$ and $t$). First we set
$t_0=s$. For any $i\geq 1$, let $E\in\mc{E}$ be the set that contains the edges
of~$P_i$. If the last vertex of $P_i$ is a terminal, then $t_i$ is that
terminal, while if the last vertex is a Steiner vertex $v$, then $t_i$ is the
terminal that the path $P_v$ included in $N_E$ connects to. If the first vertex
of $P_i$ is a terminal, then clearly it is equal to $t_{i-1}$. Moreover, if the
first vertex of $P_i$ is a Steiner vertex $v$, then by construction the graph
$N_E$ contains the reverse $t_{i-1}\to v$ path of $P_v$. Thus $N_E$ contains a
$t_{i-1}\to t_i$ path, and so the pattern $H_E$ contains the edge~$t_{i-1}t_i$.
This implies that also any arbitrary solution $N'_E$ to $H_E$ contains a
$t_{i-1}\to t_i$ path, and therefore the union $\bigcup_{E\in\mc{E}} N'_E$ of
arbitrary solutions contains a $t_0\to t_\ell$ path via the intermediate
terminals $t_i$ where $i\in\{1,\ldots,\ell-1\}$. As $t_0=s$ and $t_\ell=t$, this
means that the union is feasible for $H$.

For the third part we just bound $\sum_{E\in\mc{E}}\cost(N_E)$, i.e., the
special solutions $N^*_{H'}$ of the theorem statement are exactly the solutions
$N_E$ constructed above, which are subgraphs of the planar graph $G_N$ and thus
are planar as well. Note that the cost of each $N_E$ is the cost of the edge set
$E$ plus the cost of the paths~$P_v$ and their reversed paths attached to the
boundary Steiner vertices $v$ incident to~$E$. The sum of the costs of all edge
sets $E\in\mc{E}$ contribute exactly the cost of $N$ to
$\sum_{E\in\mc{E}}\cost(N_E)$, since $\mc{E}$ is a partition of the edges of
$N$. As we assume that each boundary vertex $v$ of $\mc{E}$ has at most three
neighbours, $v$ is incident to a constant number of edge sets of~$\mc{E}$. Thus
$\sum_{E\in\mc{E}}\cost(N_E)$ also contains the cost of path $P_v$ only a
constant number times: twice for each set $E\in\mc{E}$ incident to boundary
vertex $v$, due to $P_v$ and its reverse path, which in a bidirected instance
has the same cost as $P_v$. By \autoref{lem:division}, $\sum_{v\in B}c(v)\leq
O(\eps)\cdot \sum_{v\in V(N)}c(v)$, where~$B$ is the set of boundary vertices
of~$\mc{E}$, and the cost $c(v)$ of a vertex is the cost of the path $P_v$ if
$v$ is a Steiner vertex, and $0$ otherwise. Hence all paths $P_v$ and their
reverse paths contained in all the graphs $N_E$ for $E\in\mc{E}$ contribute at
most $O(\eps)\cdot \sum_{v\in V(N)}c(v)=O(\eps)\cdot\sum_v\cost(P_v)$ to
$\sum_{E\in\mc{E}}\cost(N_E)$. By \autoref{lem:paths},
$\sum_v\cost(P_v)=O(\cost(N))$, and so
$\sum_{E\in\mc{E}}\cost(N_E)=(1+O(\eps))\cdot\cost(N)$.
\end{proof}

We now turn to proving the two remaining lemmas, starting with finding paths
for Steiner vertices for \autoref{lem:paths}.

\begin{proof}[Proof of \autoref{lem:paths}]
We begin by analysing the structure of optimal \dsn solutions in bidirected
graphs, based on \autoref{lem:cycles}. Here a \emph{condensation graph} of a
directed graph results from contracting each strongly connected component, which
hence is a DAG.

\begin{claim}\label{clm:sym-structure}
For any solution $N\subseteq G_N$ to a pattern $H$, there is a solution
$M\subseteq G_N$ to $H$ with $\cost(M)\leq\cost(N)$ and $\ud{N}=\ud{M}$, such
that the condensation graph of $M$ is a poly-forest.
\end{claim}
\begin{proof}
Let $C$ be a subgraph of $N$, which induces a maximal $2$-connected component in
$\ud{N}$.
Assume there is a vertex pair $u,v\in V(C)$ for which no $u\to v$ path exists in
$C$. As $C$ induces a $2$-connected component in~$\ud{N}$, by Menger's
Theorem~\cite{MR2159259} there are two internally disjoint poly-paths $P$ and
$Q$ between $u$ and $v$ in $C$, which together form a poly-cycle $O$. By
\autoref{lem:cycles} we may replace $O$ by a directed cycle without increasing
the cost, and so that there is a directed path for every pair of vertices for
which such a path existed before. Additionally, this step introduces a $u\to v$
path along this new directed cycle in $C$. Repeating this for any pair of
vertices for which no directed path exists in $C$ will eventually result in a
strongly connected component. Hence we can make every component of $N$, which
induces a maximal $2$-connected component in $\ud{N}$, strongly connected
without increasing the cost. Note also that $\ud{N}$ does not change.

After this procedure we obtain the graph $M\subseteq G_N$. The maximal
$2$-connected components in $\ud{M}$ induce subgraphs of the strongly connected
components of $M$ (they may be subgraphs due to cycles of length $2$ in $M$,
which are bridges in~$\ud{M}$). Contracting all strongly connected components of
$M$ must therefore result in a poly-forest, as any poly-cycle in the
condensation graph would also induce a cycle in~$\ud{M}$.
\cqed
\end{proof}

By \autoref{clm:sym-structure} we may assume w.l.o.g.\ that the condensation
graph of the optimum solution $N$ is a poly-forest such that every terminal has
one neighbour and every Steiner vertex has three neighbours. Consider a weakly
connected component $C$ of $N$, i.e., inducing a connected component of
$\ud{N}$. We first extend $C$ to a strongly connected graph $C'$ as follows.
Let~$F$ be the edges of $C$ that do not lie in a strongly connected component,
i.e., they are the edges of the condensation graph of $C$. Let $\widetilde
F=\{uv\mid vu\in F\}$ be the set containing the reverse edges of $F$, and let
$C'$ be the strongly connected graph spanned by all edges of $C$ in addition to
the edges in~$\widetilde F$. Note that adding $\widetilde F$ to $C$ increases
the cost by at most a factor of two as $G_N$ is bidirected, while the number of
neighbours of any vertex does not change. We claim that in fact $C'$ is a
\emph{minimal} \scss solution to the terminal set $R_C\subseteq R$ contained
in~$C$, that is, removing any edge of $C'$ will disconnect some terminal pair of
$R_C$.

For this, consider any $s\to t$ path of $C'$ containing an edge $e\in\widetilde
F$ for some terminal pair $s,t\in R_C$. As the edges $F$ of the condensation
graph of $C$ form a poly-tree, every path from $s$ to $t$ in $C'$ must pass
through~$e$. In particular there is no $s\to t$ path in $C$, and thus there is
no edge $st$ in the pattern graph~$H$. Or equivalently, for any terminal pair
$s,t\in R_C$ for which there is a demand $st\in E(H)$, no $s\to t$ path in $C'$
passes through an edge of $\widetilde F$. Thus for such a terminal pair the set
of paths from $s$ to $t$ is the same in $C'$ and~$C$. Since $N$ is an optimum
solution so that every edge $e$ of $C$ is necessary for some pair $s,t\in R_C$
with $st\in E(H)$, the edge $e$ is still necessary in $C'$. Moreover, for any of
the added edges $uv\in\widetilde F$ the reverse edge $vu\in F$ was necessary in
$C$ to connect some $s\in R_C$ to some $t\in R_C$. As observed above, $uv$ is
necessary to connect $t$ to $s$ in~$C'$, since the edges $F$ of the condensation
graph form a poly-tree.

As $C'$ is a minimal \scss solution to the terminals $R_C$ contained
within, it is the union of an in-arborescence $A_{in}$ and out-arborescence
$A_{out}$, both with the same root $r\in R_C$ and leaf set~$R_C\setminus\{r\}$,
since every terminal only has one neighbour in $G_N$. A \emph{branching point}
of an arborescence $A$ is a vertex with at least two children in $A$. We let
$W\subseteq V(C')$ be the set consisting of all terminals $R_C$ and all
branching points of $A_{in}$ and~$A_{out}$. We will need that any vertex of $C'$
has a vertex of $W$ in its close vicinity. That is, if $N[v]=N(v)\cup\{v\}$
denotes the closed neighbourhood of a vertex $v$ and $N^2[v]=\bigcup_{u\in
N[v]}N[u]$, we prove the following.

\begin{claim}\label{clm:nb-W}
For every vertex $v$ of $C'$, there is a vertex of $W$ in $N^2[v]$.
\end{claim}
\begin{proof}
Assume that $v\notin W$, since otherwise we are done. Such a vertex must be a
Steiner vertex, and hence has exactly three neighbours in~$C'$. As $v$ is not a
branching point of $A_{in}$ or $A_{out}$, this means that $v$ is incident to two
edges of $A_{in}$ and two edges of $A_{out}$. This can either mean that there
are three edges incident to $v$ of which one lies in both $A_{in}$ and
$A_{out}$, or there are four edges incident to $v$ of which two connect to the
same neighbour of $v$ but point in opposite directions. Consider the former case
first, i.e., there is one of the edges $e$ incident to $v$ that lies in the
intersection of the two arborescences, another incident edge $f^v_{in}$ that
lies in $A_{in}$ but not in~$A_{out}$, and a third incident edge $f^v_{out}$
that lies in $A_{out}$ but not in $A_{in}$. Now assume for a contradiction that
the neighbour $u$ of $v$ incident to $e$ also does not belong to~$W$, and
w.l.o.g., let $e=vu$ (for the symmetric when $e=uv$, an analogous argument to
the following exists). In particular, both $f^v_{in}$ and $f^v_{out}$ are
incoming edges to $v$. By the same observations as for $v$, there must be an
incident edge $f^u_{in}$ to $u$ that lies in $A_{in}$ but not in $A_{out}$, and
an incident edge $f^u_{out}$ that lies in $A_{out}$ but not in~$A_{in}$. Both
these edges must be outgoing of $u$. See \autoref{fig:Wcase1}.

\begin{figure}[t]
 \centering
  \includegraphics[width=0.49\textwidth]{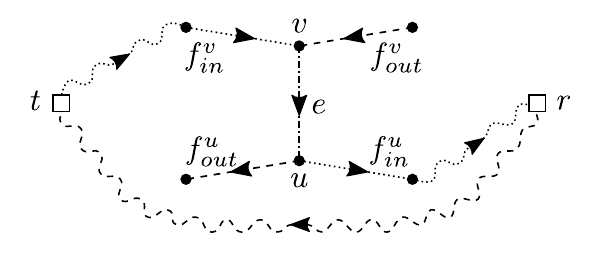}
  \includegraphics[width=0.49\textwidth]{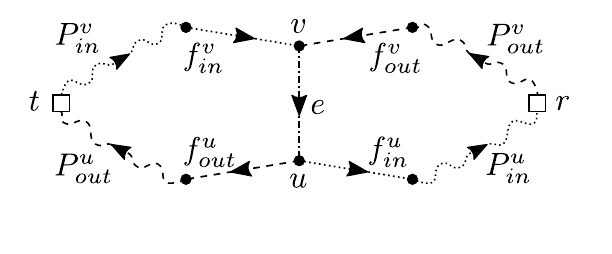}
  \caption{The case when $e=vu$ lies in both $A_{in}$ (dotted) and $A_{out}$
(dashed). Edges are straight lines, paths are wavy and may intersect. On the
left: assuming that $A_{out}$ contains a path from the root $r$ to terminal $t$
that does not contain $e$ leads to a contradiction, since~$f^v_{out}$ would then
be redundant. On the right: otherwise, the paths connecting $r$ and $t$ with $v$
and $u$ imply the existence of a poly-cycle, which due to $e$ has a chord, again
leading to a contradiction.}
\label{fig:Wcase1}
\end{figure}

The in-arborescence $A_{in}$ contains a $t\to r$ path from some terminal $t\in
R_C$ to the root $r$ passing through~$e$. We claim that $A_{out}$ must contain
an $r\to t$ path to the same terminal $t$ passing through $e$ as well. If this
were not the case there would be some other $r\to t$ path of $A_{out}$ not
containing~$e$. Together with the $t\to v$ subpath of the $t\to r$ path in
$A_{in}$, this implies an $r\to v$ path not containing~$f^v_{out}$: the latter
edge is not contained in $A_{in}$ and therefore cannot be part of the $t\to v$
subpath. However this means that every terminal reachable from $r$ via $v$ in
$C'$ is reachable by a path not containing $f^v_{out}$. As this edge is not
contained in $A_{in}$, it could safely be removed from $C'$ without
disconnecting any terminal pair. This would contradict the minimality of~$C'$,
which means there must be an $r\to t$ path in $A_{out}$ that passes through $v$.

For this terminal $t$, we can conclude that there is a $t\to v$ path
$P^v_{in}\subseteq A_{in}$ ending in $f^v_{in}$, a $u\to r$ path
$P^u_{in}\subseteq A_{in}$ starting in $f^u_{in}$, but also an $r\to v$ path
$P^v_{out}\subseteq A_{out}$ ending in $f^v_{out}$, and a $u\to t$ path
$P^u_{out}\subseteq A_{out}$ starting in $f^u_{out}$. Moreover, none of these
four paths contains $e$. Note that the union $P^v_{in}\cup P^u_{in}\cup
P^v_{out}\cup P^u_{out}$ of the four paths contains a poly-cycle $O$ for which
$e$ is a \emph{chord}, i.e., it connects two non-adjacent vertices of $O$.

The strongly connected component $C'$ was constructed from the component $C$ of
the optimum solution $N$ by adding the set $\widetilde F$ of reverse edges to
some existing edge set $F$ of $C$. Hence, even if $O$ and/or $e$ do not exist in
$N$, there still exists a poly-cycle $O'$ in $N$ with the same vertex set and
underlying undirected graph as $O$, and an edge $e'$ that is a chord to~$O'$,
which may be $e$ or its reverse edge. This contradicts the optimality of $N$ by
\autoref{lem:cycle-split}, and thus $u\in N(v)$ is in~$W$.

It remains to consider the case when $v$ has four incident edges. This means
that for one neighbour $u$ of $v$ there are two edges $uv$ and $vu$ in $C'$ of
which one belongs to $A_{in}$ and the other to $A_{out}$. W.l.o.g., let $uv$
belong to $A_{out}$ (for the other case when $uv$ belongs to $A_{in}$, by
symmetry an analogous argument to the following exists). Now let $w_{out}$ and
$w_{in}$ be the other two neighbours of $v$, for which the edge $vw_{out}$ is in
$A_{out}$, while the edge $w_{in}v$ is in $A_{in}$. If either $w_{in}$ or
$w_{out}$ is in $W$, we are done. Hence assuming that $w_{in},w_{out}\notin W$,
just as $v$, both $w_{in}$ and $w_{out}$ are Steiner vertices with three
neighbours, each incident to two edges of $A_{in}$ and two edges of $A_{out}$.
If either $w_{in}$ or $w_{out}$ has an incident edge that lies in the
intersection of $A_{in}$ and $A_{out}$, by the same argument as for $v$ above,
some vertex of $N(w_{in})\cup N(w_{out})$ must lie in $W$. As $N(w_{in})\cup
N(w_{out})\subseteq N^2[v]$ this
would conclude the proof.

\begin{figure}[t]
 \centering
  \includegraphics[width=0.49\textwidth]{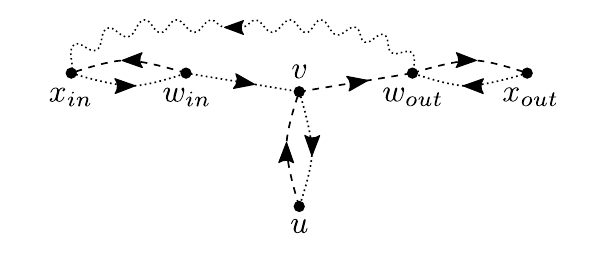}
  \caption{The case when $vu$ lies in $A_{in}$ (dotted) and $uv$ lies in
$A_{out}$ (dashed). Assuming that neither $w_{in}$ nor $w_{out}$ has an incident
edge lying in both arborescences leads to a contradiction, since $vu$ being a
bridge in $\ud{C'}$ implies a path of $A_{in}$ (wavy) connecting $w_{out}$
with~$x_{in}$, but $x_{in}w_{in}$ is also a bridge in $\ud{C'}$.}
\label{fig:Wcase2}
\end{figure}

Hence assume that neither $w_{in}$ nor $w_{out}$ has an incident edge lying in
both arborescences. Thus $w_{out}$ has a neighbour $x_{out}\neq v$ such that
$w_{out}x_{out}\in E(A_{out})$ and $x_{out}w_{out}\in E(A_{in})$, and $w_{in}$
has a neighbour $x_{in}\neq v$ such that $x_{in}w_{in}\in E(A_{in})$ and
$w_{in}x_{in}\in E(A_{out})$. See \autoref{fig:Wcase2}. Note that $x_{in}\neq
x_{out}$ as otherwise $A_{in}$ would have a vertex of out-degree more than one.
Moreover, by the following argument, we can conclude that in $\ud{C'}$, all
three undirected edges $vu$, $x_{out}w_{out}$, and $x_{in}w_{in}$ are bridges.
Consider any edge $e$ in the component $C$ of the optimum solution $N$ from
which $C'$ was constructed. By \autoref{lem:cycle-split}, the reverse edge of
$e$ can only exist in $C$ if $e$ does not lie on any poly-cycle. That is, if $e$
and its reverse edge exist in $C$ then the corresponding edge in $\ud{C}$ is a
bridge. To obtain $C'$ from $C$ we added $\widetilde F$, which contains all
reverse edges of the condensation graph of $C$. From \autoref{clm:sym-structure}
we concluded that the condensation graph of $C$ is a poly-forest. Thus any edge
of $C'$ for which the reverse edge exists in $C'$ as well, must correspond to a
bridge in~$\ud{C'}$, including $vu$, $x_{out}w_{out}$, and $x_{in}w_{in}$, which
all lie in $A_{in}$. Note also that by the same observations, $v$, $w_{out}$,
and $w_{in}$ lie in the same $2$-connected component of $\ud{C'}$, as the
reverse edges of $vw_{out}$ and $w_{in}v$ do not exist in $C'$.

This means that $A_{in}$ contains a path starting in $x_{out}w_{out}$, which
reaches the root of $A_{in}$ by passing through $vu$, as the latter is a bridge
of $\ud{C'}$ while $v$ and $w_{out}$ lie in the same $2$-connected component
of~$\ud{C'}$. Since neither $v$ nor $w_{in}$ is a branching point of~$A_{in}$
while $x_{in}w_{in}, w_{in}v, vu \in E(A_{in})$, this path of $A_{in}$ contains
the subpath given by the sequence $x_{in}w_{in}vu$. But this means that there is
a path from $x_{out}w_{out}$ to $x_{in}$ that does not pass through
$w_{in}x_{in}\in A_{out}$. This contradicts the fact that $x_{in}w_{in}$ is a
bridge of $\ud{C'}$, and thus concludes the proof.
\cqed
\end{proof}

As the graph $G_N$ is bidirected, for any $v$-$u$ path $P$ in the underlying
undirected graph $\ud{G}_N$ of $G_N$, there exists a corresponding directed
$v\to u$ path in $G_N$ of the same cost. Therefore, we can ignore the directions
of the edges in $C'$ and the arborescences $A_{out}$ and $A_{in}$ to identify
the paths $P_v$ for Steiner vertices $v$ of $N$. Thus we will only consider
paths in the graphs $\ud{C'}$, $\ud{A}_{out}$, and $\ud{A}_{in}$ from now on. In
particular, we exploit the following observation found in~\cite{du1991better}
(and also used by~\citet{borchers-du}) on undirected
trees.\footnote{In~\cite{du1991better,borchers-du} the claim is stated for
binary trees, but this is an assumption that can be made w.l.o.g.\ using similar
vertex degree transformations as presented in \autoref{app:degrees}.}

\begin{claim}[{\cite[Lemma~3.2]{du1991better}}]\label{clm:tree-paths}
For any undirected tree $T$ we can find a path $P_v\subseteq T$ for every
branching point $v$, such that $P_v$ leads from $v$ to some leaf of $T$, and
all these paths $P_v$ are pairwise edge-disjoint.
\end{claim}

If a Steiner vertex $v$ of $C'$ is a branching point of $A_{out}$ ($A_{in}$),
we let $P_v$ be the corresponding path in $\ud{A}_{out}$ ($\ud{A}_{in}$) given
by \autoref{clm:tree-paths} from $v$ to some leaf of $A_{out}$ ($A_{in}$),
which is a terminal. Note that paths in $\ud{A}_{in}$ may overlap with paths in
$\ud{A}_{out}$. However any edge in the union of all the paths $P_v$ chosen so
far is contained in at most two such paths, one for a branching point of
$A_{out}$ and one for a branching point of~$A_{in}$.

It remains to choose a path $P_v$ for every Steiner vertex $v$ that is neither a
branching point of $A_{out}$ nor of $A_{in}$, i.e., for every vertex not in $W$.
By \autoref{clm:nb-W} for any such vertex $v\notin W$ there is a vertex $u\in
N^2[v]$ for which $u\in W$. If $u$ is a terminal, then the path $P_v$ is simply
the edge~$vu$ if $u\in N(v)$ or the corresponding path $vwu$ for some $w\in
N(v)$ otherwise. If $u$ is not a terminal but a branching point of $A_{out}$ or
$A_{in}$, then we chose a path $P_u$ for $u$ above. In this case, $P_v$ is the
path contained in the walk given by extending the path $P_u$ by the edge~$vu$ or
the path $vwu$, respectively. Note that, as any vertex of $C'$ has at most $3$
neighbours, any terminal or branching point $u\in W$ can be used in this way
for some vertex $v\notin W$ at most nine times. Therefore any edge in the union
of all chosen paths is contained in $O(1)$ paths. Consequently the total cost
$\sum_{v\in V(N)\setminus R}\cost(P_v)$ is $O(\cost(C'))$, and as $\cost(C')\leq
2\cost(C)$ we also get $\sum_{v\in V(N)\setminus R}\cost(P_v)=O(\cost(C))$.

We may repeat these arguments for every weakly connected component of $N$ to
obtain the lemma.
\end{proof}

Next we give the proof of \autoref{lem:division}, which shows that there are
weighted weak $r$-divisions for planar graphs.

\begin{proof}[Proof of \autoref{lem:division}]
We will not be concerned with the edge weights of $N$ and accordingly define
the distance function $d_M(u,v)$ for any subgraph $M$ of $N$ to be the
\emph{hop-distance} between $u$ and $v$ in~$\ud{M}$, i.e., the minimum number
of edges on any path from $u$ to $v$ in $\ud{M}$. The idea (as outlined
in~\cite{KPR-Lee,fakcharoenphol2003improved}) is to iteratively ``chop'' the
vertices of $N$ into disjoint sets that induce annuli of bounded thickness
measured in the hop-distance, using the following random process. For a fixed
value $\tau$, if we are given some connected graph $M$, then we first choose an
offset $\tau_0\in\{1,\ldots,\tau\}$ uniformly at random and an arbitrary vertex
$v_0$ of~$M$. A so-called \emph{$\tau$-chop} then is the partition of the
vertices of $M$ defined by the sets
\[
A_0&=\{v\in V(M)\mid d_M(v_0,v)<\tau_0\} \text{ and}\\
A_i&=\{v\in V(M)\mid \tau_0+(i-1)\tau\leq d_M(v_0,v)<\tau_0+i\tau\} \text{
for } i\geq 1.
\]

We define a $\tau$-chop of a disconnected graph as the partition given by the
union $\bigcup_{\ell}\mc{P}_\ell$ of $\tau$-chops~$\mc{P}_\ell$, $\ell\geq 1$, 
of the connected components, where for each component we choose an offset 
$\tau_0$ uniformly at random and an arbitrary vertex~$v_0$. Finally, a 
$\tau$-chop of a partition $\mc{P}$ is the refined partition given by the union 
of $\tau$-chops on each subgraph induced by a set in~$\mc{P}$, again choosing a 
$v_0$ and a $\tau_0$ for every component of the subgraphs. Hence we may start 
with $N$ and iteratively perform $\tau$-chops to obtain smaller and smaller 
subsets of vertices.

\citet{KPR-Lee} now proves the following claim, where the \emph{weak diameter}
of a subgraph $M\subseteq N$ is the maximum hop-distance of any two vertices of
$M$ measured in the underlying graph $N$, i.e., $\max_{u,v\in V(M)} d_N(u,v)$.
Note that this claim holds independent of the choices of the vertices $v_0$ and
the offsets $\tau_0$.

\begin{claim}[Lemma~2 in~\cite{KPR-Lee}]\label{clm:Lee}
If $N$ excludes $K_h$ as a minor, then any sequence of $h-1$ iterated
$\tau$-chops on $N$ results in a partition $\mc{P}$ of $V(N)$, such that each
graph induced by a set $S\in\mc{P}$ has weak diameter~$O(h\tau)$.
\end{claim}

Let $\mc{P}$ be the partition of $V(N)$ from \autoref{clm:Lee}. Since $N$ is
planar, it excludes $K_5$ as a minor, and so the weak diameter of each set
$S\in\mc{P}$ is $O(\tau)$. We define a partition $\ud{\mc{E}}$ of the edges of
$\ud{N}$, consisting of sets $E_S\subseteq E(\ud{N})$ for each~$S\in\mc{P}$. In
particular, if $S$ is the set containing the lexicographically smaller vertex
incident to an edge $e$ of $\ud{N}$, then $e$ is contained in $E_S$. Note that
the weak diameter of a region $M_S$ spanned by an edge set $E_S$ is at most the
weak diameter of the graph induced by $S$ plus $2$, i.e., also the weak diameter
of $M_S$ is $O(\tau)$. Since $\ud{N}$ has maximum degree~$3$, the weak diameter
bounds the number of vertices in each region by $|V(M_S)|=2^{O(\tau)}$ for every
$S\in\mc{P}$. As $\ud{\mc{E}}$ corresponds to a partition $\mc{E}$ of the edges
of $N$, for some~$\tau=\Theta(\log r)$ we obtain an $r$-division given by
$\mc{E}$ with the required bound on the sizes of the regions.

It remains to bound the weight of the boundary vertices, for which we bound the
expected weight among the random choices of offsets. More concretely, note that
when performing a single $\tau$-chop on a connected graph $M$ from a fixed
vertex~$v_0$, two adjacent vertices $u,v$ end up in different sets $S$ with
probability at most $1/\tau$ by the choice of the offset~$\tau_0$ and the
definition of the sets~$A_i$,~$i\geq 0$. We assign the edge $uv$ of $\ud{N}$ to
the set $E_S$ containing the lexicographically smaller vertex among $u$ and $v$.
Thus any vertex $w$, which has degree at most $3$ in $\ud{N}$, is a boundary
vertex of a region spanned by some set $E_S$ with probability at most $3/\tau$
when performing a single $\tau$-chop from a fixed vertex $v_0$. As we perform
$h-1=4$ iterative $\tau$-chops, the expected weight of the boundary vertices is 
at most $\frac{3\cdot 4}{\tau} \sum_{v\in V(N)} c(v)$. Hence, since $N$ is 
planar and by our choice of $\tau=\Theta(\log r)$, there exists an $r$-division 
with only a $O(1/\log r)$-fraction of the total vertex weight in the boundary 
vertices.
\end{proof}

\paragraph*{Proving \autoref{crl:lossy}.}
Finally, we can also prove that \autoref{lem:full-comp} implies a PSAKS for
\bidsnP, by utilizing some of the insights of the above proofs.
The proof essentially follows the same lines as the one given for the
\pname{ST} problem by \citet{lokshtanov2017lossy} based on the
\citeauthor{borchers-du} Theorem.

\begin{proof}[Proof of \autoref{crl:lossy}]
To obtain a polynomial-sized $(1+\eps)$-approximate kernel we proceed similar
to the algorithm described at the beginning of this section, by first computing
an optimum solution for every possible pattern graph on at most
$g(\eps)=2^{1+1/\eps}$ terminals from $R$. Using the XP algorithm of
\autoref{thm:alg-planar-biDSN}, this takes $n^{2^{O(1/\eps)}}$ time, as
determined before. If $\eps$ is a constant, this amounts to a polynomial
runtime. Taking the union of all precomputed solutions gives a graph, which due
to \autoref{lem:full-comp} contains a $(1+\eps)$-approximation to the optimum of
the input graph~$G$ (by the same arguments showing that the algorithm computes a
$(1+\eps)$-approximation). However the union is not a kernel, since its size is
not necessarily bounded as a function of the parameter $k$. In particular, it
may contain many vertices and the edge weights might be large.

To reduce the number of vertices, we apply the vertex degree transformations
from \autoref{app:degrees} to each computed optimum solution $N_{H'}$ to
patterns~$H'$ on at most $g(\eps)$ terminals. In particular, every Steiner
vertex of $N_{H'}$ now has exactly three neighbours. We use the insights from
the proof of \autoref{lem:paths} to argue that $N_{H'}$ has a bounded number of
vertices. By \autoref{clm:sym-structure} we may assume that $N_{H'}$ is an
optimum solution to $H'$ for which the condensation graph is a poly-forest. Now
consider a weakly connected component $C$ of~$N_{H'}$. As argued in the proof of
\autoref{lem:paths}, if we add to $C$ the edge set $\widetilde F$, which
contains the reverse edges to those of the condensation graph of~$C$, then we
obtain a minimal \scss solution~$C'$ for the terminal set $R_C\subseteq R$
contained in $C$. This means that $C'$ is the union of an in-arborescence
$A_{in}$ and an out-arborescence $A_{out}$, both rooted at some terminal $r\in
R_C$ and with leaves from~$R_C$. The number of branching points of each of these
arborescences is at most~$|R_C|$. Hence the set $W$ (as defined earlier) of
branching points and terminals $R_C$ contains at most $3|R_C|$ vertices. Due to
\autoref{clm:nb-W} we can map any vertex not contained in $W$ to a vertex in $W$
at hop-distance at most $2$. Since every vertex not in $W$ is a Steiner vertex
and has three neighbours, at most $9$ vertices map to any particular vertex of
$W$. Thus the number of vertices of $C$ not in $W$ is at most $27|R_C|$, which
brings the total to at most $30|R_C|$ after adding~$W$. This means that the
number of vertices of $N_{H'}$ is at most $30g(\eps)=2^{O(1/\eps)}$. As
calculated earlier, the total number of pattern graphs $H'$ on at most $g(\eps)$
terminals is~$k^{2^{O(1/\eps)}}$. Hence taking the union of all computed
solutions $N_{H'}$ after applying the vertex degree transformations of
\autoref{app:degrees} gives a graph $G'$ with $2^{O(1/\eps)}\cdot
k^{2^{O(1/\eps)}}=k^{2^{O(1/\eps)}}$ vertices, which is polynomial in the 
parameter $k$ if $\eps$ is constant.

For the edge weights, \citet{lokshtanov2017lossy} show how to round them in such
a way that each edge weight can be stored using $O(\log(k/\eps))$ bits for the
\pname{ST} problem. Here we will need slightly more bits. As an ingredient we
use that a polynomial time constant approximation algorithm exists, which is
provided by \autoref{thm:alg-approx-biDSN}. In particular, let $M\subseteq G$ be
a $4$-approximate solution computed by this algorithm for the input graph $G$.
If the weight of an edge $e$ of the union graph $G'$ currently is $w(e)$, then
we define a rounded integer weight
\[
\widehat w(e)=\left\lfloor \frac{|E(G')|w(e)}{\eps \cost(M)} \right\rfloor,
\]
and set the edge weights of the union graph $G'$ to $\widehat w(e)$ instead. By
\autoref{crl:estimate-biDSN_P} we have $\cost(N)\leq 2\cost(M)$ for the optimum
planar solution $N$ to the input instance, and so we may remove any edge of cost
more than $2\cost(M)$. This implies that $\widehat w(e)\leq 2|E(G')|/\eps$, 
which asymptotically is $k^{2^{O(1/\eps)}}$, as the graph $G'$ has 
$k^{2^{O(1/\eps)}}$ vertices. Hence each edge weight can be encoded using 
$2^{O(1/\eps)}\log(k)$ bits, and the size of the kernel including the edge 
weights is~$k^{2^{O(1/\eps)}}$.

It remains to show that rounding the edge weights does not distort the solution 
costs by too much. Let $N'\subseteq G'$ be a $\beta$-approximation of the 
optimum planar solution in the kernel, i.e., using weights $\widehat w(e)$. Let 
also $N^\star$ be the optimum planar solution in $G'$ when using the original 
weights $w(e)$. In particular, we get $\sum_{e\in E(N')}\widehat 
w(e)\leq\beta\sum_{e\in E(N^\star)}\widehat w(e)$, since the optimum planar 
solution in the kernel has cost at most that of $N^\star$ according to 
weights~$\widehat w(e)$. Since $N'$ has at most as many edges as $G'$, the cost 
of $N'$ measured by the original edge weights $w(e)$ compared to $N^\star$ and 
$M$ is
\[
\sum_{e\in E(N')} w(e)
&\leq \sum_{e\in E(N')}\left(\frac{\eps \cost(M)}{|E(G')|}(1+\widehat
w(e))\right)\\
&\leq \eps\cost(M)+\frac{\eps \cost(M)}{|E(G')|}\sum_{e\in 
E(N')}\widehat w(e)\\
&\leq \eps\cost(M)+\frac{\eps \cost(M)}{|E(G')|}\cdot\beta\sum_{e\in
E(N^\star)} \widehat w(e)\\
&\leq \eps\cost(M)+\frac{\eps \cost(M)}{|E(G')|}\cdot\beta\sum_{e\in
E(N^\star)}\frac{|E(G')|w(e)}{\eps\cost(M)}\\
&= \eps\cost(M) + \beta\sum_{e\in E(N^\star)} w(e).
\]

Each of $N'$ and $N^\star$ can clearly be lifted to a solution in the input
instance with the same or lower cost (when using weights $w(e)$) in polynomial
time. By \autoref{lem:full-comp}, we have $\cost(N^\star)\leq(1+\eps)\cost(N)$
for the optimum planar solution $N$ of the input instance. At the same time, $M$
is a $4$-approximation for the input instance, which means that $\cost(M)\leq
4\cost(N)$. By the above calculations we hence get that $\cost(N')\leq
(1+5\eps)\beta\cost(N)$. By making $\eps$ sufficiently small, this implies the
desired approximation bound for \bidsnP. Moreover, due to
\autoref{crl:estimate-biDSN_P}, the planar solution $N'$ is also a
$2(1+5\eps)\beta$-approximation of the overall optimum of the input instance,
and thus the claimed approximation bound for \bidsn follows as well.
\end{proof}

\section{Computing optimum solutions in bidirected graphs}
\label{sec:alg-opt}

In this section we show how to compute optimum solutions to \biscss and
to \bidsnP, and we start with the latter.

\subsection{An XP algorithm for \altbidsnP}

In this section we prove \autoref{thm:alg-planar-biDSN}, which is restated
below.

\thmalgplanarbiDSN*
\begin{proof}
The proof hinges on the fact that an optimum solution $N\subseteq G$ to \bidsnP
has treewidth less than $6\sqrt k$. To show this, assume to the contrary that
the treewidth of $N$ is at least $6\sqrt{k}$. It is well-known that this
implies that $N$ contains a $6\sqrt{k}\times 6\sqrt{k}$ grid minor. Consider a
planar drawing of $N$. Since there are only $k$ terminal pairs, we can have at
most $2k$ terminals. By the pigeon-hole principle, the grid minor contains some
$3\times 3$ grid minor $M$ for which no terminal touches any of the faces in the
interior of $M$ in the drawing. We can see $M$ as consisting of a poly-cycle $O$
with all other vertices of $M$ touching faces in the interior of~$O$ in the
drawing. In particular, removing $O$ from $M$ will leave a non-empty connected
component (the interior of~$O$) which contains no terminals. This however
contradicts \autoref{lem:cycle-split}. Note that such a solution would also be
planar, and thus the treewidth of the optimum planar solution is~$O(\sqrt k)$.

Since by \autoref{thm:alg-FMarx} there is an algorithm to compute the optimum
among all solutions of treewidth at most $\omega$ in time
$2^{O(k\omega\log\omega)}\cdot n^{O(\omega)}$, the above treewidth bound
implies \autoref{thm:alg-planar-biDSN}.
\end{proof}

Note that the algorithm in \autoref{thm:alg-planar-biDSN} does not necessarily
compute a planar solution, if the input is not planar, but the found solution
will still have cost at most that of the cheapest planar solution.
\autoref{thm:lb-scheme-biDSN} shows that the running time obtained in
\autoref{thm:alg-planar-biDSN} for \bidsnP is asymptotically optimal under ETH.

\subsection{FPT algorithm for \altbiscss}

We now turn to \biscss (without restricting the optimum) and show that this
problem is FPT for parameter~$k$ (recall that for \scss the number of demands
equals the number of terminals). The formal theorem is restated below:

\thmalgbiSCSS*

An optimum solution to \biscss can have treewidth $\Omega(k)$, as the following
lemma shows. This is particularly interesting, since the results
in~\cite{DBLP:conf/icalp/FeldmannM16} show that any problem with optima of
unbounded treewidth on general input graphs is W[1]-hard. Note that no solution
with larger treewidth can exist, as by~\cite{DBLP:conf/icalp/FeldmannM16} any
optimum solution to \pname{DSN} has treewidth~$O(k)$.

\begin{lem}
There are instances of \biscss in which the optimum solution has
treewidth~$\Omega(k)$.
\end{lem}
\begin{proof}
We will describe the underlying undirected graph $\ud{G}$ of an input graph $G$ 
to \biscss. We begin with a constant degree expander graph with $k$ vertices, 
for which we subdivide each edge twice. The resulting graph is going to be 
$\ud{G}$, where each edge has unit weight. All vertices of the graph are going 
to be terminals, which means that the number of terminals is $\Theta(k)$, since 
the number of edges in a constant degree expander is linear in the number of 
vertices. Also, the treewidth of the expander graph is $\Theta(k)$, which is not 
changed by subdividing edges.

Consider any of the twice subdivided edges, i.e.\ let $P$ be a path of length
$3$ in $\ud{G}$ for which both internal vertices $u,v$ have degree $2$ in
$\ud{G}$. Let $e$ be one of the edges of $P$. If a strongly connected \scss
solution containing all vertices of the bidirected graph $G$ does not use any of
the two edges corresponding to $e$ in~$G$, then it needs to use all four of the
other edges of $G$ corresponding to the two edges of $P$ different from $e$:
this is the only way in which all other terminals can reach $u$ and $v$, and $u$
and $v$ can reach all other terminals. Note also that it is not possible for a
strongly connected solution to only use two directed edges corresponding to
edges of $P$.

We can however construct a solution $N$ in which for every edge of $P$ we use
exactly one directed edge of $G$, and this must then be optimal: the solution
$N$ initially contains one of the directed edges of $G$ corresponding to an edge
of $\ud{G}$ each. As the underlying undirected graph of $N$ would be
exactly~$\ud{G}$, its treewidth is $\Omega(k)$, as claimed. However $N$ might
not yet be strongly connected. If there are two vertices $u$ and $v$, for which
no $u\to v$ path exists in $N$, we introduce such a path as follows. An expander
cannot contain any bridge, and so $\ud{G}$ is $2$-edge-connected. Thus by
Menger's Theorem~\cite{MR2159259} there are two edge-disjoint paths $P$ and $Q$
between $u$ and $v$ in~$\ud{G}$. Consider any poly-cycle $O$ formed by edges of
the paths in $N$ corresponding to $P$ and $Q$. By \autoref{lem:cycles} we may
replace $O$ by a directed cycle without losing the connectivity between any pair
of vertices of $N$ for which a directed path already existed. Also the
underlying undirected graph of the resulting solution $N$ is still $\ud{G}$.
After replacing every poly-cycle formed by edges corresponding to those of $P$
and $Q$ in this way, there will be a $u\to v$ path in $N$. We may repeat this
procedure for any pair of vertices that does not have a path between them, until
the solution is strongly connected.
\end{proof}

We prove that \biscss is FPT via a similar decomposition to the one of
\autoref{lem:full-comp} for \bidsnP (or the \citeauthor{borchers-du} Theorem for
\pname{ST}). More concretely, we show that any optimum solution to \biscss can
be decomposed into non-overlapping (i.e., edge-disjoint) poly-trees, each of
which is a feasible solution to some demand pairs of the terminals. As a
consequence, similar to the PAS of \autoref{thm:scheme}, we can compute optimum
poly-tree solutions via \autoref{thm:alg-FMarx}, among which we can find a
solution to \biscss. Since, compared to \autoref{lem:full-comp}, here we have
the stronger property that poly-tree solutions in the decomposition for \biscss
do not overlap, we obtain an optimum solution this way. However, in contrast to
\autoref{lem:full-comp} the number of terminals in each poly-tree is not bounded
by any constant, and thus applying the algorithm of \autoref{thm:alg-FMarx}
needs FPT time instead of polynomial time. Due to this weaker property compared
to \autoref{lem:full-comp}, also no kernelization is implied by our
decomposition for \biscss, as this would require a polynomial time
algorithm.\footnote{For this reason \autoref{lem:bi-SCSS-struct} is ``merely'' a
lemma, while \autoref{lem:full-comp} is a theorem.} For the following statement
we reuse the formulation of \dsn (and thus in particular for \biscss) in terms
of pattern graphs, as introduced in \autoref{sec:scheme}.

\begin{lem}\label{lem:bi-SCSS-struct}
Let $G$ be a bidirected graph with terminal set $R$, and $N\subseteq G$ be the
cheapest strongly connected subgraph containing $R$. There exists a set of
patterns $\mc{H}$ such that
\begin{enumerate}
\item $V(H)\subseteq R$ for each $H\in\mc{H}$,
\item given any feasible solutions $N_{H}\subseteq G$ for all $H\in\mc{H}$,
the union~$\bigcup_{H\in\mc{H}} N_{H}$ of the these solutions strongly
connects $R$, and
\item there exist feasible solutions $T^*_{H}\subseteq G$ for all $H\in\mc{H}$
where each $T^*_H$ is a poly-tree and $\sum_{H\in\mc{H}}\cost(T^*_H)=\cost(N)$.
\end{enumerate}
\end{lem}

Before proving this lemma we show that it implies the claimed FPT algorithm of
\autoref{thm:alg-biSCSS}. Note that we do not know the set $\mc{H}$ of
\autoref{lem:bi-SCSS-struct} without knowing the solution $N$, which we wish to
compute. However, a simple dynamic programming approach can be used to find $N$.

\begin{proof}[Proof of \autoref{thm:alg-biSCSS}]
We present an algorithm that is very similar to the one for
\autoref{thm:scheme}.
The first step of the algorithm is to compute the optimum poly-tree solutions
to all patterns on the terminals~$R$. Note that poly-trees are exactly the
directed graphs with treewidth~$1$, so that for every pattern graph $H$ on $R$
we can set $\omega=1$ in \autoref{thm:alg-FMarx} to compute the best poly-tree
solution (if any) in $2^{O(k)}\polyn$ time. Since there are $2{k\choose
2}=k^2-k$ possible edges for any pattern on~$R$, and any subset of these may
span a pattern $H$, there are $2^{k^2-k}$ possible patterns. Thus up to now the
algorithm uses $2^{k^2+O(k)}\polyn$ time.

The next step is to use a dynamic program to compute a solution strongly
connecting all of~$R$ by putting together these poly-tree solutions. More
concretely, let $T_1,\ldots,T_p$ be all the poly-tree solutions computed in the
first step (given in any arbitrary order). For any subset $\mc{T}$ of these
poly-trees, in the following we denote by
$\cost(\mc{T}):=\sum_{T\in\mc{T}}\cost(T)$ their total cost and by
$\bigcup\mc{T}:=\bigcup_{T\in\mc{T}}T$ their union. For $1\leq i\leq p$ and any
pattern graph~$H$, we define
\[\label{eq:sigma2}
\sigma(H,i)=\min\left\{\cost(\mc{T})~\Big\vert~\mc{T}\subseteq\{T_1,
\ldots,T_i\} \text{ and } \bigcup\mc{T} \text{ feasible for } H\right\}
\]
to be the minimum total cost of a subset of the first $i$ poly-trees
$T_1,\ldots,T_i$ that forms a feasible solution to~$H$. Note that the total cost
of a set $\mc{T}$ counts edges appearing in more than one poly-tree of $\mc{T}$
several times. If no feasible solution to $H$ can be obtained from any subset of
$T_1,\ldots,T_i$, then we define $\sigma(H,i)$ to be~$\infty$.

Let $\mc{H}$ be the set of patterns given by \autoref{lem:bi-SCSS-struct} for
the optimum \biscss solution $N$, and let $T^*_H$ be the poly-tree solution to
each $H\in\mc{H}$ given by the lemma. The existence of $T^*_H$ in particular
implies that for each $H\in\mc{H}$ there is a feasible poly-tree solution $N_H$
among  $T_1,\ldots,T_p$, and by \autoref{lem:bi-SCSS-struct} their union
$\bigcup_{H\in\mc{H}} N_H$ strongly connects $R$. Thus if~$H^*$ is any strongly
connected pattern graph on $R$ (e.g., a directed cycle on $R$), then
\[
\sigma(H^*,p)\leq \sum_{H\in\mc{H}} \cost(N_H) \leq \sum_{H\in\mc{H}}
\cost(T^*_H) =\cost(N),
\]
where the second inequality follows since each computed poly-tree $T_i$ (and
thus each $N_H$ where $H\in\mc{H}$) is an optimum poly-tree solution according
to \autoref{thm:alg-FMarx}. We conclude that $\sigma(H^*,p)$ is the value of the
optimum \biscss solution we wish to compute.

To recursively compute $\sigma(H,i)$ for any pattern graph $H$ on $R$ and any
$1\leq i\leq p$, we keep track of the subset
$\mc{T}^i_H\subseteq\{T_1,\ldots,T_i\}$ of poly-trees that obtain the cost
stored by the following dynamic program in~$\sigma(H,i)$. For $i=1$ we just
check whether~$T_1$ is a feasible solution to $H$. If so, we set
$\sigma(H,1)=\cost(T_1)$ and~$\mc{T}^1_H=\{T_1\}$, while otherwise we set
$\sigma(H,1)=\infty$ and $\mc{T}^1_H=\emptyset$. This obviously computes
$\sigma(H,1)$ correctly. To compute $\sigma(H,i)$ for any $i\geq 2$, we check
for each pattern graph $H'$ whether $(\bigcup\mc{T}^{i-1}_{H'}) \cup T_i$ is a
feasible solution to~$H$. Among all such solutions and the
graph~$\bigcup\mc{T}^{i-1}_H$ we store the cost of the cheapest option. More
formally, we claim that for $i\geq 2$
\begin{multline}
\label{eq:dp2}
 \sigma(H,i)=\min\Big\{\sigma(H,i-1), \sigma(H',i-1)+\cost(T_i)~\Big\vert\\
H'\text{ is a pattern with } \big(\bigcup\mc{T}^{i-1}_{H'}\big) \cup T_i\text{ 
feasible for } H\Big\}.
\end{multline}
If the right-hand side of~\eqref{eq:dp2} is some finite value, we set
$\mc{T}^i_H$ to the subset obtaining the minimum (i.e., either $\mc{T}^{i-1}_H$
or~$\mc{T}^{i-1}_{H'}\cup\{T_i\}$ for some $H'$). Otherwise, we let
$\mc{T}^i_H=\emptyset$.

To show that the recursion given by~\eqref{eq:dp2} is correct, fix $H$ and
$i\geq 2$, and let $\mc{T}^*\subseteq\{T_1,\ldots,T_i\}$ be the subset of
poly-trees defining $\sigma(H,i)$, i.e., $\mc{T}^*$ minimizes the right-hand
side of~\eqref{eq:sigma2}. We need to show that
$\cost(\mc{T}^i_H)=\cost(\mc{T}^*)$. First note that by~\eqref{eq:dp2},
$\bigcup\mc{T}^i_H$ is a feasible solution to $H$ and is the union of some
subset of $T_1,\ldots,T_i$, and so $\cost(\mc{T}^i_H)\geq\cost(\mc{T}^*)$ by
definition of $\mc{T}^*$. In case $T_i\notin\mc{T}^*$, we have
$\cost(\mc{T}^{i-1}_H)=\cost(\mc{T}^*)$ by induction, and so
$\cost(\mc{T}^i_H)\leq\cost(\mc{T}^*)$, since
$\sigma(H,i-1)=\cost(\mc{T}^{i-1}_H)$ is considered as one of the values over
which~\eqref{eq:dp2} minimizes. In the other case when~$T_i\in\mc{T}^*$,
consider the graph $\bigcup(\mc{T}^*\setminus\{T_i\})$ obtained by taking the
union of all poly-trees in $\mc{T}^*$ except~$T_i$ (note that it may still
contain edges of $T_i$). Now let $H'$ be the pattern graph on~$R$, which
contains an edge $st$ if and only if $\bigcup(\mc{T}^*\setminus\{T_i\})$
contains an $s\to t$ path. By induction we have
$\cost(\mc{T}^{i-1}_{H'})\leq\cost(\mc{T}^*\setminus\{T_i\})$, and adding
$\cost(T_i)$ to both sides of this inequality we get
$\cost(\mc{T}^{i-1}_{H'})+\cost(T_i)\leq\cost(\mc{T}^*)$, since~$\mc{T}^*$
contains $T_i$. Moreover, $(\bigcup\mc{T}^{i-1}_{H'})\cup T_i$ is a feasible
solution to $H$, since $\bigcup\mc{T}^{i-1}_{H'}$ is a feasible solution to $H'$
and adding $T_i$ we obtain an $s\to t$ path between terminals $s,t\in R$ if and
only if $\bigcup\mc{T}^*$ contains some $s\to t$ path as well. Hence
$\cost(\mc{T}^i_H)\leq\cost(\mc{T}^{i-1}_{H'})+\cost(T_i)$, as the latter term
is equal to $\sigma(H',i-1)+\cost(T_i)$ and is considered as one of the values
over which~\eqref{eq:dp2} minimizes. In conclusion, also if $T_i\in\mc{T}^*$ we
have $\cost(\mc{T}^i_H)\leq\cost(\mc{T}^*)$ and so
$\cost(\mc{T}^i_H)=\cost(\mc{T}^*)$. Thus the recursion given in~\eqref{eq:dp2}
correctly computes the value of $\sigma(H,i)$ according to its
definition in~\eqref{eq:sigma2}.

To bound the runtime of the dynamic program, recall that there are $2^{k^2-k}$
possible pattern graphs $H$ on $R$, and the first step of the algorithm computes
at most one poly-tree solution to each pattern $H$, i.e.,~$p\leq 2^{k^2-k}$.
Thus the size of the table given by all entries $\sigma(H,i)$ (with $1\leq i\leq
p$) is at most $2\cdot 2^{k^2-k}$. To compute one entry of the table
via~\eqref{eq:dp2}, we need to consider every pattern $H'$ for each of which we
perform a feasibility check, which can be done in polynomial time. Thus the
runtime for each of the at most $2\cdot 2^{k^2-k}$ entries is $2^{k^2-k}\polyn$, 
and the total runtime of the algorithm (including the first step) is bounded by 
$4^{k^2+O(k)}\polyn$.
\end{proof}

To complete this section we now prove \autoref{lem:bi-SCSS-struct} and show how
to decompose any strongly connected solution in a bidirected graph.

\begin{proof}[Proof of \autoref{lem:bi-SCSS-struct}]
We will assume w.l.o.g.\ that in the cheapest solution $N\subseteq G$ each
terminal has only $1$ neighbour, and each Steiner vertex has exactly $3$
neighbours. We may assume this according to the transformation given in
\autoref{app:degrees}, just as for our earlier proof of \autoref{lem:full-comp}.
Furthermore, let $G_N$ again be the graph spanned by the edge set~$\{uv,vu\mid 
uv\in E(N)\}$. It is not hard to see that proving \autoref{lem:bi-SCSS-struct} 
for the obtained optimum solution $N$ in $G_N$ implies the same result for the 
original optimum solution in $G$, by reversing all transformations given in 
\autoref{app:degrees}.

We will first reduce the claim to solutions that have a $2$-connected underlying
undirected graph. In particular, consider a maximal $2$-connected component
$\ud{C}$ of $\ud{N}$, and the set of articulation points $W$ of $\ud{N}$
contained in $\ud{C}$, i.e., $w\in W$ if and only if $w\in V(C)$ and $w$ is
adjacent to some vertex of $\ud{N}$ that is not in $\ud{C}$. We now claim that
the directed subgraph $C$ of $N$ corresponding to $\ud{C}$ is an optimum
strongly connected solution for the terminal set given by $W$. First off, note
that $C$ cannot contain any terminals from $R$, as we assume that every terminal
in $R$ has only one neighbour, while $\ud{C}$ is $2$-connected. Since $\ud{C}$
is a maximal $2$-connected component of $\ud{N}$, no path leaving $C$ can
return to $C$. So any $u\to v$ path connecting a pair of vertices $u,v\in W$
must be entirely contained in~$C$. This means that $C$ strongly connects $W$,
since $N$ is strongly connected. If $C$ was not an optimum strongly connected
solution for $W$, we could replace it by a cheaper one in~$N$. This would result
in a feasible solution to $R$ but with smaller cost than $N$, which would
contradict the optimality of~$N$.

Since $C$ is an optimum strongly connected solution for $W$, we are able to
prove the next claim, which essentially follows from our main observation on
solutions in bidirected graphs given by \autoref{lem:cycle-split}.

\begin{claim}\label{clm:cycles}
Every cycle of $\ud{C}$ contains at least two vertices of $W$.
\end{claim}
\begin{proof}
Assume $\ud{C}$ contains a cycle $O$ with at most one vertex from $W$. As
$\ud{C}$ is $2$-connected and $C$ is a minimum cost solution for $W$, there are
at least two vertices in $W$, and at least one of these does not lie on $O$. In
particular, on the cycle $O$ there must be vertices that have degree more than
$2$ in $\ud{C}$ where paths lead to vertices of $\ud{C}$ not on $O$. As $\ud{C}$
is $2$-connected, by Menger's Theorem~\cite{MR2159259} any such path leading
away from $O$ from a vertex $u\in V(O)$ must eventually lead back to some vertex
$v\in V(O)$. We assume that every vertex of $\ud{N}$ has degree at most $3$, and
so $u\neq v$. This means that there exists a non-empty path $P\subseteq O$
between $u$ and $v$ along $O$ that contains no vertex of $W$ as an internal
vertex, since $O$ contains at most one vertex from $W$.

We will now fix such a pair of vertices $u,v\in V(O)$ of degree $3$ in $\ud{C}$,
such that there is a $u$--$v$ path $Q\subseteq\ud{C}$, which contains no edge of
$O$. We choose the pair $u,v$ under the minimality condition that the $u$--$v$
path $P\subseteq O$ not containing an internal vertex from $W$ is of minimum
length. That is, there is no pair $u',v'$ of vertices on $P$, so that at least
one of $u'$ and $v'$ is an internal vertex of~$P$, and so that there is a
$u'$--$v'$ path in~$\ud{C}$, which contains no edge of $O$: otherwise the
$u'$--$v'$ subpath of $P$ would be a shorter path not containing an internal
vertex from $W$ than $P$ for the pair $u,v$. In particular, this means that any
path from an internal vertex of $P$ that leads away from $O$ must lead back to a
vertex of $O$ that does not lie on~$P$.

Assume that $P$ has internal vertices of degree $3$ in $\ud{C}$, and let $w$ be
the closest one to $u$ on $P$. That is, there is a $w$--$w'$ path $Q'$ not
containing any edge of $O$, such that $w'$ lies on $O$ but not on $P$, by our
choice of $u$ and $v$. Furthermore, the $w$--$u$ subpath $P'$ of $P$ has no
internal vertex of degree $3$ in $\ud{C}$ by our choice of $w$, but it has
length at least $1$, as $w$ is an internal vertex of $P$. Now consider the cycle
$O'$ formed by the $u$--$v$ path $Q$, the $v$--$w$ subpath of $P$ (with edges
not on~$P'$), the $w$--$w'$ path $Q'$, and the $w'$--$u$ path on $O$ not
containing~$v$. As $P'$ does not lie on $O'$ but connects the vertices $w$ and
$v$ of $O'$, by \autoref{lem:cycle-split} the path $P'$ cannot be a single edge,
since $C$ is an optimum solution for $W$. Thus removing $O'$ from $\ud{C}$
results in a connected component that consists of the subpath of $P'$ connecting
the non-empty set of internal vertices of~$P'$. This is because $w$ is the
closest internal vertex of $P$ with degree $3$ to $u$, so that each internal
vertex of $P'$ has degree $2$ in~$\ud{C}$. However none of the vertices of this
connected component is from $W$, as $P$, and therefore $P'$, has no internal
vertex from $W$. This contradicts \autoref{lem:cycle-split}, as $C$ is an
optimum \scss solution for the set $W$.

Thus we are left with the case when all internal vertices of $P$ have degree $2$
in $\ud{C}$. In this case we consider the cycle $O'$ formed by the $u$--$v$ path
$Q$ and the $v$--$u$ path $Q'\subseteq O$ containing no edge of~$P$. Again, note
that $P$ connects the two vertices $u$ and $v$ of the cycle $O'$ but $P$ does
not lie on~$O'$. Thus, as before, $P$ cannot be a single edge by
\autoref{lem:cycle-split}, so that the non-empty set of internal vertices of $P$
induce a connected component after removing $O'$ from $\ud{C}$. This connected
component contains no vertex from $W$, which once more contradicts
\autoref{lem:cycle-split} since $C$ is optimum.
\cqed
\end{proof}

This claim implies that we can partition the edges of $\ud{C}$ into sets
spanning edge-disjoint trees with leaves from $W$ and internal vertices not in
$W$, as follows. Take any edge $e$ of~$\ud{C}$ and consider the set of paths
$\mc{P}$ in $\ud{C}$ that contain $e$, have two vertices of $W$ as endpoints,
and only vertices not in $W$ as internal vertices. Assume the paths in $\mc{P}$
together span a graph containing a cycle. By \autoref{clm:cycles} there is a
vertex $w\in W$ on this cycle. This vertex $w$ is the endpoint of two paths
in~$\mc{P}$, each of which contains a different edge incident to $w$ on the
cycle. Since both these paths also contain~$e$, they span a cycle~$O$ containing
$w$ ($O$ may be different from the former cycle). As none of the internal
vertices of the two paths is from $W$ while the endpoints are, the cycle $O$
also contains no vertex from $W$ apart from $w$ (otherwise the paths could not
share $e$). Hence we found a cycle $O$ with only one vertex from $W$, which
contradicts \autoref{clm:cycles}, and so the set $\mc{P}$ spans a tree. As we
can find such a set of paths for every edge of~$\ud{C}$, we can also find the
desired edge partition for which each set spans a tree with leaves from $W$ and
internal vertices not from~$W$. Let $\mc{T}_C$ be the set containing the graphs
in $C$ of treewidth $1$ corresponding to these trees in~$\ud{C}$.

We now extend the graphs of $\mc{T}_C$ of all $2$-connected components $\ud{C}$
into edge-disjoint poly-trees of~$N$, for which the leaves are terminals in $R$,
as follows. Each graph $T\in\mc{T}_C$ is a poly-tree of~$C$, since every edge of
$C$ lies on a cycle, for which by \autoref{lem:cycle-split} no reverse edge
exists in~$N$. However a leaf $w$ of $T$ is not a terminal from $R$ but an
articulation point of $\ud{N}$, i.e\ a vertex of the corresponding set $W$. The
$2$-connected components of $\ud{N}$ are connected through these articulation
points by trees, for which the leaves are terminals or articulation points of
$\ud{N}$. As $N$ is strongly connected, such a tree corresponds to a bidirected
graph $T$ of treewidth $1$ in $N$. This means that fixing one of the leaves $r$
of $T$, the graph $T$ is the edge-disjoint union of an in- and an
out-arborescence on the same vertex set both with root~$r$. We denote by
$\mc{A}$ the set of edge-disjoint in- and out-arborescences connecting the
earlier-defined components $C$ of $N$ for which $\ud{C}$ is $2$-connected. Note
that $N$ is the disjoint union of all poly-trees in the sets $\mc{T}_C$ and the
arborescences in $\mc{A}$.

Since we assume that every vertex of $N$ has at most $3$ neighbours and $\ud{C}$
is $2$-connected, an articulation point $w$ of $\ud{N}$ in $\ud{C}$ has two
neighbours in $C$ and one neighbour outside of $C$. Thus $w$ is either the root
or a leaf of the two arborescences of $\mc{A}$ containing $w$, and it is a leaf
of two edge-disjoint poly-trees of~$\mc{T}_C$. In particular there are exactly
four edges incident to $w$. One of the arborescences $A\in\mc{A}$ has an edge
for which $w$ is the tail, while the other $A'\in\mc{A}$ has an edge for which
$w$ is the head. This means that $w$ must be the head of an edge of a poly-tree
$T\in\mc{T}_C$, and the tail of an edge of a poly-tree $T'\in\mc{T}_C$. Taking
the union of $T$ and $A$, and also the union of $T'$ and $A'$, results in two
edge-disjoint poly-trees in each of which every directed path of maximal length
has two leaves of the resulting poly-tree as endpoints. These endpoints are
either terminals or articulation points of $\ud{N}$ different from $w$. We can
repeat this procedure at every articulation point of $\ud{N}$ to form two new
edge-disjoint poly-trees, each from the union of two smaller poly-trees and/or
arborescences. This will result in larger and larger poly-trees, until we obtain
a partition of the edges of $N$ into sets, each of which spans a poly-tree in
which every maximal length directed path connects two terminals of $R$ that are
leaves of the poly-tree. Let $\mc{T}_N$ denote the set of all these poly-trees.

For each poly-tree $T\in\mc{T}_N$, we introduce a pattern graph $H$ to $\mc{H}$
having the subset of $R$ contained in~$T$ as its vertex set, and having an edge
$st$ whenever $T$ contains an $s\to t$ path. The solution $T^*_H$ to $H$ is
exactly the poly-tree $T$. Note that each pattern $H$ has only terminals of $R$ 
as vertices, and since the solutions $T^*_H\subseteq N$ are edge-disjoint we 
get $\sum_{H\in\mc{H}}\cost(T^*_H)=\cost(N)$. It remains to show that any union 
of feasible solutions $N_H$ for all $H\in\mc{H}$ strongly connects $R$. 

For this it suffices to argue that the union~$\bigcup\mc{H}$ of all pattern 
graphs is strongly connected. We prove this by induction on the above procedure 
constructing the poly-trees in $\mc{T}_N$ from the poly-trees in $\mc{T}_C$ and 
arborescences in~$\mc{A}$. Initially, consider a pattern graph $H'$ encoding the 
connectivity given by $\mc{T}_C$ and~$\mc{A}$: the vertex set of $H'$ consists 
of $R$ and all articulation points of $\ud{N}$ in 2-connected 
components~$\ud{C}$, and~$H'$ contains an edge $st$ if and only if there is some 
poly-tree $T\in\mc{T}_C\cup\mc{A}$ with an $s\to t$ path. Note that if $W$ is 
the set containing all leaves and the root of an arborescence $T\in\mc{A}$, then 
the induced pattern graph~$H'[W]$ is strongly connected: $\mc{A}$ contains an 
in- and an out-arborescence for this set $W$ with the same root. At the same 
time, if $W$ denotes the set of articulation points of $\ud{N}$ in a 2-connected 
component~$\ud{C}$, then also $H'[W]$ is strongly connected: any two vertices 
$u,v\in W$ lie on a directed cycle $O$ of $C$, for which the (at least two) 
vertices of $O\cap W$ are strongly connected by paths, each of which lies in 
some poly-tree of~$\mc{T}_C$. Therefore the whole pattern graph $H'$ is 
strongly connected, as each articulation point of $\ud{N}$ in 2-connected 
components $\ud{C}$ is a leaf or the root of some arborescence in $\mc{A}$.

Now consider any step of the above procedure in which we form the union of a
pair $A,T$ and a pair~$A',T'$ of poly-trees intersecting at some articulation 
point $w$ of $\ud{N}$ in a 2-connected component $\ud{C}$. In this step we also 
modify the pattern graph $H'$ by short-cutting the vertex $w$, i.e., for any 
edges $sw$ given by $s\to w$ paths of~$T$ and edges $wt$ given by $w\to t$ 
paths of $A$ we add an edge $st$ to $H'$. At the same time we also add an edge 
$st$ to $H'$ for any edges $sw$ given by $s\to w$ paths of $A'$ and edges $wt$ 
given by $w\to t$ paths of $T'$. Now we may remove the vertex $w$ and all its 
incident edges from $H'$, and by induction the resulting pattern graph still 
strongly connects all remaining vertices. At the same time, the new pattern 
graph encodes the connectivity between all leaves of the poly-trees after 
including $A\cup T$ and $A'\cup T'$ and removing $A$, $A'$, $T$, and $T'$. At 
the end of this procedure, we are left with the poly-trees in~$\mc{T}_N$ 
connecting only terminals of $R$, and a strongly connected pattern graph $H'$ 
encoding the connectivity of these poly-trees. In particular, the union 
$\bigcup\mc{H}$ of pattern graphs is exactly $H'$, which concludes the proof.
\end{proof}

\autoref{thm:lb-biSCSS} shows that \biscss is NP-hard, and even has a
$2^{o(k)}\cdot n^{O(1)}$ runtime lower bound under ETH. Hence, to the best of our
knowledge, the class of bidirected graphs is the first example where \scss
remains NP-hard but turns out to be FPT parameterized by the number of
terminals~$k$.

\section{Runtime lower bounds}
\label{sec:lb}

This section is devoted to proving runtime lower bounds. First, in
\autoref{sec:uniqueness-gadget} we describe a general gadget which is used
in both \autoref{thm:lb-scheme-biDSN} and \autoref{thm:lb-biDSN}. Then \autoref{sec:bidsn-planar-w[1]} and \autoref{sec:bidsn-w[1]} contain the
proof of W[1]-hardness of \bidsnP and \bidsn respectively. Finally,
\autoref{sec:np-hard} contains the proof of NP-hardness and the
$2^{o(k)}\polyn$ runtime lower bound for \altbiscssP.

\subsection{Constructing a ``uniqueness'' gadget}
\label{sec:uniqueness-gadget}

For every integer $n$ we define the following gadget $U_n$ which contains
$4n+4$ vertices (see \autoref{fig:uniqueness}). All edges will have
the same weight~$M$, which we will fix later during the reductions. The gadget
$U_n$ is constructed as follows (we first construct an undirected graph, and
then bidirect each edge):

        \begin{itemize}
        \item Introduce two source vertices $s_{1}, s_{2}$,
two target vertices $t_{1}, t_{2}$, and for each $i\in [n]$ the
four vertices $0_i,1_i,2_i,3_i$.
        \item $U_n$ has a path of three edges corresponding to each $i\in [n]$.
                \begin{itemize}
                \item Let $i\in [n]$. Then we denote the path in $U_{n}$
corresponding to $i$ by $P_{U_{n}}(i):= 0_{i}- 1_{i}-
2_{i}- 3_{i}$.
                \item Each of these edges is called a \emph{base} edge and has
weight $M$
                 \end{itemize}

        \item Finally we add the following edges:
                \begin{itemize}
                \item for each $i\in [n]$, the edges $s_{1}- 1_{i}$ and $t_{1}-1_{i}$
                \item for each $i\in [n]$, the edges $s_{2}- 2_{i}$ and $t_{2}- 2_{i}$
                \item Each of these edges is called a \emph{connector} edge and has weight $M$.
                \end{itemize}

        \end{itemize}

\begin{figure}

\centering
\begin{tikzpicture}[scale=0.3]

\foreach \j in {0,1,2,3,4}
{
\begin{scope}[shift={(0,5*\j)}]

\foreach \i in {0,1,2,3}
{
\draw [black] plot [only marks, mark size=8, mark=*] coordinates {(10*\i,0)};
}

\foreach \i in {0,10,20}
{
\path (\i, 0) node(a) {} (10+\i,0) node(b) {};
        \draw[thick,black] (a) -- (b);
}

\end{scope}
}

\draw [black] plot [only marks, mark size=8, mark=*] coordinates {(10,25)}
node[label={[xshift=-6mm,yshift=-2mm] $s_1$}] {} ;

\draw [black] plot [only marks, mark size=8, mark=*] coordinates {(20,25)}
node[label={[xshift=-6mm,yshift=-2mm] $s_2$}] {} ;

\draw [black] plot [only marks, mark size=8, mark=*] coordinates {(10,-5)}
node[label={[xshift=6mm,yshift=-7mm] $t_1$}] {} ;

\draw [black] plot [only marks, mark size=8, mark=*] coordinates {(20,-5)}
node[label={[xshift=6mm,yshift=-7mm] $t_2$}] {} ;

\foreach \j in {0,1}
{
\begin{scope}[shift={(10*\j,0)}]

\path (10,25) node(a) {} (10,20) node(b) {};
 \draw [thick,dotted] (a) to (b);

\path (10,25) node(a) {} (10,15) node(b) {};
 \draw [thick,dotted] (a.north) to [out=60,in=60] (b.north);

\path (10,25) node(a) {} (10,10) node(b) {};
 \draw [thick,dotted] (a.north) to [out=60,in=60] (b.north);

\path (10,25) node(a) {} (10,5) node(b) {};
 \draw [thick,dotted] (a.north) to [out=60,in=60] (b.north);

\path (10,25) node(a) {} (10,0) node(b) {};
 \draw [thick,dotted] (a.north) to [out=60,in=60] (b.north);

\end{scope}
}

\foreach \j in {0,1}
{
\begin{scope}[shift={(10*\j,0)}]

\path (10,-5) node(a) {} (10,20) node(b) {};
  \draw [thick,dotted] (a.north) to [out=240,in=240] (b.south);

\path (10,-5) node(a) {} (10,15) node(b) {};
 \draw [thick,dotted] (a.north) to [out=240,in=240] (b.south);

\path (10,-5) node(a) {} (10,10) node(b) {};
 \draw [thick,dotted] (a.north) to [out=240,in=240] (b.south);

\path (10,-5) node(a) {} (10,5) node(b) {};
 \draw [thick,dotted] (a.north) to [out=240,in=240] (b.south);

\path (10,-5) node(a) {} (10,0) node(b) {};
  \draw [thick,dotted] (a) to (b);

\end{scope}
}

\draw [black] plot [only marks, mark size=0, mark=*] coordinates {(0,20)}
node[label={[xshift=-2mm,yshift=-1mm] $0_1$}] {} ;

\draw [black] plot [only marks, mark size=0, mark=*] coordinates {(10,20)}
node[label={[xshift=-2mm,yshift=-1mm] $1_1$}] {} ;

\draw [black] plot [only marks, mark size=0, mark=*] coordinates {(20,20)}
node[label={[xshift=-2mm,yshift=-1mm] $2_1$}] {} ;

\draw [black] plot [only marks, mark size=0, mark=*] coordinates {(30,20)}
node[label={[xshift=-2mm,yshift=-1mm] $3_1$}] {} ;

\draw [black] plot [only marks, mark size=0, mark=*] coordinates {(0,10)}
node[label={[xshift=-2mm,yshift=-1mm] $0_i$}] {} ;

\draw [black] plot [only marks, mark size=0, mark=*] coordinates {(10,10)}
node[label={[xshift=0mm,yshift=-1mm] $1_i$}] {} ;

\draw [black] plot [only marks, mark size=0, mark=*] coordinates {(20,10)}
node[label={[xshift=0mm,yshift=-1mm] $2_i$}] {} ;

\draw [black] plot [only marks, mark size=0, mark=*] coordinates {(30,10)}
node[label={[xshift=-2mm,yshift=-1mm] $3_i$}] {} ;

\draw [black] plot [only marks, mark size=0, mark=*] coordinates {(0,0)}
node[label={[xshift=-2mm,yshift=-1mm] $0_n$}] {} ;

\draw [black] plot [only marks, mark size=0, mark=*] coordinates {(10,0)}
node[label={[xshift=0mm,yshift=-1mm] $1_n$}] {} ;

\draw [black] plot [only marks, mark size=0, mark=*] coordinates {(20,0)}
node[label={[xshift=0mm,yshift=-1mm] $2_n$}] {} ;

\draw [black] plot [only marks, mark size=0, mark=*] coordinates {(30,0)}
node[label={[xshift=-2mm,yshift=-1mm] $3_n$}] {} ;

\end{tikzpicture}

\caption{The construction of the uniqueness gadget $U_n$. Note that the
gadget has $4n+4$ vertices. Each \emph{base} edge is denoted by a filled edge
and each connector edge is denoted by a dotted edge in the figure.
 \label{fig:uniqueness}}
\end{figure}

After bidirecting all above undirected edges in the gadgets, we give the
following definitions for the directed graph $U_n$.

\begin{dfn}
The set of \emph{boundary} vertices of $U_n$ is $\bigcup_{i=1}^{n}
\big\{0_i, 3_i\big\} $. For each $R\in \{0,1,2,3\}$ the set of $R$-vertices of $U_{n}$ is $\big\{ R_i : 1\leq i\leq n \big\}$.
\label{defn-boundary}
\end{dfn}

\begin{dfn}
A set of edges $E'$ of $U_n$ satisfies the \emph{in-out} property if each of
the following four conditions is satisfied
\begin{itemize}
\item $s_1$ can reach some boundary vertex
\item $s_2$ can reach some boundary vertex
\item $t_1$ can be reached from some boundary vertex
\item $t_2$ can be reached from some boundary vertex
\end{itemize}
\label{defn-in-out}
\end{dfn}

\begin{dfn}
A set of edges $E'$ of $U_n$ is \emph{represented} by $i\in [n]$ and
\emph{right-oriented} if
\begin{itemize}
\item the connector edges in $E'$ are $(s_1, 1_{i}),
(s_2, 2_{i}), (1_{i},t_1)$, and
$(2_{i},t_2)$, and
\item base edges in $E'$ are $(0_i, 1_i), (1_i, 2_i)$ and $(2_i, 3_i)$ which 
form the directed $0_i \to 3_i$ path denoted by~$P^{\text{right}}_{U_n}(i)$.
\end{itemize}
A set of edges $E''$ of $U_n$ is \emph{represented} by $i\in [n]$ and
~\emph{left-oriented} if
\begin{itemize}
\item the connector edges in $E''$ are $(s_1, 1_{i}),
(s_2, 2_{i}), (1_{i},t_1)$, and
$(2_{i},t_2)$, and
\item base edges in $E''$ are $(3_i, 2_i), (2_i, 1_i)$ and $(1_i, 0_i)$ which 
form the directed $3_i \to 0_i$ path denoted by~$P^{\text{left}}_{U_n}(i)$.
\end{itemize}
\label{defn-representation}
\end{dfn}

We now show a lower bound on the weight of edges we need to pick from $U_n$ to
satisfy the in-out property.

\begin{lem}
\label{lem:macro-uniqueness-gadget}
Let $E'$ be a set of edges of $U_n$ which satisfies the in-out property.
Then we have that either\\
(i) the weight of $E'$ is at least $8M$, or\\
(ii) the weight of $E'$ is exactly $7M$ and there is an integer $i\in [n]$
such that $E'$ is represented by~$i$ and is either left-oriented or
right-oriented.
\end{lem}
\begin{proof}
We clearly need at least four connector edges in $E'$:
\begin{itemize}
\item one outgoing edge from $s_1$ so that it can reach some boundary
vertex,
\item one outgoing edge from $s_2$ so that it can reach some boundary
vertex,
\item one incoming edge into $t_1$ so that it can be reached from some
boundary vertex, and
\item one incoming edge into $t_2$ so that it can be reached from some
boundary vertex.
\end{itemize}
This incurs a cost of $4M$ in $E'$.
We now see how many base edges we must have in $E'$. We define the following:
\begin{itemize}
\item \underline{``0-1'' edges}: this is the set of edges $\big\{
(0_{i}, 1_i)\ :\ 1\leq i\leq n\big\} \cup \big\{
(1_{i}, 0_i)\ :\ 1\leq i\leq n\big\}$
\item \underline{``1-2'' edges}: this is the set of edges $\big\{
(1_{i}, 2_i)\ :\ 1\leq i\leq n\big\} \cup \big\{
(2_{i}, 1_i)\ :\ 1\leq i\leq n\big\}$
\item \underline{``2-3'' edges}: this is the set of edges $\big\{
(2_{i}, 3_i)\ :\ 1\leq i\leq n\big\} \cup \big\{
(3_{i}, 2_i)\ :\ 1\leq i\leq n\big\}$
\end{itemize}
In each of the following four cases, we show that weight of $E'$ is at least $8M$ (note that we have already shown that $E'$ must contain at least four connector edges, and hence to show the lower bound of $8M$ on weight of $E'$ we just need to show that it contains at least four base edges):
\begin{enumerate}
\item \underline{$E'$ has no ``0-1'' edges}: This implies that $E'$ has at
least
4 base edges from $U_n$: two rightward edges (one ``1-2'' and one ``2-3'') so
that $s_1$ can reach some boundary vertex, and two leftward edges (one
``1-2'' and one ``2-3'') so that $t_1$ can be reached from some boundary
vertex.

\item \underline{$E'$ has no ``2-3'' edges}: This implies that $E'$ has at
least 4 base edges from $U_n$: two leftward edges (one ``0-1'' and one ``1-2'')
so that $s_2$ can reach some boundary vertex, and two rightward edges
(one ``0-1'' and one ``1-2'') so that $t_2$ can be reached from some
boundary vertex.

\item \underline{$E'$ has no ``1-2'' edges}: This implies that $E'$ has at
least 4 base edges from $U_n$: a leftward ``0-1'' edge so that $s_1$
can reach some boundary vertex, a rightward ``0-1'' edge so that $t_1$
can be reached from some boundary vertex, a leftward ``2-3'' edge so that
$t_2$ can be reached from some boundary vertex and a rightward ``2-3''
edge so that $s_2$ can reach some boundary vertex.

\item \underline{$E'$ has more than one edge of at least one of ``0-1'', ``1-2'' and ``2-3'' types}:
    If $E'$ does not contain at least one edge from each of the types ``0-1'', ``1-2'' and ``2-3'', then we are done by the three previous cases. Hence, $E'$ contains at least one edge from each of the types ``0-1'', ``1-2'' and ``2-3''. Now, in the given case, if $E'$ has more than one edge of at least one of ``0-1'', ``1-2'' and ``2-3'' types then $E'$ contains at least four base edges which is what we had to prove.
\end{enumerate}

In each of the aforementioned four cases we have shown that weight of $E'$ is at least $8M$. The only case that remains to be considered is when $E'$ has \emph{exactly one} edge of each of the types ``0-1'', ``1-2'' and ``2-3''.
In this case, $E'$ has  weight exactly $7M$ and contains exactly four connector edges (one incident on each source and target
vertex) and exactly three base edges (one each from ``0-1'', ``1-2'',
and ``2-3'').
Let the four connector edges in $E'$ be given by
\begin{itemize}
\item $(s_1, 1_{\beta_1})$
\item $(s_2, 2_{\beta_2})$
\item $(1_{\beta_3}, t_1)$
\item $(2_{\beta_4}, t_2)$
\end{itemize}

Suppose that the (only) ``1-2'' edge of $E'$ is rightward and given by
$(1_{\beta}, 2_{\beta})$. We will now show that
$\beta=\beta_1=\beta_2=\beta_3=\beta_4$ and that the three base edges of $E'$
are exactly those which form the path~$P^{\text{right}}_{U_n}(\beta)$, i.e., the 
edges $(0_{\beta}, 1_{\beta}), (1_{\beta}, 2_{\beta})$ and $(2_{\beta}, 
3_{\beta})$.

\begin{itemize}
\item \underline{The unique ``0-1'' edge is rightward and given by
$(0_{\beta_3}, 1_{\beta_3})$}. Suppose the (unique)
``0-1'' edge is leftward: however this implies there is no incoming path to
$t_1$ which contradicts the fact that it can be reached from some
boundary vertex. Since the unique ``1-2'' edge is rightward, it follows that
the path in $E'$, which connects some boundary vertex to $t_1$, must use
the unique ``0-1'' rightward edge which is hence forced to be
$(0_{\beta_3}, 1_{\beta_3})$.

\item \underline{The unique ``2-3'' edge is rightward and given by
$(2_{\beta_2}, 3_{\beta_2})$}. Suppose the (unique)
``2-3'' edge is leftward: however this implies there is no outgoing path from
$s_2$ (since the unique ``1-2'' edge is rightward and the unique ``2-3''
edge is leftward), which contradicts the fact that $s_2$ can reach some
boundary vertex. Since both the unique ``1-2'' edge and the unique ``2-3'' edge
is rightward, it follows that the path in $E'$ from $s_2$ to some
boundary vertex must use the unique ``2-3'' rightward edge which is hence forced
to be $(2_{\beta_2}, 3_{\beta_2})$.
\end{itemize}

Hence, we have that the only base edges in $E'$ are given by
\begin{itemize}
\item the rightward ``0-1'' edge $(0_{\beta_3}, 1_{\beta_3})$,
\item the rightward ``1-2'' edge $(1_{\beta}, 2_{\beta})$, and
\item the rightward ``2-3'' edge $(2_{\beta_2}, 3_{\beta_2})$.
\end{itemize}

Now the existence of a path in $E'$ from boundary vertex to $t_2$
implies $\beta_3=\beta=\beta_4$. Similarly, the existence of a path in $E'$
from $s_1$ to some boundary vertex implies $\beta_1=\beta=\beta_2$.
Hence, we have that $\beta=\beta_1=\beta_2=\beta_3=\beta_4$, i.e., the three
base edges in $E'$ are exactly those which form the path
$P^{\text{right}}_{U_{n}}(\beta)$, i.e., the edges $(0_{\beta}, 1_{\beta}), (1_{\beta}, 2_{\beta})$ and $(2_{\beta}, 3_{\beta})$. If the unique ``1-2'' edge in $E'$ is leftward, then the arguments are symmetric.
\end{proof}

The following corollary follows immediately from the second part of proof of \autoref{lem:macro-uniqueness-gadget}.

\begin{crl}
For every $i\in [n]$ there is a set of edges $E^{\emph{right}}_{U_n}(i)$
\big(resp.~$E^{\emph{left}}_{U_n}(i)$\big) of cost exactly $7M$ which represents $i$,
is right-oriented (resp.~left-oriented) and satisfies the ``in-out'' property.
\label{crl:macro}
\end{crl}

\subsection{W[1]-hardness for \altbidsnP}
\label{sec:bidsn-planar-w[1]}

The goal of this section is to prove \autoref{thm:lb-scheme-biDSN}.
We reduce from the \gt problem introduced by Marx~\cite{daniel-grid-tiling}:

\begin{center}
\noindent\framebox{\begin{minipage}{5.00in}
\textbf{$k\times k$ \textsc{Grid Tiling}}\\
\emph{Input}: integers $k, n$, and a collection of $k^2$ non-empty sets $S_{i,j}\subseteq
[n]\times [n]$ where $1\leq i, j\leq k$.\\
\emph{Question}: for each $1\leq i, j\leq k$ does there exist a value
$\gamma_{i,j}\in S_{i,j}$ such that
\begin{itemize}
\item if $\gamma_{i,j}=(x,y)$ and $\gamma_{i,j+1}=(x',y')$ then $x=x'$, and
\item if $\gamma_{i,j}=(x,y)$ and $\gamma_{i+1,j}=(x',y')$ then $y=y'$.
\end{itemize}
\end{minipage}}
\end{center}

\begin{figure}[t]
\centering
\vspace{-5mm}
\includegraphics[width=6in]{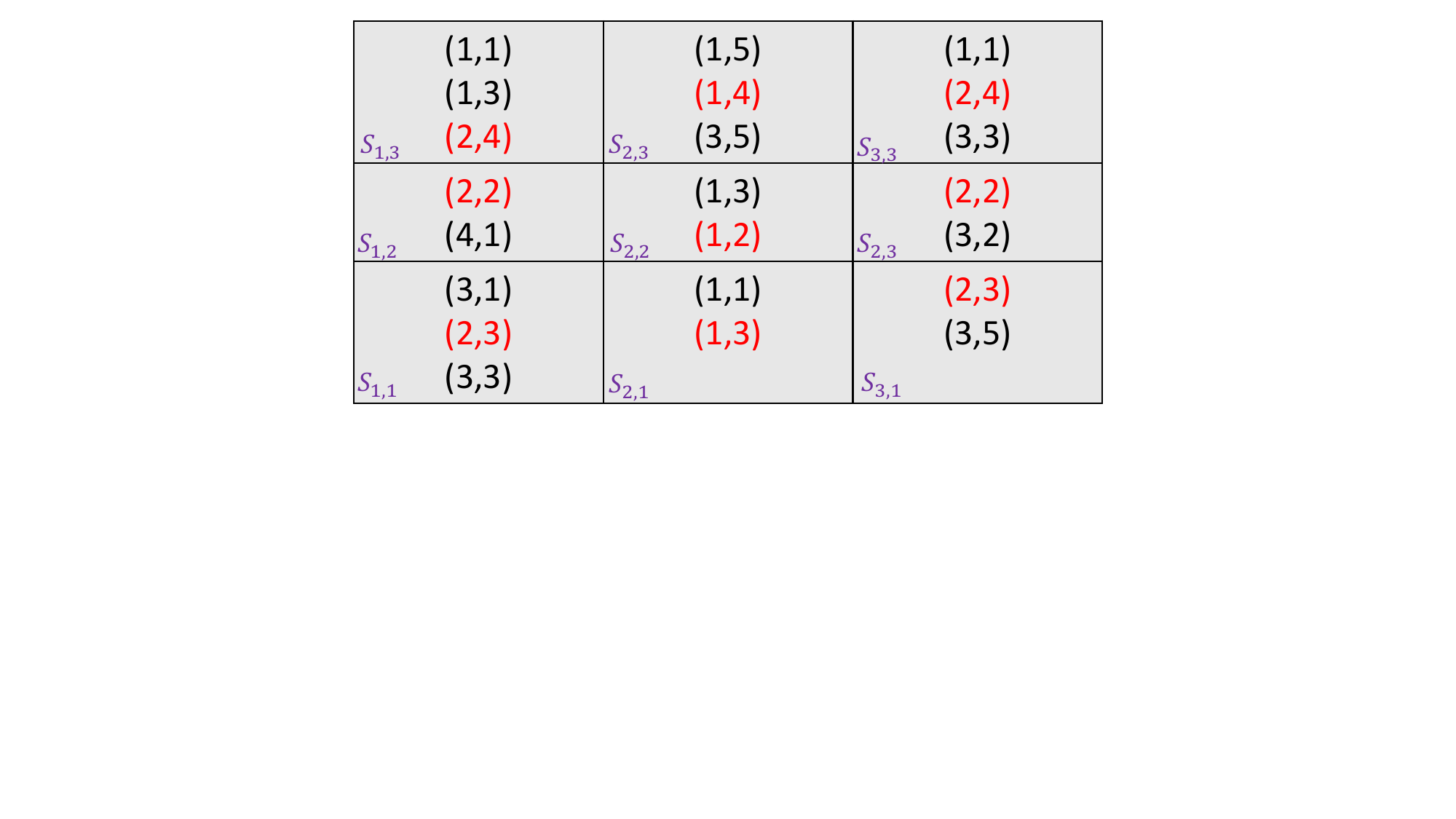}
\vspace{-35mm}
\caption{An instance of \gt with $k=3, n=5$ with a solution highlighted in red.
Note that in a solution, all entries from a row agree in the second coordinate
and all entries from a column agree in the first coordinate.}
\label{fig:gt}
\end{figure}

See \autoref{fig:gt} for example of an instance of \gt.
Under ETH~\cite{ImpagliazzoP01,ImpagliazzoPZ01}, it was shown by Chen et al.~\cite{chen-hardness} that $k$-\textsc{Clique} does not admit an algorithm running in time $f(k)\cdot n^{o(k)}$ for any computable function $f$. There is a simple reduction~\cite[Theorem 14.28]{pc-book} from $k$-\textsc{Clique} to $k\times k$ \gt implying the same runtime lower bound for the latter problem.
To prove \autoref{thm:lb-scheme-biDSN}, %
we give a reduction which transforms an instance $\big(k,n,\{S_{i,j}\}_{1\leq i,j\leq k}\big)$ of \gt into an instance $(G^*, \mathcal{D})$ of \bidsnP which has a planar optimum and the number of terminals is $|\mathcal{D}|=O(k^2)$. We design two types of gadgets: the \emph{main gadget} and the \emph{secondary
gadget}. The reduction from \gt represents each cell of the grid with a copy of
the main gadget, and each main gadget is surrounded by four secondary gadgets:
on the top, right, bottom and left. Each of these gadgets are actually
copies of the ``uniqueness gadget'' from \autoref{sec:uniqueness-gadget}
with $M=k^4$:
each secondary gadget is a copy
of $U_{n}$ and for each $1\leq i,j\leq k$ the main gadget $\M_{i,j}$
(corresponding to the set $S_{i,j}$) is a copy of~$U_{|S_{i,j}|}$. Since we have many copies of the uniqueness gadget from~\autoref{sec:uniqueness-gadget}, we need some notation to help us distinguish between similar vertices from different copies of the gadgets. This will be achieved by using the notation $M_{i,j}(v)$ to represent the vertex $v$ from $U_n$ in the the gadget $M_{i,j}$ which is a copy of the uniqueness gadget $U_{|S_{i,j}|}$: for example, the two source vertices of $M_{i,j}$ are $M_{i,j}(s_1)$ and $M_{i,j}(s_2)$. We refer to
\autoref{fig:bi-dsn-planar-big-picture} (bird's-eye view) and
\autoref{fig:bi-dsn-planar-main-and-4-surrounding} (zoomed-in view) for an
illustration of the reduction.

 \begin{figure}

 \centering
 \begin{tikzpicture}[scale=0.35]

 \foreach \j in {0,1,2}
 {
 \begin{scope}[shift={(0,12*\j)}]

 \foreach \j in {0,1,2}
 {
 \begin{scope}[shift={(12*\j,0)}]

     \draw[rectangle] (0,0) rectangle (5,5);

     \foreach \j in {0,1}
     {
     \begin{scope}[shift={(-12*\j,0)}]
     \foreach \j in {0,1,2,3}
     {
     \begin{scope}[shift={(0,\j)}]
     \foreach \i in {0,1,2,3}
    \draw [black] plot [only marks, mark size=3, mark=*] coordinates {(7+\i,1)};
     \end{scope}
     }
     \end{scope}
     }

     \foreach \j in {0,1}
     {
     \begin{scope}[shift={(-12*\j,0)}]
     \foreach \j in {0,1,2,3}
     {
     \begin{scope}[shift={(0,\j)}]
     \foreach \i in {0,1,2}
     {
     \path (7+\i, 1) node(a) {} (7+\i+1, 1) node(b) {};
         \draw[thick,black] (a) -- (b);
     }

     \end{scope}
     }
     \end{scope}
     }

     \foreach \j in {0,1}
     {
     \begin{scope}[shift={(0,-12*\j)}]
     \foreach \j in {0,1,2,3}
     {
     \begin{scope}[shift={(\j,0)}]
     \foreach \i in {0,1,2,3}
    \draw [black] plot [only marks, mark size=3, mark=*] coordinates {(1,7+\i)};
     \end{scope}
     }
     \end{scope}
     }

     \foreach \j in {0,1}
     {
     \begin{scope}[shift={(0,-12*\j)}]
     \foreach \j in {0,1,2,3}
     {
     \begin{scope}[shift={(\j,0)}]
     \foreach \i in {0,1,2}
     {
     \path (1,7+\i) node(a) {} (1,7+\i+1) node(b) {};
         \draw[thick,black] (a) -- (b);
     }
     \end{scope}
     }
     \end{scope}
     }
 \end{scope}
 }

 \end{scope}
 }

 \foreach \i in {0,1,2}
 {
 \begin{scope}[shift={(0,12*\i)}]

 \foreach \j in {-1,0,1,2}
     {
     \begin{scope}[shift={(12*\j,0)}]

 \draw[green] (6.5,-0.5) rectangle (10.5,5.5) ;

 \draw [red] plot [only marks, mark size=3, mark=*] coordinates {(8,5)}
 node[label={[xshift=-3mm,yshift=-4mm] \small{$s_1$}}] {} ;

 \draw [red] plot [only marks, mark size=3, mark=*] coordinates {(9,5)}
 node[label={[xshift=3mm,yshift=-4mm] \small{$s_2$}}] {} ;

 \draw [red] plot [only marks, mark size=3, mark=*] coordinates {(8,0)}
 node[label={[xshift=-3mm,yshift=-4mm] \small{$t_1$}}] {} ;

 \draw [red] plot [only marks, mark size=e, mark=*] coordinates {(9,0)}
 node[label={[xshift=3mm,yshift=-4mm] \small{$t_2$}}] {} ;

 \end{scope}
 }
 \end{scope}
 }

 \foreach \i in {0,1,2}
 {
 \begin{scope}[shift={(12*\i,0)}]

 \foreach \j in {0,1,2,3}
     {
     \begin{scope}[shift={(0,12*\j)}]

 \draw[blue] (-0.5,-5.5) rectangle (5.5,-1.5) ;

 \draw [red] plot [only marks, mark size=3, mark=*] coordinates {(5,-3)}
 node[label={[xshift=0mm,yshift=-1mm] \small{$s_1$}}] {} ;

 \draw [red] plot [only marks, mark size=3, mark=*] coordinates {(5,-4)}
 node[label={[xshift=0mm,yshift=-7mm] \small{$s_2$}}] {} ;

 \draw [red] plot [only marks, mark size=3, mark=*] coordinates {(0,-3)}
 node[label={[xshift=0mm,yshift=-1mm] \small{$t_1$}}] {} ;

 \draw [red] plot [only marks, mark size=e, mark=*] coordinates {(0,-4)}
 node[label={[xshift=0mm,yshift=-7mm] \small{$t_2$}}] {} ;

 \end{scope}
 }
 \end{scope}
 }

 \foreach \i in {1,2,3}
 {
 \foreach \j in {1,2,3}
 {
 \draw [black] plot [only marks, mark size=0, mark=*] coordinates
{(-9.5+12*\i,-9.5+12*\j)}
 node[label={[xshift=0mm,yshift=-4mm] $\M_{\i,\j}$}] {} ;
 }
 }

 \draw [black] plot [only marks, mark size=3, mark=*] coordinates {(-7.5,2.5)}
node[label={[xshift=0mm,yshift=0mm] $c_1$}] {};
 \draw [black] plot [only marks, mark size=3, mark=*] coordinates {(-7.5,14.5)}
node[label={[xshift=0mm,yshift=0mm] $c_2$}] {};
 \draw [black] plot [only marks, mark size=3, mark=*] coordinates {(-7.5,26.5)}
node[label={[xshift=0mm,yshift=0mm] $c_3$}] {};

 \foreach \i in {0,1,2}
 {
 \begin{scope}[shift={(0,12*\i)}]
 \foreach \i in {0,1,2,3}
 {
 \path (-7.5,2.5) node(a) {} (-5,\i+1) node(b) {};
         \draw[ultra thick,orange] (a) -- (b);
 }
 \end{scope}
 }

 \draw [black] plot [only marks, mark size=3, mark=*] coordinates {(36.5,2.5)}
node[label={[xshift=0mm,yshift=0mm] $d_1$}] {};
 \draw [black] plot [only marks, mark size=3, mark=*] coordinates {(36.5,14.5)}
node[label={[xshift=0mm,yshift=0mm] $d_2$}] {};
 \draw [black] plot [only marks, mark size=3, mark=*] coordinates {(36.5,26.5)}
node[label={[xshift=0mm,yshift=0mm] $d_3$}] {};

 \foreach \i in {0,1,2}
 {
 \begin{scope}[shift={(0,12*\i)}]
 \foreach \i in {0,1,2,3}
 {
 \path (36.5,2.5) node(a) {} (34,\i+1) node(b) {};
         \draw[ultra thick,orange] (a) -- (b);
 }
 \end{scope}
 }

 \draw [black] plot [only marks, mark size=3, mark=*] coordinates {(2.5,-7.5)}
node[label={[xshift=0mm,yshift=-7mm] $b_1$}] {};
 \draw [black] plot [only marks, mark size=3, mark=*] coordinates {(14.5,-7.5)}
node[label={[xshift=0mm,yshift=-7mm] $b_2$}] {};
 \draw [black] plot [only marks, mark size=3, mark=*] coordinates {(26.5,-7.5)}
node[label={[xshift=0mm,yshift=-7mm] $b_3$}] {};

 \foreach \i in {0,1,2}
 {
 \begin{scope}[shift={(12*\i,0)}]
 \foreach \i in {0,1,2,3}
 {
 \path (2.5,-7.5) node(a) {} (\i+1,-5) node(b) {};
         \draw[orange, ultra thick] (a) -- (b);
 }
 \end{scope}
 }

 \draw [black] plot [only marks, mark size=3, mark=*] coordinates {(2.5,36.5)}
node[label={[xshift=0mm,yshift=0mm] $a_1$}] {};
 \draw [black] plot [only marks, mark size=3, mark=*] coordinates {(14.5,36.5)}
node[label={[xshift=0mm,yshift=0mm] $a_2$}] {};
 \draw [black] plot [only marks, mark size=3, mark=*] coordinates {(26.5,36.5)}
node[label={[xshift=0mm,yshift=0mm] $a_3$}] {};

 \foreach \i in {0,1,2}
 {
 \begin{scope}[shift={(12*\i,0)}]
 \foreach \i in {0,1,2,3}
 {
 \path (2.5,36.5) node(a) {} (\i+1,34) node(b) {};
         \draw[ultra thick,orange] (a) -- (b);
 }
 \end{scope}
 }

 \foreach \j in {1,2,3}
 {
 \foreach \i in {1,2,3,4}
 {
 \draw [green] plot [only marks, mark size=0, mark=*] coordinates
{(-15.5+12*\i,-13+12*\j)} node[label={[xshift=0mm,yshift=-6mm] $\VS_{\i,\j}$}]
{};

 }
 }

 \foreach \j in {1,2,3,4}
 {
 \foreach \i in {1,2,3}
 {
 \draw [blue] plot [only marks, mark size=0, mark=*] coordinates
{(-11+12*\i,-18+12*\j)} node[label={[xshift=-12mm,yshift=-1mm] $\HS_{\i,\j}$}]
{};
 }
 }

 \end{tikzpicture}

 \caption{A bird's-eye view of the instance of $G^*$ with $k=3$ and $n=4$ (see
\autoref{fig:bi-dsn-planar-main-and-4-surrounding} for a zoomed-in view). The
connector edges within each main and secondary gadget are not shown. Similarly,
the vertices and edges within each main gadget are not shown here either.
Additionally we have some red edges between each main gadget and the four
secondary gadgets surrounding it which are omitted in this figure for clarity
(they are shown in \autoref{fig:bi-dsn-planar-main-and-4-surrounding} which
gives a more zoomed-in view).
  \label{fig:bi-dsn-planar-big-picture}}
 \end{figure}

 \begin{figure}

 \centering
 \begin{tikzpicture}[scale=0.5]

 \foreach \j in {0,1}
 {
 \begin{scope}[shift={(-21*\j,0)}]

 \foreach \j in {0,1,2,3}
 {
 \begin{scope}[shift={(0,\j)}]
 \foreach \i in {0,1,2,3}
 {
 \draw [black] plot [only marks, mark size=3, mark=*] coordinates {(12+2*\i,3)};
 }

 \foreach \i in {0,1,2}
 {
 \path (12+2*\i,3) node(a) {} (12+2*\i+2,3) node(b) {};
         \draw[ultra thick,black] (a) -- (b);
 }
 \end{scope}
 }

 \draw [red] plot [only marks, mark size=3, mark=*] coordinates {(14,2)} ;
 \draw [red] plot [only marks, mark size=3, mark=*] coordinates {(16,2)} ;
 \draw [red] plot [only marks, mark size=3, mark=*] coordinates {(14,7)} ;
 \draw [red] plot [only marks, mark size=3, mark=*] coordinates {(16,7)} ;

 \foreach \i in {0,1}
 {
 \begin{scope}[shift={(2*\i,0)}]
 \path (14,7) node(a) {} (14,6) node(b) {};
  \draw [thick,dotted] (a) to (b);

 \path (14,7) node(a) {} (14,5) node(b) {};
  \draw [thick,dotted] (a.north) to [out=60,in=60] (b.north);

 \path (14,7) node(a) {} (14,4) node(b) {};
  \draw [thick,dotted] (a.north) to [out=45,in=45] (b.north);

 \path (14,7) node(a) {} (14,3) node(b) {};
  \draw [thick,dotted] (a.north) to [out=45,in=45] (b.north);

 \path (14,2) node(a) {} (14,3) node(b) {};
  \draw [thick,dotted] (a) to (b);

 \path (14,2) node(a) {} (14,4) node(b) {};
  \draw [thick,dotted] (a.north) to [out=120,in=120] (b.south);

 \path (14,2) node(a) {} (14,5) node(b) {};
  \draw [thick,dotted] (a.north) to [out=135,in=135] (b.south);

 \path (14,2) node(a) {} (14,6) node(b) {};
  \draw [thick,dotted] (a.north) to [out=135,in=135] (b.south);
 \end{scope}
 }

 \end{scope}
 }

 \foreach \j in {0,1}
 {
 \begin{scope}[shift={(0,21*\j)}]

 \foreach \j in {0,1,2,3}
 {
 \begin{scope}[shift={(\j,0)}]
 \foreach \i in {0,1,2,3}
 {
 \draw [black] plot [only marks, mark size=3, mark=*] coordinates {(3,-3-2*\i)};
 }

 \foreach \i in {0,1,2}
 {
 \path (3,-3-2*\i) node(a) {} (3,-3-2*\i-2) node(b) {};
         \draw[ultra thick,black] (a) -- (b);
 }
 \end{scope}
 }

 \draw [red] plot [only marks, mark size=3, mark=*] coordinates {(2,-5)} ;
 \draw [red] plot [only marks, mark size=3, mark=*] coordinates {(2,-7)} ;
 \draw [red] plot [only marks, mark size=3, mark=*] coordinates {(7,-5)} ;
 \draw [red] plot [only marks, mark size=3, mark=*] coordinates {(7,-7)} ;

 \foreach \i in {0,1}
 {
 \begin{scope}[shift={(0,-2*\i)}]
 \path (7,-5) node(a) {} (6,-5) node(b) {};
  \draw [thick,dotted] (a) to (b);

 \path (7,-5) node(a) {} (5,-5) node(b) {};
  \draw [thick,dotted] (a.south) to [out=210,in=210] (b.south);

 \path (7,-5) node(a) {} (4,-5) node(b) {};
  \draw [thick,dotted] (a.south) to [out=210,in=210] (b.south);

 \path (7,-5) node(a) {} (3,-5) node(b) {};
  \draw [thick,dotted] (a.south) to [out=210,in=210] (b.south);

 \path (2,-5) node(a) {} (3,-5) node(b) {};
  \draw [thick,dotted] (a) to (b);

 \path (2,-5) node(a) {} (4,-5) node(b) {};
  \draw [thick,dotted] (a.north) to [out=135,in=135] (b.south);

 \path (2,-5) node(a) {} (5,-5) node(b) {};
  \draw [thick,dotted] (a.north) to [out=135,in=135] (b.south);

 \path (2,-5) node(a) {} (6,-5) node(b) {};
  \draw [thick,dotted] (a.north) to [out=135,in=135] (b.south);
 \end{scope}
 }

 \end{scope}
 }

 \foreach \j in {0,1,2}
 {
 \begin{scope}[shift={(0,1.5*\j)}]
 \foreach \i in {0,1,2,3}
 {
\draw [black] plot [only marks, mark size=3, mark=*] coordinates {(1.5+2*\i,3)};
 }

 \foreach \i in {0,1,2}
 {
 \path (1.5+2*\i,3) node(a) {} (1.5+2*\i+2,3) node(b) {};
         \draw[ultra thick,black] (a) -- (b);
 }
 \end{scope}
 }

 \draw [red] plot [only marks, mark size=3, mark=*] coordinates {(3.5,1.5)} ;
 \draw [red] plot [only marks, mark size=3, mark=*] coordinates {(5.5,1.5)} ;
 \draw [red] plot [only marks, mark size=3, mark=*] coordinates {(3.5,7.5)} ;
 \draw [red] plot [only marks, mark size=3, mark=*] coordinates {(5.5,7.5)} ;

 \foreach \i in {0,1}
 {
 \begin{scope}[shift={(2*\i,0)}]
 \path (3.5,7.5) node(a) {} (3.5,6) node(b) {};
  \draw [thick,dotted] (a) to (b);

 \path (3.5,7.5) node(a) {} (3.5,4.5) node(b) {};
  \draw [thick,dotted] (a.north) to [out=60,in=60] (b.north);

 \path (3.5,7.5) node(a) {} (3.5,3) node(b) {};
  \draw [thick,dotted] (a.north) to [out=45,in=45] (b.north);

 \path (3.5,1.5) node(a) {} (3.5,3) node(b) {};
  \draw [thick,dotted] (a) to (b);

 \path (3.5,1.5) node(a) {} (3.5,4.5) node(b) {};
  \draw [thick,dotted] (a.north) to [out=120,in=120] (b.south);

 \path (3.5,1.5) node(a) {} (3.5,6) node(b) {};
  \draw [thick,dotted] (a.north) to [out=135,in=135] (b.south);

 \end{scope}
 }

 \path (1.5,4.5) node(a) {} (-3,4) node(b) {};
  \draw [thick,red] (a) to (b);

 \path (7.5,4.5) node(a) {} (12,4) node(b) {};
  \draw [thick,red] (a) to (b);

 \path (1.5,4.5) node(a) {} (5,12) node(b) {};
  \draw [thick,red] (a) .. controls (-2,9) .. (b);

 \path (7.5,4.5) node(a) {} (5,-3) node(b) {};
  \draw [thick,red] (a) .. controls (11,0) .. (b);

 \draw [green] plot [only marks, mark size=0, mark=*] coordinates {(-6,1)}
node[label={[xshift=0mm,yshift=-7mm] \Large{$\VS_{i,j}$}}] {};

 \draw [green] plot [only marks, mark size=0, mark=*] coordinates {(15,1)}
node[label={[xshift=0mm,yshift=-7mm] \Large{$\VS_{i+1,j}$}}] {};

 \draw [blue] plot [only marks, mark size=0, mark=*] coordinates {(4.5,-10)}
node[label={[xshift=0mm,yshift=-6mm] \Large{$\HS_{i,j}$}}] {};

 \draw [blue] plot [only marks, mark size=0, mark=*] coordinates {(4.5,19)}
node[label={[xshift=0mm,yshift=-4mm] \Large{$\HS_{i,j+1}$}}] {};

 \draw [black] plot [only marks, mark size=0, mark=*] coordinates {(2,0)}
node[label={[xshift=4mm,yshift=-1mm] \Large{$\M_{i,j}$}}] {};

 \draw [black] plot [only marks, mark size=0, mark=*] coordinates {(1.5,4.5)}
node[label={[xshift=-10mm,yshift=-2mm] \footnotesize{$\M_{i,j}(0_{x,y})$}}] {};

 \draw [black] plot [only marks, mark size=0, mark=*] coordinates {(7.5,4.5)}
node[label={[xshift=3mm,yshift=-1mm] \footnotesize{$\M_{i,j}(3_{x,y})$}}] {};

 \draw [black] plot [only marks, mark size=0, mark=*] coordinates {(5,12)}
node[label={[xshift=6mm,yshift=-7mm] \footnotesize{$\HS_{i,j+1}(3_{x})$}}] {};

 \draw [black] plot [only marks, mark size=0, mark=*] coordinates {(-3,4)}
node[label={[xshift=6mm,yshift=-7mm] \footnotesize{$\VS_{i,j}(3_{y})$}}] {};

 \draw [black] plot [only marks, mark size=0, mark=*] coordinates {(5,-3)}
node[label={[xshift=-2mm,yshift=0mm] \footnotesize{$\HS_{i,j}(0_{x})$}}] {};

 \draw [black] plot [only marks, mark size=0, mark=*] coordinates {(12,4)}
node[label={[xshift=-5mm,yshift=-6mm] \footnotesize{$\VS_{i+1,j}(0_{y})$}}] {};

 \end{tikzpicture}

 \caption{A zoomed-in view of the main gadget $\M_{i,j}$ surrounded by four
secondary gadgets: vertical gadget $\HS_{i,j+1}$ on the top, horizontal gadget
$\VS_{i,j}$ on the left, vertical gadget $\HS_{i,j}$ on the bottom and horizontal
gadget $\VS_{i+1,j}$ on the right. Each of the secondary gadgets is a copy of
the uniqueness gadget $U_n$ (see \autoref{sec:uniqueness-gadget}) and the main
gadget $\M_{i,j}$ is a copy of the uniqueness gadget $U_{|S_{i,j}|}$. The only
inter-gadget edges are the red edges: they have one end-point in a main gadget
and the other end-point in a secondary gadget. We have shown four such red
edges which are introduced for every $(x,y)\in S_{i,j}$.
  \label{fig:bi-dsn-planar-main-and-4-surrounding}}
 \end{figure}

Fix some $1\leq i,j\leq k$. The main gadget $\M_{i,j}$ has four secondary
gadgets\footnote{Half of the secondary gadgets are called ``horizontal'' since
their base edges are horizontal (as seen by the reader), and the other half of
the secondary gadgets are called ``vertical''.} surrounding it:
\begin{itemize}
\item above $\M_{i,j}$ is the vertical secondary gadget $\HS_{i,j+1}$,
\item on the right of $\M_{i,j}$ is the horizontal secondary gadget $\VS_{i+1,j}$,
\item below $\M_{i,j}$ is the vertical secondary gadget $\HS_{i,j}$, and
\item on the left of $\M_{i,j}$ is the horizontal secondary gadget $\VS_{i,j}$.
\end{itemize}
Hence, there are $k(k+1)$ horizontal secondary gadgets and $k(k+1)$ vertical
secondary gadgets.
Recall that $\M_{i,j}$ is a copy of $U_{|S_{i,j}|}$ and each of the secondary
gadgets are copies of $U_n$ (with $M=k^4$). With slight abuse of notation, we
assume that the rows of $\M_{i,j}$ are indexed by the set $\big\{(x,y)\ :\ (x,y)\in
S_{i,j}\big\}$.
We add the following edges (in \textcolor[rgb]{1.00,0.00,0.00}{red} color) of
weight $1$. For each~$(x,y)\in S_{i,j}$ add an edge connecting
\begin{itemize}
\item $\HS_{i,j+1}(3_x)$ and $\M_{i,j}(0_{(x,y)})$,
\item $\VS_{i,j}(3_y)$ and $\M_{i,j}(0_{(x,y)})$,
\item $\VS_{i+1,j}(0_y)$ and $\M_{i,j}(3_{(x,y)})$, and
\item $\HS_{i,j}(0_x)$ and $\M_{i,j}(3_{(x,y)})$.
\end{itemize}
Introduce the following $4k$ vertices (which we call \emph{border} vertices):
\begin{itemize}
\item $a_1, a_2, \ldots, a_k$,
\item $b_1, b_2, \ldots, b_k$,
\item $c_1, c_2, \ldots, c_k$,
\item $d_1, d_2, \ldots, d_k$.
\end{itemize}
For each $i\in [k]$ add an edge (in \textcolor[rgb]{1.00,0.50,0.00}{orange} color in
\autoref{fig:bi-dsn-planar-big-picture}) with weight 1 connecting
\begin{itemize}
\item $a_{i}$ and $\HS_{i,k+1}(0_{j})$ for each $j\in [n]$,
\item $b_{i}$ and $\HS_{i,1}(3_{j})$ for each $j\in [n]$,
\item $c_{i}$ and $\VS_{1,i}(0_{j})$ for each $j\in [n]$, and
\item $d_{i}$ and $\VS_{k+1,i}(3_{j})$ for each $j\in [n]$.
\end{itemize}

\begin{center}
\noindent\framebox{\begin{minipage}{\textwidth}
We follow the convention that for each $i\in [k]$:
\begin{itemize}
\item $\M_{0,i}(\cdot) = c_i $ and $\M_{k+1, i}(\cdot) = d_i$ irrespective of the argument, and
\item $\M_{i,0}(\cdot) = b_i $ and $\M_{i,k+1}(\cdot) =a_i$ irrespective of the argument.
\end{itemize}
\end{minipage}}
\end{center}

This concludes the construction of the graph $G^*$. Note that we bidirect each edge of $G^*$. Finally, the set of demand pairs $\mathcal{D}$ is given by:
\begin{itemize}
\item \underline{Type I}: Let $1\leq i\leq k+1, 1\leq j\leq k$. Consider the
horizontal secondary gadget $\VS_{i,j}$. We add the pairs $\big(\M_{i-1,j}(s_1),
\VS_{i,j}(t_1)\big)$ and $\big(\M_{i-1,j}(s_2), \VS_{i,j}(t_2)\big)$ in addition to the pairs
$\big(\VS_{i,j}(s_1), \M_{i,j}(t_1)\big)$ and $\big(\VS_{i,j}(s_2), \M_{i,j}(t_2)\big)$.

\item \underline{Type II}: Let $1\leq j\leq k+1, 1\leq i\leq k$. Consider the
vertical secondary gadget $\HS_{i,j}$. We add the pairs $(\M_{i,j}(s_1),
\HS_{i,j}(t_1))$ and $(\M_{i,j}(s_2), \HS_{i,j}(t_2))$ in addition to the pairs
$(\HS_{i,j}(s_1), \M_{i,j-1}(t_1))$ and $(\HS_{i,j}(s_2), \M_{i,j-1}(t_2))$
\end{itemize}

We have $O(k^2)$ vertical and horizontal secondary gadgets and we add $O(1)$
demand pairs corresponding to each of these gadgets. Hence, the total number of
demand pairs is
\begin{equation}\label{eqn:D-bidsnP}
|\mathcal{D}|=O(k^2)
\end{equation}

Fix the budget $B^*= 4k+ k(k+1)\cdot B + k(k+1)\cdot B + k^{2}\cdot (B+4)$
where $B=7M= 7k^4$. The high-level intuition is the following: we need $4k$ from the budget (via \textcolor[rgb]{1.00,0.50,0.00}{orange} edges) just to include one edge incident on each of the $4k$ border vertices (each of which is part of a demand pair). We have $k(k+1)$
horizontal and vertical secondary gadgets each. We argue that any solution for
\bidsn must satisfy the in-out property in each of the secondary gadgets,
and then invoke \autoref{lem:macro-uniqueness-gadget}. Finally, for each main
gadget, we again show that it must satisfy the in-out property and hence has
cost at least $B$. %
However, here we show that we
additionally need at least four
red edges and hence the cost of any \bidsn solution restricted to a main gadget
is at least $B+4$. Since we have $k^2$ main gadgets, this completely uses up the
budget $B^*$.

We now show the correctness of our reduction by showing that the instance $\big(k,n,\{S_{i,j}\}_{1\leq i,j\leq k} \big)$ of \gt has a solution if and only if the instance $(G^*, \mathcal{D})$ of \bidsnP has a planar solution of cost $\leq B^*$. First we show the forward direction:
\begin{lem}
\label{lem:bidsnp-easy-gt}
If the instance $\big(k,n,\{S_{i,j}\}_{1\leq i,j\leq k} \big)$ of \gt has a 
solution then the instance~$(G^*, \mathcal{D})$ of \bidsnP has a planar 
solution of cost at most $B^*$
\end{lem}
\begin{proof}
Suppose that the instance $\big(k,n,\{S_{i,j}\}_{1\leq i,j\leq k} \big)$ of \gt has a solution, i.e., there exist $1\leq \alpha_1, \alpha_2, \ldots, \alpha_k\leq n$ and $1\leq \beta_1, \beta_2, \ldots, \beta_k \leq n$ such that for each $1\leq i,j\leq k$ we have $(\alpha_i, \beta_j)\in S_{i,j}$.
We now build an edge set $E' \subseteq E(G^*)$ of weight $\leq B^*$ such that the network $\big(V(G^*),E'\big)$ is a \emph{planar} solution for the \bidsn instance $(G^*,
\mc{D})$. In the edge set $E'$, we take the following edges:
\begin{enumerate}
\item The \textcolor[rgb]{1.00,0.50,0.00}{orange} edges $(c_j, \VS_{1,j}(0_{\beta_j}))$ and $(\VS_{k+1,j}(3_{\beta_j}),
d_j)$ for each $j\in [k]$. This uses up $2k$ from the budget since each of
these edges has weight 1.

\item The \textcolor[rgb]{1.00,0.50,0.00}{orange} edges $(a_i, \HS_{i,k+1}(0_{\alpha_i}))$ and $(\HS_{i,1}(3_{\alpha_i}),
b_i)$ for each $i\in [k]$. This uses up $2k$ from the budget since each of
these edges has weight 1.

\item For each $1\leq i,j\leq k$ for the main gadget $\M_{i,j}$, use
\autoref{crl:macro} to pick a set of edges
$E^{\text{right}}_{\M_{i,j}}\big( (\alpha_i, \beta_j) \big)$
which is right-oriented, represented by $(\alpha_i, \beta_j)$ and has weight
exactly $B$.
Additionally we also pick the following four
\textcolor[rgb]{1.00,0.00,0.00}{red} edges (each of which has weight 1):
        \begin{itemize}
        \item $\HS_{i,j+1}(3_{\alpha_i})\rightarrow \M_{i,j}(0_{\alpha_i,
\beta_j})$
        \item $\VS_{i,j}(3_{\beta_j})\rightarrow \M_{i,j}(0_{\alpha_i,
\beta_j})$
        \item $\HS_{i,j}(0_{\alpha_i})\leftarrow \M_{i,j}(3_{\alpha_i, \beta_j})$
        \item $\VS_{i+1,j}(0_{\beta_j})\leftarrow \M_{i,j}(3_{\alpha_i, \beta_j})$
        \end{itemize}

\item For each $1\leq j\leq k+1$ and $1\leq i\leq k$ for the vertical
secondary gadget $\HS_{i,j}$, use
    \autoref{crl:macro} to pick a set of edges
$E^{\text{right}}_{\HS_{i,j}}(\alpha_i)$
which is right-oriented, represented by $\alpha_i$ and has weight exactly~$B$.

\item For each $1\leq j\leq k$ and $1\leq i\leq k+1$ for the horizontal secondary
gadget $\VS_{i,j}$, use
    \autoref{crl:macro} to pick a set of edges
$E^{\text{right}}_{\VS_{i,j}}(\beta_j)$
which is right-oriented, represented by $\beta_j$ and has weight exactly~$B$.
\end{enumerate}

It is easy to see that $\big(V(G^*),E'\big)$ is planar, as each application of
\autoref{crl:macro} gives a planar graph, separate applications
give vertex-disjoint graphs, and moreover, the \textcolor[rgb]{1.00,0.00,0.00}{red} and \textcolor[rgb]{1.00,0.50,0.00}{orange} edges do not
destroy planarity either. The weight of the edge set $E'$ is exactly $4k+ k(k+1)\cdot B + k(k+1)\cdot B + k^{2}\cdot
(B+4) = B^*$.
We now show that the network $\big(V(G^*),E'\big)$ is indeed a solution for the \bidsn instance $(G^*, \mathcal{D})$.
Fix $i,j$ such that $1\leq i\leq k+1, 1\leq j\leq k$. Consider the four demand
pairs of Type~I.
\begin{itemize}
\item \underline{$\big( \M_{i-1,j}(s_1), \VS_{i,j}(t_1) \big)$}. This path is obtained by
concatenation of the following paths:
        \begin{itemize}
        \item from $E_{\M_{i-1,j}}^{\text{right}}\big( (\alpha_{i-1},\beta_j) \big)$ use the path $\M_{i-1,j}(s_1)\rightarrow
\M_{i-1,j}(1_{\alpha_{i-1}, \beta_j})\rightarrow \M_{i-1,j}(2_{\alpha_{i-1},
\beta_j})\rightarrow \M_{i-1,j}(3_{\alpha_{i-1}, \beta_j})$
        \item use the edge
$\M_{i-1,j}(3_{\alpha_{i-1}, \beta_j})\rightarrow \VS_{i,j}(0_{\beta_j})$
        \item from $E_{\VS_{i,j}}^{\text{right}}(\beta_j)$ use the path $\VS_{i,j}(0_{\beta_j})\rightarrow
\VS_{i,j}(1_{\beta_j})\rightarrow \VS_{i,j}(t_1)$.
        \end{itemize}

\item \underline{$\big( \M_{i-1,j}(s_2), \VS_{i,j}(t_2) \big)$}. This path is obtained by
concatenation of the following paths:
        \begin{itemize}
        \item from $E_{\M_{i-1,j}}^{\text{right}}\big( (\alpha_{i-1},\beta_j) \big)$ use the path $\M_{i-1,j}(s_2)\rightarrow
\M_{i-1,j}(2_{\alpha_{i-1}, \beta_j})\rightarrow \M_{i-1,j}(3_{\alpha_{i-1}, \beta_j})$
        \item use the edge $\M_{i-1,j}(3_{\alpha_{i-1}, \beta_j})\rightarrow \VS_{i,j}(0_{\beta_j})$
        \item from $E_{\VS_{i,j}}^{\text{right}}(\beta_j)$ use the path $\VS_{i,j}(0_{\beta_j})\rightarrow
\VS_{i,j}(1_{\beta_j})\rightarrow \VS_{i,j}(2_{\beta_j})\rightarrow
\VS_{i,j}(t_2)$.
        \end{itemize}

\item \underline{$\big( \VS_{i,j}(s_1), \M_{i,j}(t_1) \big)$}. This path is obtained by
concatenation of the following paths:
        \begin{itemize}
        \item from $E_{\VS_{i,j}}^{\text{right}}(\beta_j)$ use the path $\VS_{i,j}(s_1)\rightarrow
\VS_{i,j}(1_{\beta_j})\rightarrow \VS_{i,j}(2_{\beta_j})\rightarrow
\VS_{i,j}(3_{\beta_j})$,
        \item use the edge $\VS_{i,j}(3_{\beta_j})\rightarrow \M_{i,j}(0_{\alpha_i, \beta_j})$
        \item from $E_{\M_{i,j}}^{\text{right}}\big( (\alpha_i,\beta_j) \big)$ use the path $\M_{i,j}(0_{\alpha_i, \beta_j})\rightarrow \M_{i,j}(1_{\alpha_i,
\beta_j})\rightarrow \M_{i,j}(t_1)$.
        \end{itemize}

\item \underline{$\big( \VS_{i,j}(s_2), \M_{i,j}(t_2) \big)$}. This path is obtained by
concatenation of the following paths:
        \begin{itemize}
        \item from $E_{\VS_{i,j}}^{\text{right}}(\beta_j)$ use the path $\VS_{i,j}(s_2)\rightarrow
\VS_{i,j}(2_{\beta_j})\rightarrow \VS_{i,j}(3_{\beta_j})$
        \item use the edge $\VS_{i,j}(3_{\beta_j})\rightarrow \M_{i,j}(0_{\alpha_i, \beta_j})$
        \item from $E_{\M_{i,j}}^{\text{right}}\big( (\alpha_i,\beta_j) \big)$  use the path $\M_{i,j}(0_{\alpha_i, \beta_j})\rightarrow \M_{i,j}(1_{\alpha_i,
\beta_j})\rightarrow \M_{i,j}(2_{\alpha_i, \beta_j})\rightarrow \M_{i,j}(t_2)$.
        \end{itemize}
\end{itemize}
The analysis for demand pairs of Type II is analogous, and therefore
omitted here.
\end{proof}

Our next lemma shows the reverse direction of the correctness of the reduction: 
a solution of cost at most $B^*$ for the instance $(G^*, \mathcal{D})$ of 
\bidsnP implies a solution for the instance $\big(k,n,\{S_{i,j}\}_{1\leq i,j\leq 
k} \big)$ of \gt. This implies that if the instance 
$\big(k,n,\{S_{i,j}\}_{1\leq i,j\leq k} \big)$ of \gt does not have a solution 
then the cost of an optimal solution (and hence the cost of an optimal planar 
solution, if one exists) for the instance $(G^*, \mathcal{D})$ of \bidsnP is 
more than $B^*$.

\begin{lem}
\label{lem:bidsnp-hard-gt}
If the instance $(G^*, \mathcal{D})$ of \bidsnP has a solution of cost at most 
$B^*$ then the instance $\big(k,n,\{S_{i,j}\}_{1\leq i,j\leq k} \big)$ of 
\gt has a solution.
\end{lem}
\begin{proof}
Suppose that the instance $(G^*, \mathcal{D})$ of \bidsnP has a solution, say $N:=\big(V(G^*), E^*\big)$, of cost at most $B^*$.

\begin{claim}
For every $1\leq i\leq k+1, 1\leq j\leq k$ the edge set $E^*$ restricted to the
horizontal secondary gadget $\VS_{i,j}$ satisfies the in-out property. Hence,
$\VS_{i,j}$ uses up weight of at least $B$ from the budget.
\label{lem:in-out-vertical}
\end{claim}
\begin{proof}
Looking at the demand pairs in $\mathcal{D}$ of Type I, we observe that
\begin{itemize}
\item $\VS_{i,j}(s_1)$ is the source of some demand pair whose other end-point
lies outside of $\VS_{i,j}$,
\item $\VS_{i,j}(s_2)$ is the source of some demand pair whose other end-point
lies outside of $\VS_{i,j}$,
\item $\VS_{i,j}(t_1)$ is the target of some demand pair whose other end-point
lies outside of $\VS_{i,j}$,
\item $\VS_{i,j}(t_2)$ is the target of some demand pair whose other end-point
lies outside of $\VS_{i,j}$.
\end{itemize}
Since the network $\big(V(G^*),E^*\big)$ is a solution of the \bidsn instance, it follows that there is a
path starting at $\VS_{i,j}(s_1)$ which must leave the gadget $\VS_{i,j}$, i.e.,
$\VS_{i,j}(s_1)$ can reach either a 0-vertex or a 3-vertex. The other three
conditions of \autoref{defn-in-out} follow by similar reasoning. By
\autoref{lem:macro-uniqueness-gadget}, it follows that $\VS_{i,j}$ uses up weight
of at least $B$ from the budget.
\end{proof}

Analogous claims hold also for the vertical secondary gadgets and main gadgets:
\begin{claim}
For every $1\leq i\leq k, 1\leq j\leq k+1$ the edge set $E^*$ restricted to the
vertical secondary gadget $\HS_{i,j}$ satisfies the in-out property.
Hence, $\HS_{i,j}$ uses up weight of at least $B$ from the budget.
\label{lem:in-out-horizontal}
\end{claim}

\begin{claim}
For every $1\leq i,j\leq k$ the edge set $E^*$ restricted to the main gadget
$\M_{i,j}$ satisfies the in-out property. Hence, $\M_{i,j}$ uses up weight of
at least $B$ from the budget.
\label{lem:in-out-main}
\end{claim}

From \autoref{lem:in-out-vertical}, \autoref{lem:in-out-horizontal} and
\autoref{lem:in-out-main} we know that each of the gadgets (horizontal
secondary, vertical secondary and main) use up at least weight $B$ in $E^*$. We
now claim that $E^*$ restricted to each vertical secondary gadget, horizontal
secondary gadget and main gadget has weight exactly $B$. Suppose there is at
least one gadget where $E^*$ has weight more than $B$. By
\autoref{lem:macro-uniqueness-gadget}, the weight of $E^*$ in this gadget is at
least $8M= B+M$. Since $M=k^4$ and $B^* = 4k+ k(k+1)\cdot B + k(k+1)\cdot B +
k^{2}\cdot (B+4)$, where $B=7M$, the weight of $E^*$ is at least
\begin{multline*}
\Big(k(k+1)\cdot B +k(k+1)\cdot B + k^{2}\cdot B \Big)+ M  = 
\Big(k(k+1)\cdot B +k(k+1)\cdot B + k^{2}\cdot B \Big)+ k^4 > \\
\Big(k(k+1)\cdot B +k(k+1)\cdot B + k^{2}\cdot B \Big)+ 4(k+k^2) = B^*,
\end{multline*}
which is a contradiction (since $k^4 > 4k+4k^2$ for $k\geq 3$).

Therefore, together with \autoref{lem:macro-uniqueness-gadget}, we have the following:

\begin{claim}
The weight of $E^*$ restricted to each main gadget and each secondary gadget is exactly $B=7M$. Moreover,
\begin{itemize}

\item for each $1\leq i\leq k+1, 1\leq j\leq k$, the horizontal secondary gadget
$\VS_{i,j}$ is represented by some $y_{i,j}\in [n]$ and is either right-oriented
or left-oriented,

\item for each $1\leq i\leq k, 1\leq j\leq k+1$, the vertical secondary
gadget $\HS_{i,j}$ is represented by some $x_{i,j}\in [n]$ and is either
right-oriented or left-oriented, and

\item for each $1\leq i, j\leq k$, the main gadget $\M_{i,j}$ is represented by
some $(\lambda_{i,j}, \delta_{i,j})\in S_{i,j}$ and is either right-oriented or
left-oriented.
\end{itemize}
\label{thm:representing-planar}
\end{claim}

We now show that for each $1\leq i,j\leq k$, the edge set $E^*$ must also contain
some \textcolor[rgb]{1.00,0.00,0.00}{red} edges which have exactly one end-point in
vertices of $\M_{i,j}$.

\begin{claim}
For each $1\leq i,j\leq k$, the edge set $E^*$ must contain at least one
\textcolor[rgb]{1.00,0.00,0.00}{red} edge of each of the following four types:
\begin{itemize}
  \item an edge with one end-point in the set of $3$-vertices of $\M_{i,j}$ and other end-point in the set of $0$-vertices of $\HS_{i,j}$
  \item an edge with one end-point in the set of $3$-vertices of $\M_{i,j}$ and other end-point in the set of $0$-vertices of $\VS_{i+1,j}$
  \item an edge with one end-point in the set of $0$-vertices of $\M_{i,j}$ and other end-point in the set of $3$-vertices of $\VS_{i,j}$
  \item  an edge with one end-point in the set of $0$-vertices of $\M_{i,j}$ and other end-point in the set of $3$-vertices of $\HS_{i,j+1}$
\end{itemize}
\label{lem:main-at-least-4-red-planar}
\end{claim}
\begin{proof}
We show that $E^*$ must use at least one red edge which has one end-point in
the set of $3$-vertices of $\M_{i,j}$ and the other end-point in the set of
$0$-vertices of $\HS_{i,j}$. Analogous arguments hold for the other three
secondary gadgets surrounding the main gadget $\M_{i,j}$. Note that these four red edges are distinct since the vertex sets of the secondary gadgets are pairwise disjoint.

We know by \autoref{thm:representing-planar} that $\HS_{i,j}$ is either
right-oriented or left-oriented. Suppose $\HS_{i,j}$ is right-oriented (the
other case is analogous). By
\autoref{lem:macro-uniqueness-gadget} and ~\autoref{thm:representing-planar}, we
know that the only base edges of $\HS_{i,j}$ picked in $E^*$ are
$\HS_{i,j}(0_{x_{i,j}})\rightarrow \HS_{i,j}(1_{x_{i,j}})\rightarrow
\HS_{i,j}(2_{x_{i,j}})\rightarrow \HS_{i,j}(3_{x_{i,j}})$ and the only connector
edges of $\HS_{i,j}$ picked in $E^*$ are $\HS_{i,j}(s_1)\rightarrow
\HS_{i,j}(1_{x_{i,j}}), \HS_{i,j}(s_2)\rightarrow \HS_{i,j}(2_{x_{i,j}}),
\HS_{i,j}(t_1)\leftarrow \HS_{i,j}(1_{x_{i,j}})$ and $\HS_{i,j}(t_2)\leftarrow
\HS_{i,j}(2_{x_{i,j}})$. But there is a Type II demand pair whose target is
$\HS_{i,j}(t_1)$: the path satisfying this demand pair has to enter $\HS_{i,j}$
through the vertex $\HS_{i,j}(0_{x_{i,j}})$. The only edges (which are not base
or connector edges of $\HS_{i,j}$) incident on the 0-vertices of $\HS_{i,j}$ have
their other end-point in the set of 3-vertices of $\M_{i,j}$. That is, $E^*$
contains a
red edge whose start vertex is a 3-vertex of $\M_{i,j}$ and end vertex is a
0-vertex of $\HS_{i,j}$.
\end{proof}

\begin{claim}
For each $1\leq i,j\leq k$, the edge set $E^*$ must contain at least one \textcolor[rgb]{1.00,0.50,0.00}{orange}
edge of each of the following types:
\begin{itemize}
\item an outgoing edge from $a_i$,
\item an incoming edge into $b_i$,
\item an outgoing edge from $c_j$, and
\item an incoming edge into $d_j$.
\end{itemize}
\label{lem:dotted-planar}
\end{claim}
\begin{proof}
Note that for each $i\in [k]$ the vertex $a_i$ is the source of two Type II demand pairs, viz.\ $(a_i,
\HS_{i,k+1}(t_1))$ and $(a_i, \HS_{i,k+1}(t_2))$. Hence $E^*$ must contain at
least one outgoing edge from $a_i$. The other three statements follow similarly.
\end{proof}

We now show that there is no slack, i.e., the weight of $E^*$ must be exactly $B^*$.

\begin{claim}
The weight of $E^*$ is exactly $B^*$, and is inclusion-wise minimal.
\label{thm:exact-B^*-planar}
\end{claim}
\begin{proof}
From \autoref{thm:representing-planar}, we know that each main gadget and each secondary gadget contributes towards a weight of
$B$ in $E^*$. \autoref{lem:main-at-least-4-red-planar} says that each main
gadget needs at least 4 red edges, and \autoref{lem:dotted-planar} says that the
\textcolor[rgb]{1.00,0.50,0.00}{orange} edges contribute at least $4k$ to weight of $E^*$. Since all these edges
are distinct, we have that the weight of $E^*$ is at least $k(k+1)\cdot B +
k(k+1)\cdot B + k^{2}\cdot B + 4(k+k^2) = B^*$. Hence, the weight of $E^*$ is
exactly $B^*$, and it is minimal (under edge deletions) since no edges have
zero weights.
\end{proof}

Consider a main gadget $\M_{i,j}$. The main gadget has four secondary gadgets
surrounding
it: $\HS_{i,j}$ below it, $\HS_{i,j+1}$ above it, $\VS_{i,j}$ to the left and
$\VS_{i+1,j}$ to the right. By \autoref{thm:representing-planar}, these gadgets
are represented by $x_{i,j}, x_{i,j+1}, y_{i,j}$ and $y_{i+1,j}$ respectively.
The main gadget $\M_{i,j}$ is represented by $(\lambda_{i,j}, \delta_{i,j})$.

\begin{claim}[propagation]
For every main gadget $\M_{i, j}$, we have
$x_{i,j}=\lambda_{i,j}=x_{i,j+1}$ and $y_{i,j}=\delta_{i,j}=y_{i+1,j}$.
\label{lem:agreement-tight-planar}
\end{claim}
\begin{proof}
Due to symmetry, it suffices to only argue that $x_{i,j}=\lambda_{i,j}$. Let us
assume for the sake of contradiction that $x_{i,j}\ne\lambda_{i,j}$. We will
now show that there is a vertex such that (1) there is exactly one edge adjacent
to it in the solution $E^*$ and (2) it does not belong to any demand pair.
Observe that removing its only adjacent edge from $E^*$ does not effect the
validity of the solution. This contradicts \autoref{thm:exact-B^*-planar} which
states that $E^*$ is minimal.

From \autoref{lem:main-at-least-4-red-planar} and
\autoref{thm:exact-B^*-planar}, it follows that $E^*$ contains exactly one red
edge, say $e_1$, which has one end-point in the set of $3$-vertices of
$\M_{i,j}$ and the other end-point in set of $0$-vertices of $\HS_{i,j}$. Also,
$E^*$ contains exactly one red edge, say $e_2$, which has one end-point in the
set of $3$-vertices of $\M_{i,j}$ and the other end-point in set of $0$-vertices
of $\VS_{i+1,j}$.

Observe that if $e_1$ does not have one endpoint at $\HS_{i, j}(0_{x_{i, j}})$,
then the vertex $\HS_{i, j}(0_{x_{i, j}})$ is the desired vertex. Hence, suppose
that one endpoint of the edge $e_1$ is $\HS_{i, j}(0_{x_{i, j}})$. The other
endpoint of $e_1$ must be $\M_{i, j}(3_{x_{i, j}, y})$ for some $y \in V_j$.
Since $x_{i, j} \ne \lambda_{i, j}$, we have $\M_{i, j}(3_{x_{i, j}, y}) \ne
\M_{i, j}(3_{\lambda_{i, j}, \delta_{i, j}})$. Suppose that one of the
end-points of $e_2$ is $\M_{i, j}(3_{x', y'})$. Since $\M_{i, j}(3_{x_{i, j}, y})
\ne \M_{i, j}(3_{\lambda_{i, j}, \delta_{i, j}})$, at least one of the following
must be true: $\M_{i, j}(3_{x_{i, j}, y}) \ne \M_{i, j}(3_{x', y'})$ or $\M_{i,
j}(3_{\lambda_{i, j}, \delta_{i, j}}) \ne \M_{i, j}(3_{x', y'})$. If $\M_{i,
j}(3_{x_{i, j}, y}) \ne \M_{i, j}(3_{x',y'})$, then
$\M_{i, j}(3_{x_{i, j}, y})$ is the desired vertex. Otherwise, if $\M_{i,
j}(3_{\lambda_{i, j}, \delta_{i, j}}) \ne \M_{i, j}(3_{x', y'})$, then $\M_{i,
j}(3_{\lambda_{i, j}, \delta_{i, j}})$ is the desired vertex.
In all cases, we have found a vertex with the desired properties. Hence, we
have arrived at a contradiction.
\end{proof}

By \autoref{lem:agreement-tight-planar}, it follows that for each
$1\leq i,j\leq k$ we have $x_{i,j}=\lambda_{i,j}=x_{i,j+1}$ and
$y_{i,j}=\delta_{i,j}=y_{i+1,j}$ in addition to $(\lambda_{i,j},
\delta_{i,j})\in S_{i,j}$ (by the definition of the main gadget). This implies
that the instance $\big( k, n, \{S_{i,j}\}_{1\leq i,j\leq k} \big)$ of \gt has a solution.
This concludes the proof of~\autoref{lem:bidsnp-hard-gt}.
\end{proof}

Finally we are ready to prove~\autoref{thm:lb-scheme-biDSN} which is restated below:
\thmlbschemebiDSN*
\begin{proof}
Each main gadget has $O(n^2)$ vertices and $G^*$ has $O(k^2)$ main gadgets. Each secondary gadget has $O(n)$ vertices and $G^*$ has $O(k^2)$ secondary gadgets. The number of border vertices in $G^*$ is $4k$. Hence, the total number of vertices in the graph $G^*$ is $O(n^{2}k^{2})=\poly(n,k)$. It is easy to see that $G^*$ can be constructed in $\poly(n,k)$ time from a given instance $\big( k, n, \{S_{i,j}\}_{1\leq i,j\leq k} \big)$ of \gt. It is known~\cite[Theorem 14.28]{pc-book} that $k\times k$ \gt is W[1]-hard parameterized by $k$, and under ETH cannot be solved in $f(k)\cdot n^{o(k)}$ time for any computable function $f$. Combining the two directions from~\autoref{lem:bidsnp-easy-gt} and~\autoref{lem:bidsnp-hard-gt}, we get a parameterized reduction from $k\times k$ \gt to an instance of \bidsnP with $|\mathcal{D}|=O(k^2)$ terminal pairs (\autoref{eqn:D-bidsnP}). Hence, it follows that \bidsnP is W[1]-hard parameterized by the number $k$ of terminal pairs and under ETH cannot be solved in $f(k)\cdot n^{o(\sqrt{k})}$ time for any computable function $f$.

Suppose now that there is an algorithm $\mathbb{A}$ which runs in time
$f(k,\eps)\cdot n^{o(\sqrt{k})}$ and computes an $(1+\eps)$-approximate
solution for \bidsnP. Recall that our reduction works as follows: the instance $\big( k, n, \{S_{i,j}\}_{1\leq i,j\leq k} \big)$ of \gt
is a YES instance if and only if the instance $(G^*, \mathcal{D})$ of \bidsnP has a planar solution of cost $B^* =
4k+ k(k+1)\cdot B
+ k(k+1)\cdot B + k^{2}\cdot (B+4) = O(k^6)$ since $B=O(k^4)$.
Consequently, consider running $\mathbb{A}$ with $\eps$ set to a
value such that $(1+\eps)\cdot B^{*} <B^{*}+1$ with $\eps$ being a function of
the parameter $k$ independent of $n$. Every edge of our constructed graph $G^*$
has weight at least $1$, and hence a $(1+\epsilon)$-approximation is in fact
forced to find a solution of cost at most~$B^*$, i.e., $\mathbb{A}$ finds
an optimum solution. By the previous paragraph, this is not possible.
\end{proof}

\subsection{W[1]-hardness for \altbidsn}
\label{sec:bidsn-w[1]}

The goal of this section is to prove \autoref{thm:lb-biDSN}.
We reduce from the \csi problem\footnote{This is sometimes also known as \textsc{Colored Subgraph Isomorphism}} introduced by \citet{marx-beat-treewidth}.

\begin{center}
\noindent\framebox{\begin{minipage}{5.0in}
\textbf{\csi} (PSI)\\
\emph{Input}: An undirected graph $G=(V_G,E_G)$ and $H=(V_H = \{1, 2, \ldots,
\ell\},E_H)$, and a partition of $V_G$ into
disjoint subsets $V_1,V_2,\ldots, V_{\ell}$ \\
\emph{Question}: Is there a function $\phi: V_H\rightarrow V_G$ such that
                \begin{enumerate}
                \item for every $i\in [\ell]$ we have $\phi(i)\in V_i$, and
                \item for every edge $ij\in E_H$ we have
$\phi(i)\phi(j)\in E_G$.
                \end{enumerate}
\end{minipage}}
\end{center}

The W[1]-hardness of \csi parameterized by $|E_H|$ follows since the W[1]-hard problem \mcc~\cite{mcc-1,mcc-2} is a special case when $H$ is a clique. Marx~\cite[Corollary 6.3]{marx-beat-treewidth} showed the following stronger lower bound: under ETH, \csi cannot be solved in time $f(|E_H|)\cdot |V_G|^{o\big(|E_H|/\log |E_H|\big)}$ for any computable function $f$. To prove~\autoref{thm:lb-biDSN}, we give a reduction which transforms an instance $(G,H)$ of \csi into an instance $(G^*, \mathcal{D})$ of \bidsn which has $|\mathcal{D}|=O(|E_H|)$ demand pairs. This reduction is a modification of that given for \autoref{thm:lb-scheme-biDSN} in \autoref{sec:bidsn-planar-w[1]}: there we had the special case when $H$ is a clique on $\ell$ vertices and hence could ensure planarity (\autoref{fig:bi-dsn-planar-big-picture}) of the optimum at the cost of a quadratic blowup in the number of demand pairs as compared to $|V_H|$. In this reduction, we lose that structure but achieve the condition that number of demand pairs in the instance of \bidsn is linear in $|E_H|$.

Consider an instance of \csi
given by two undirected graphs $G=(V_G,E_G)$ and $H=(V_H = \{1, 2, \ldots,
\ell\},E_H)$, and a partition of $V_G$ into disjoint subsets
$V_1,V_2,\ldots, V_{\ell}$. We now build an instance $(G^*,\mathcal{D})$ of \bidsn. %
Let $|E_H|=k$. We define the following quantities:
\begin{itemize}
  \item For each $1\leq i\leq \ell$, $E^{*}_H := E_H \cup \{ii\ : 1\leq i\leq \ell\}$, i.e., we introduce self-loops
  \item For each $1\leq i\neq j\leq \ell$, $E_{\{i,j\}}\subseteq E_G$ is the set of edges which have one endpoint in $V_i$ and the other endpoint in $V_j$. For $1\leq j\leq \ell$ let $E_{\{j,j\}} = \{xx\ :\ x\in V_j\}$.

  \item For each $1\leq i\leq \ell$, $N_{H}(i):= \{j\ :\ ij\in E^*_H\}$
  \item For each $1\leq i\leq \ell$, $\alpha_i= \min \big\{j\ :\ j\in N_{H}(i)\big\}$ and $\beta_i=\max \big\{j\ :\ j\in N_{H}(i)\big\}$.
  \item For each $1\leq i\leq \ell$, $N'_{H}(i):= N_{H}(i)\cup \{\ell+1\}$

  \item For each $1\leq i, j\leq \ell$ such that $ij\in E^*_H$ we define  $\nextone_{j}(i)=\min\big\{\{\ell+1\}\cup \{r\ : r>j, ri\in E^*_H\}\big\}$. For each $i\in [\ell]$ we define $\nextone_{0}(i)=\alpha_i$ and $\nextone_{\ell+1}(i)=\ell+1$.
  \item For each $1\leq i, j\leq \ell$ such that $ij\in E^*_H$ we define $\prevone_{j}(i)=\max\big\{\{0\}\cup \{r\ : r<j, ri\in E^*_H\}\big\}$. For each $i\in [\ell]$ we define $\prevone_{\ell+1}(i)=\beta_i$ and $\prevone_{0}(i)=0$.
\end{itemize}

\begin{remark}
\label{remark:H-is-connected-for-CSI}
Note that the quantities $\alpha_i$ and $\beta_i$ are well-defined since we can
assume that $H$ is connected (and hence has no isolated vertices) since
otherwise we can solve the \csi instance separately
for each component of $H$.
\end{remark}

The graph~$G^*$ has two types of gadgets: the \emph{main gadget} and the \emph{secondary gadget}. Each of these gadgets are copies of the ``uniqueness gadget'' from
\autoref{sec:uniqueness-gadget} with~$M=k^4$.
\begin{itemize}
  \item For each $1\leq i,j\leq \ell$ such that $ij\in E^{*}_H$ the gadget $\M_{i,j}$ (corresponding to the set $E_{\{i,j\}}$) is a copy of~$U_{|E_{\{i,j\}}|}$.
  \item For each $i\in [\ell]$ and each $j\in N'_{H}(i)$, the vertical secondary gadget $\HS_{i,j}$ is a copy of $U_{|V_i|}$.
  \item For each $j\in [\ell]$ and each $i\in N'_{H}(j)$, the horizontal secondary gadget $\VS_{i,j}$ is a copy of $U_{|V_j|}$.
\end{itemize}

Hence, we have a total of $2k+\ell$ main gadgets, and $2k+2\ell$ horizontal and vertical secondary gadgets each. For each $1\leq i,j\leq \ell$ such that $ij\in E^{*}_H$, the main gadget
$\M_{i,j}$ is surrounded (\autoref{fig:bi-dsn-main-and-4-surrounding}) by the following four secondary gadgets:
\begin{itemize}
  \item $\VS_{i,j}$ on the left,
  \item $\HS_{i,j}$ on the bottom,
  \item $\VS_{\nextone_{i}(j),j}$ on the right, and
  \item $\HS_{i,\nextone_{j}(i)}$ on the top.
\end{itemize}

Recall that for each $1\leq i,j\leq \ell$ such that $ij\in E^{*}_H$, the main gadget $\M_{i,j}$ is a copy of $U_{|E_{\{i,j\}}|}$ with $M=k^4$. With
slight abuse of notation, we assume that the rows of $\M_{i,j}$ are indexed by
the set $\big\{\{x,y\}\ :\ \{x,y\}\in E_{\{i,j\}, x\in V_i, y\in V_j}\big\}$.
For each $1\leq i,j\leq \ell$ such that $ij\in E^{*}_H$, we add an edge (in \textcolor[rgb]{1.00,0.00,0.00}{red} color) of
weight $1$ for each $\{x,y\}\in E_{\{i,j\}}$ connecting
\begin{itemize}
\item $\HS_{i,\nextone_{j}(i)}(3_x)$ and $\M_{i,j}(0_{(x,y)})$,
\item $\VS_{i,j}(3_y)$ and $\M_{i,j}(0_{(x,y)})$,
\item $\VS_{\nextone_{i}(j),j}(0_y)$ and $\M_{i,j}(3_{(x,y)})$, and
\item $\HS_{i,j}(0_x)$ and $\M_{i,j}(3_{(x,y)})$.
\end{itemize}

Introduce the following $4\ell$ vertices (which we call \emph{border} vertices):
\begin{itemize}
\item $a_1, a_2, \ldots, a_{\ell}$
\item $b_1, b_2, \ldots, b_\ell$
\item $c_1, c_2, \ldots, c_\ell$
\item $d_1, d_2, \ldots, d_\ell$
\end{itemize}

\begin{center}
\noindent\framebox{\begin{minipage}{\textwidth}
We follow the convention that for each $i\in [\ell]$:
\begin{itemize}
\item $\M_{0,i}(\cdot) = c_i$ and $\M_{\ell+1, i}(\cdot) = d_i$ irrespective of the argument, and
\item $\M_{i,0}(\cdot) = b_i$ and $\M_{i, \ell+1}(\cdot) = a_i$ irrespective of the argument.
\end{itemize}
\end{minipage}}
\end{center}

For each $i\in [\ell]$ add an edge (in \textcolor[rgb]{1.00,0.50,0.00}{orange} color in \autoref{fig:bi-dsn-planar-big-picture}) with weight 1 connecting
\begin{itemize}
\item $a_{i}$ and $\HS_{i,\ell+1}(0_{x})$ for each $x\in V_i$,
\item $b_{i}$ and $\HS_{i,\alpha_i}(3_{x})$ for each $x\in V_i$,
\item $c_{i}$ and $\VS_{\alpha_i,i}(0_{x})$ for each $x\in V_i$, and
\item $d_{i}$ and $\VS_{\ell+1,i}(3_{x})$ for each $x\in V_i$.
\end{itemize}

This concludes the construction of the graph $G^*$. Note that we bidirect each edge of $G^*$.
We now define the set of demand pairs:

 \begin{figure}

 \centering
 \begin{tikzpicture}[scale=0.5]

 \foreach \j in {0,1}
 {
 \begin{scope}[shift={(-21*\j,0)}]

 \foreach \j in {0,1,2,3}
 {
 \begin{scope}[shift={(0,\j)}]
 \foreach \i in {0,1,2,3}
 {
 \draw [black] plot [only marks, mark size=3, mark=*] coordinates {(12+2*\i,3)};
 }

 \foreach \i in {0,1,2}
 {
 \path (12+2*\i,3) node(a) {} (12+2*\i+2,3) node(b) {};
         \draw[ultra thick,black] (a) -- (b);
 }
 \end{scope}
 }

 \draw [red] plot [only marks, mark size=3, mark=*] coordinates {(14,2)} ;
 \draw [red] plot [only marks, mark size=3, mark=*] coordinates {(16,2)} ;
 \draw [red] plot [only marks, mark size=3, mark=*] coordinates {(14,7)} ;
 \draw [red] plot [only marks, mark size=3, mark=*] coordinates {(16,7)} ;

 \foreach \i in {0,1}
 {
 \begin{scope}[shift={(2*\i,0)}]
 \path (14,7) node(a) {} (14,6) node(b) {};
  \draw [thick,dotted] (a) to (b);

 \path (14,7) node(a) {} (14,5) node(b) {};
  \draw [thick,dotted] (a.north) to [out=60,in=60] (b.north);

 \path (14,7) node(a) {} (14,4) node(b) {};
  \draw [thick,dotted] (a.north) to [out=45,in=45] (b.north);

 \path (14,7) node(a) {} (14,3) node(b) {};
  \draw [thick,dotted] (a.north) to [out=45,in=45] (b.north);

 \path (14,2) node(a) {} (14,3) node(b) {};
  \draw [thick,dotted] (a) to (b);

 \path (14,2) node(a) {} (14,4) node(b) {};
  \draw [thick,dotted] (a.north) to [out=120,in=120] (b.south);

 \path (14,2) node(a) {} (14,5) node(b) {};
  \draw [thick,dotted] (a.north) to [out=135,in=135] (b.south);

 \path (14,2) node(a) {} (14,6) node(b) {};
  \draw [thick,dotted] (a.north) to [out=135,in=135] (b.south);
 \end{scope}
 }

 \end{scope}
 }

 \foreach \j in {0,1}
 {
 \begin{scope}[shift={(0,21*\j)}]

 \foreach \j in {0,1,2,3}
 {
 \begin{scope}[shift={(\j,0)}]
 \foreach \i in {0,1,2,3}
 {
 \draw [black] plot [only marks, mark size=3, mark=*] coordinates {(3,-3-2*\i)};
 }

 \foreach \i in {0,1,2}
 {
 \path (3,-3-2*\i) node(a) {} (3,-3-2*\i-2) node(b) {};
         \draw[ultra thick,black] (a) -- (b);
 }
 \end{scope}
 }

 \draw [red] plot [only marks, mark size=3, mark=*] coordinates {(2,-5)} ;
 \draw [red] plot [only marks, mark size=3, mark=*] coordinates {(2,-7)} ;
 \draw [red] plot [only marks, mark size=3, mark=*] coordinates {(7,-5)} ;
 \draw [red] plot [only marks, mark size=3, mark=*] coordinates {(7,-7)} ;

 \foreach \i in {0,1}
 {
 \begin{scope}[shift={(0,-2*\i)}]
 \path (7,-5) node(a) {} (6,-5) node(b) {};
  \draw [thick,dotted] (a) to (b);

 \path (7,-5) node(a) {} (5,-5) node(b) {};
  \draw [thick,dotted] (a.south) to [out=210,in=210] (b.south);

 \path (7,-5) node(a) {} (4,-5) node(b) {};
  \draw [thick,dotted] (a.south) to [out=210,in=210] (b.south);

 \path (7,-5) node(a) {} (3,-5) node(b) {};
  \draw [thick,dotted] (a.south) to [out=210,in=210] (b.south);

 \path (2,-5) node(a) {} (3,-5) node(b) {};
  \draw [thick,dotted] (a) to (b);

 \path (2,-5) node(a) {} (4,-5) node(b) {};
  \draw [thick,dotted] (a.north) to [out=135,in=135] (b.south);

 \path (2,-5) node(a) {} (5,-5) node(b) {};
  \draw [thick,dotted] (a.north) to [out=135,in=135] (b.south);

 \path (2,-5) node(a) {} (6,-5) node(b) {};
  \draw [thick,dotted] (a.north) to [out=135,in=135] (b.south);
 \end{scope}
 }

 \end{scope}
 }

 \foreach \j in {0,1,2}
 {
 \begin{scope}[shift={(0,1.5*\j)}]
 \foreach \i in {0,1,2,3}
 {
\draw [black] plot [only marks, mark size=3, mark=*] coordinates {(1.5+2*\i,3)};
 }

 \foreach \i in {0,1,2}
 {
 \path (1.5+2*\i,3) node(a) {} (1.5+2*\i+2,3) node(b) {};
         \draw[ultra thick,black] (a) -- (b);
 }
 \end{scope}
 }

 \draw [red] plot [only marks, mark size=3, mark=*] coordinates {(3.5,1.5)} ;
 \draw [red] plot [only marks, mark size=3, mark=*] coordinates {(5.5,1.5)} ;
 \draw [red] plot [only marks, mark size=3, mark=*] coordinates {(3.5,7.5)} ;
 \draw [red] plot [only marks, mark size=3, mark=*] coordinates {(5.5,7.5)} ;

 \foreach \i in {0,1}
 {
 \begin{scope}[shift={(2*\i,0)}]
 \path (3.5,7.5) node(a) {} (3.5,6) node(b) {};
  \draw [thick,dotted] (a) to (b);

 \path (3.5,7.5) node(a) {} (3.5,4.5) node(b) {};
  \draw [thick,dotted] (a.north) to [out=60,in=60] (b.north);

 \path (3.5,7.5) node(a) {} (3.5,3) node(b) {};
  \draw [thick,dotted] (a.north) to [out=45,in=45] (b.north);

 \path (3.5,1.5) node(a) {} (3.5,3) node(b) {};
  \draw [thick,dotted] (a) to (b);

 \path (3.5,1.5) node(a) {} (3.5,4.5) node(b) {};
  \draw [thick,dotted] (a.north) to [out=120,in=120] (b.south);

 \path (3.5,1.5) node(a) {} (3.5,6) node(b) {};
  \draw [thick,dotted] (a.north) to [out=135,in=135] (b.south);

 \end{scope}
 }

 \path (1.5,4.5) node(a) {} (-3,4) node(b) {};
  \draw [thick,red] (a) to (b);

 \path (7.5,4.5) node(a) {} (12,4) node(b) {};
  \draw [thick,red] (a) to (b);

 \path (1.5,4.5) node(a) {} (5,12) node(b) {};
  \draw [thick,red] (a) .. controls (-2,9) .. (b);

 \path (7.5,4.5) node(a) {} (5,-3) node(b) {};
  \draw [thick,red] (a) .. controls (11,0) .. (b);

 \draw [green] plot [only marks, mark size=0, mark=*] coordinates {(-6,1)}
node[label={[xshift=0mm,yshift=-7mm] \Large{$\VS_{i,j}$}}] {};

 \draw [green] plot [only marks, mark size=0, mark=*] coordinates {(15,1)}
node[label={[xshift=0mm,yshift=-7mm] \Large{$\VS_{\nextone_{i}(j),j}$}}] {};

 \draw [blue] plot [only marks, mark size=0, mark=*] coordinates {(4.5,-10)}
node[label={[xshift=0mm,yshift=-6mm] \Large{$\HS_{i,j}$}}] {};

 \draw [blue] plot [only marks, mark size=0, mark=*] coordinates {(4.5,19)}
node[label={[xshift=0mm,yshift=-4mm] \Large{$\HS_{i,\nextone_{j}(i)}$}}] {};

 \draw [black] plot [only marks, mark size=0, mark=*] coordinates {(2,0)}
node[label={[xshift=4mm,yshift=-1mm] \Large{$\M_{i,j}$}}] {};

 \draw [black] plot [only marks, mark size=0, mark=*] coordinates {(1.5,4.5)}
node[label={[xshift=-10mm,yshift=-2mm] \footnotesize{$\M_{i,j}(0_{x,y})$}}] {};

 \draw [black] plot [only marks, mark size=0, mark=*] coordinates {(7.5,4.5)}
node[label={[xshift=3mm,yshift=-1mm] \footnotesize{$\M_{i,j}(3_{x,y})$}}] {};

 \draw [black] plot [only marks, mark size=0, mark=*] coordinates {(5,12)}
node[label={[xshift=6mm,yshift=-7mm]
\footnotesize{$\HS_{i,\nextone_{j}(i)}(3_{x})$}}] {};

 \draw [black] plot [only marks, mark size=0, mark=*] coordinates {(-3,4)}
node[label={[xshift=6mm,yshift=-7mm] \footnotesize{$\VS_{i,j}(3_{y})$}}] {};

 \draw [black] plot [only marks, mark size=0, mark=*] coordinates {(5,-3)}
node[label={[xshift=-2mm,yshift=0mm] \footnotesize{$\HS_{i,j}(0_{x})$}}] {};

 \draw [black] plot [only marks, mark size=0, mark=*] coordinates {(12,4)}
node[label={[xshift=-5mm,yshift=-6mm]
\footnotesize{$\VS_{\nextone_{i}(j),j}(0_{y})$}}] {};

 \end{tikzpicture}

 \caption{A zoomed-in view of the main gadget $\M_{i,j}$ surrounded by four
secondary gadgets: vertical gadget $\HS_{i,\nextone_{j}(i)}$ on the top,
horizontal gadget $\VS_{i,j}$ on the left, vertical gadget $\HS_{i,j}$ on the
bottom and horizontal gadget $\VS_{\nextone_{i}(j),j}$ on the right. The main gadget $\M_{i,j}$ is a copy of the uniqueness gadget $U_{|E_{\{i,j\}}|}$ (see~\autoref{sec:uniqueness-gadget}).
The vertical gadgets $\HS_{i,j}$ and $\HS_{i,\nextone_{j}(i)}$ are copies of $U_{|V_i|}$. The horizontal gadgets $\VS_{i,j}$ and $\VS_{\nextone_{i}(j),j}$ are copies of $U_{|V_j|}$.
The only inter-gadget edges are the red
edges: they have one end-point in a main gadget and the other end-point in a
secondary gadget. We have shown four such red edges which are introduced for
every~$(x,y)\in E_{\{i,j\}}$.
\label{fig:bi-dsn-main-and-4-surrounding}}
\end{figure}

\begin{center}
\noindent\framebox{\begin{minipage}{\textwidth}
The set of demand pairs $\mathcal{D}$ is given by:
\begin{itemize}
\item \underline{Type I}: For each $j\in [\ell]$ and each $i\in N'_{H}(j)$, we add the following four pairs involving the vertices of the horizontal secondary gadget $\VS_{i,j}$:
        \begin{itemize}
          \item $\big(\M_{\prevone_{i}(j),j}(s_1), \VS_{i,j}(t_1)\big)$
          \item $\big(\M_{\prevone_{i}(j),j}(s_2), \VS_{i,j}(t_2)\big)$
          \item $\big(\VS_{i,j}(s_1), \M_{i,j}(t_1)\big)$
          \item $\big(\VS_{i,j}(s_2), \M_{i,j}(t_2)\big)$
        \end{itemize}

\item \underline{Type II}: For each $i\in [\ell]$ and each $j\in N'_{H}(i)$, we add the following four pairs involving the vertices of the vertical secondary gadget $\HS_{i,j}$:
        \begin{itemize}
          \item $\big(\M_{i,j}(s_1), \HS_{i,j}(t_1)\big)$
          \item $\big(\M_{i,j}(s_2), \HS_{i,j}(t_2)\big)$
          \item $\big(\HS_{i,j}(s_1), \M_{i,\prevone_{j}(i)}(t_1)\big)$
          \item $\big(\HS_{i,j}(s_2), \M_{i,\prevone_{j}(i)}(t_2)\big)$
        \end{itemize}

\end{itemize}
\end{minipage}}
\end{center}

We have $2k+2\ell$ vertical and horizontal secondary gadgets each, and we add
$O(1)$ demand pairs corresponding to each of these gadgets. Hence, the total
number of demand pairs is
\begin{equation}\label{eqn:D-bidsn}
|\mathcal{D}|=O(2k+2\ell) = O(k)
\end{equation}
since we can assume that $H$ is connected (\autoref{remark:H-is-connected-for-CSI}) which implies $k\geq \ell-1$.

Fix the budget $B^*= 4\ell+ (2k+2\ell)\cdot B + (2k+2\ell)\cdot B +
(2k+\ell)\cdot (B+4)$
where $B=7M= 7k^4$. The high-level intuition is the following: we need $4\ell$ (via \textcolor[rgb]{1.00,0.50,0.00}{orange} edges)
from the budget just to include one edge incident on each of the $4\ell$ border vertices (each of which is part of a demand pair). We have
$(2k+2\ell)$
horizontal and vertical secondary gadgets each. We argue that any solution for
\bidsn must satisfy the in-out property in each of the secondary gadgets,
and then invoke \autoref{lem:macro-uniqueness-gadget}. Finally, for each main
gadget, we again show that it must satisfy the in-out property and hence has
cost at least~$B$. %
However, here we show that we additionally need at least four red edges (which
have exactly one end-point in a main gadget) and hence the cost of any \bidsn
solution restricted to edges having at least one end-point in each main gadget
is at least $B+4$. Since we have $2k+\ell$ main gadgets, this completely uses up
the budget $B^*$.

We now prove the correctness of our reduction by showing that the instance $(G,H)$ of \csi is a YES instance if and only if there is a solution to the \bidsn instance ($G^*,\mathcal{D})$ with cost at most $B^*$. First we show the forward direction:

\begin{lem}
If the instance $(G,H)$ of \csi is a YES instance, then the instance $(G^*, \mathcal{D})$ of \bidsn has a solution of cost at most $B^*$.
\label{lem:bidsn-easy-csi}
\end{lem}
\begin{proof}
Suppose that the \pname{PSI} instance is a YES instance, i.e., there exists a
function $\phi:
V_H\rightarrow V_G$ such that
        \begin{enumerate}
          \item for every $i\in [\ell]$ we have $\phi(i)\in V_i$, and
          \item for every edge $ij\in E_H$ we have $\phi(i)\phi(j)\in E_G$.
        \end{enumerate}
We now show that there is an edge set $E'\subseteq E(G^*)$ of weight $B^*$ such that the network $\big(V(G^*),E'\big)$ is a solution for the \bidsn instance $(G^*,\mathcal{D})$. The edge set $E'$ consists of the following edges:
\begin{itemize}
  \item For each $1\leq i\leq \ell$, pick the edges $a_i \to
\HS_{i,\ell+1}(0_{\phi(i)})$ and $\HS_{i,\alpha_i}(3_{\phi(i)})\to
b_i$.
  \item For each $1\leq j\leq \ell$, pick the edges $c_j \to
\HS_{\alpha_j,j}(0_{\phi(j)})$ and $\HS_{\ell+1, j}(3_{\phi(j)})
\to d_j$.
  \item For each $1\leq i,j\leq \ell$ such that $ij\in E^*_H$ pick the
following edges (guaranteed to exist by \autoref{crl:macro}).
        \begin{itemize}
           \item The set of edges in $\M_{i,j}$ which is oriented rightwards,
represents $\{\phi(i), \phi(j)\}$ and has weight $M$. Note that this is possible since $\M_{i,j}$ has a row corresponding to each edge of $E_{\{i,j\}}$, and $\phi(i)\phi(j)\in E_{\{i,j\}}\subseteq E_G$ by $(1)$ and $(2)$.
        \end{itemize}
  \item For each $i\in [\ell], j\in N'_{H}(i) $ pick the following edges
(guaranteed to exist by \autoref{crl:macro}).
        \begin{itemize}
          \item The set of edges in $\HS_{i,j}$ which is oriented rightwards,
represents $\phi(i)$ and has weight $M$. Note that this is possible since $\HS_{i,j}$ has a row corresponding to each vertex of $V_{i}$, and $\phi(i)\in V(i)$ by $(1)$.
        \end{itemize}
    \item For each $j\in [\ell], i\in N'_{H}(j)$ pick the following edges
(guaranteed to exist by \autoref{crl:macro}).
        \begin{itemize}
          \item The set of edges in $\VS_{i,j}$ which is oriented rightwards,
represents $\phi(j)$ and has weight $M$. Note that this is possible since $\VS_{i,j}$ has a row corresponding to each vertex of $V_{j}$, and $\phi(j)\in V(j)$ by $(1)$.
        \end{itemize}

    \item For each $1\leq i,j\leq \ell$ such that $ij\in E^*_H$ pick the four
red edges connecting main gadgets and secondary gadgets
given by
        \begin{itemize}
           \item $\HS_{i,\nextone_{j}(i)}(3_{\phi(i)}) \to
\M_{i,j}(0_{\phi(i),\phi(j)})$,
           \item $\VS_{i,j}(3_{\phi(j)}) \to \M_{i,j}(0_{\phi(i),\phi(j)})$,
           \item $\M_{i,j}(3_{\phi(i),\phi(j)}) \to \HS_{i,j}(0_{\phi(i)})$, and
           \item $\M_{i,j}(3_{\phi(i),\phi(j)}) \to
\VS_{\nextone_{i}(j),j}(0_{\phi(j)})$.
         \end{itemize}
    Observe that these four red edges are guaranteed to exist since either $i\neq j$ which implies $\phi(i)\phi(j)\in
E_{\{i,j\}}\subseteq E_G$ by $(2)$ or otherwise $i=j$ which implies $\phi(i)\phi(i)\in
E_{\{i,i\}}$ by $(1)$.
\end{itemize}

It is easy to see that the cost of $E'$ is $4\ell+ (2k+2\ell)\cdot B +
(2k+2\ell)\cdot B + (2k+\ell)\cdot (B+4)=B^*$. We now show that each demand pair
of Type I is satisfied by the edge set $E'$.
\begin{itemize}
  \item The pair $(\M_{\prevone_{i}(j),j}(s_1), \VS_{i,j}(t_1))$ is satisfied by
the path $\M_{\prevone_{i}(j),j}(s_1)\rightarrow
\M_{\prevone_{i}(j),j}(1_{\phi(i),\phi(j)})\rightarrow
\M_{\prevone_{i}(j),j}(2_{\phi(i),\phi(j)})\rightarrow
\M_{\prevone_{i}(j),j}(3_{\phi(i),\phi(j)})\rightarrow
\VS_{i,j}(0_{\phi(j)})\rightarrow \VS_{i,j}(1_{\phi(j)})\rightarrow \VS_{i,j}(t_1)$
in~$N$, since $i=\nextone_{\prevone_{i}(j)}(j)$.
  \item The pair $(\M_{\prevone_{i}(j),j}(s_2), \VS_{i,j}(t_2))$ is satisfied by
the path $\M_{\prevone_{i}(j),j}(s_2)\rightarrow
\M_{\prevone_{i}(j),j}(2_{\phi(i),\phi(j)})\rightarrow
\M_{\prevone_{i}(j),j}(3_{\phi(i),\phi(j)})\rightarrow
\VS_{i,j}(0_{\phi(j)})\rightarrow \VS_{i,j}(1_{\phi(j)})\rightarrow
\VS_{i,j}(2_{\phi(j)}) \rightarrow \VS_{i,j}(t_2)$ in $N$, since
$i=\nextone_{\prevone_{i}(j)}(j)$.
  \item The pair $(\VS_{i,j}(s_1), \M_{i,j}(t_1))$ is satisfied by the path
$\VS_{i,j}(s_1)\rightarrow \VS_{i,j}(1_{\phi(j)})\rightarrow
\VS_{i,j}(2_{\phi(j)})\rightarrow \VS_{i,j}(3_{\phi(j)})\rightarrow
\M_{i,j}(0_{\phi(i),\phi(j)})\rightarrow \M_{i,j}(1_{\phi(i),\phi(j)})\rightarrow
\M_{i,j}(t_1)$ in $N$
  \item The pair $(\VS_{i,j}(s_2), \M_{i,j}(t_2))$ is satisfied by the path
$\VS_{i,j}(s_2)\rightarrow \VS_{i,j}(2_{\phi(j)})\rightarrow
\VS_{i,j}(3_{\phi(j)})\rightarrow \M_{i,j}(0_{\phi(i),\phi(j)})\rightarrow
\M_{i,j}(1_{\phi(i),\phi(j)})\rightarrow \M_{i,j}(2_{\phi(i),\phi(j)})\rightarrow
\M_{i,j}(t_2)$ in $N$
\end{itemize}

The proof of satisfiability for the demand pairs of Type II is analogous.
Hence, the network $\big(V(G^*),E'\big)$ is indeed a solution for the \bidsn instance.
\end{proof}

Our next lemma shows the reverse direction of the correctness of the reduction: the existence of a solution of small cost for the \bidsn instance implies a solution for the \csi instance.

\begin{lem}
\label{lem:bidsn-hard-csi}
If the instance $(G^*, \mathcal{D})$ of \bidsn has a solution of cost at most $B^*$, then the instance $(G,H)$ of \csi is a YES instance.
\end{lem}
\begin{proof}

The arguments here are almost identical to those from~\autoref{lem:bidsnp-hard-gt}.
Suppose that \bidsn has a solution, say the network $\big(V(G^*),N\big)$, of cost at most $B^* = 4\ell+
(2k+2\ell)\cdot B + (2k+2\ell)\cdot B + (2k+\ell)\cdot (B+4)$.

\begin{claim}
For any $j\in [\ell], i\in N'_{H}(j)$ the edges in $N$ which have both
end-points in the horizontal secondary gadget $\VS_{i,j}$ satisfy the in-out
property. Hence, $\VS_{i,j}$ uses up weight of at least $B$ from the budget.
\label{lem:in-out-vertical-general}
\end{claim}
\begin{proof}
Looking at the demand pairs in $\mathcal{D}$ of Type I, we observe that
\begin{itemize}
\item $\VS_{i,j}(s_1)$ is the source of some demand pair whose other end-point
lies outside of $\VS_{i,j}$,
\item $\VS_{i,j}(s_2)$ is the source of some demand pair whose other end-point
lies outside of $\VS_{i,j}$,
\item $\VS_{i,j}(t_1)$ is the target of some demand pair whose other end-point
lies outside of $\VS_{i,j}$,
\item $\VS_{i,j}(t_2)$ is the target of some demand pair whose other end-point
lies outside of $\VS_{i,j}$.
\end{itemize}
Since the network $\big(V(G^*),N\big)$ is a solution of the \bidsn instance, it follows that there is a
path in $N$ starting at $\VS_{i,j}(s_1)$ which must leave the gadget $\VS_{i,j}$, i.e.,
$\VS_{i,j}(s_1)$ can reach either a $0$-vertex or a $3$-vertex. The other three
conditions of \autoref{defn-in-out} follow by similar reasoning. By
\autoref{lem:macro-uniqueness-gadget}, it follows that the edges in $N$ which
have both endpoints in $\VS_{i,j}$ use up weight
of at least $B$ from the budget.
\end{proof}

Analogous claims hold also for vertical secondary gadgets and main gadgets:
\begin{claim}
For any $i\in [\ell], j\in N'_{H}(i)$ the edges in $N$ which have both
end-points in the vertical secondary gadget $\HS_{i,j}$ satisfy the
in-out property. Hence, $\HS_{i,j}$ uses up weight of at least $B$ from the
budget.
\label{lem:in-out-horizontal-general}
\end{claim}

\begin{claim}
For every $i,j$ such that $ij\in E^*_H$ the edges in $N$ which have both
end-points in the main gadget $\M_{i,j}$ satisfy the in-out property. Hence,
$\M_{i,j}$ uses up weight of at least $B$ from the budget.
\label{lem:in-out-main-general}
\end{claim}

From \autoref{lem:in-out-vertical-general},
\autoref{lem:in-out-horizontal-general} and
\autoref{lem:in-out-main-general} we know that each of the gadgets (horizontal
secondary, vertical secondary and main) use up at least weight $B$ in $N$. We
now claim that the edge set $N$ restricted to each vertical secondary gadget, horizontal
secondary gadget and main gadget has weight exactly $B$. Suppose there is at
least one gadget where the edges of $N$ have weight more than $B$. By
\autoref{lem:macro-uniqueness-gadget}, the weight of edges of $N$ in this gadget
is at
least $8M= B+M$. Since $M=k^4$ and $B^* = 4\ell+ (2k+2\ell)\cdot B +
(2k+2\ell)\cdot B +
(2k+\ell)\cdot (B+4)$, where $B=7M$, the weight of $N$ is at least
\begin{align*}
\Big((2k+2\ell)\cdot B +(2k+2\ell)\cdot B &+ (2k+\ell)\cdot B \Big)+ M\\  
&=\Big((2k+2\ell)\cdot B +(2k+2\ell)\cdot B + (2k+\ell)\cdot B \Big)+ k^4 \\
&> \Big((2k+2\ell)\cdot B +(2k+2\ell)\cdot B + (2k+\ell)\cdot B \Big) + 24k \\
&\geq \Big((2k+\ell)\cdot B +(2k+\ell)\cdot B + (2k+\ell)\cdot B \Big) + 4\ell +
4(2k+\ell)\\
&= B^*,
\end{align*}
which is a contradiction (since $k^4> 24k$ for $k\geq 3$, and $k\geq 
\ell-1$ since by Remark~\ref{remark:H-is-connected-for-CSI} we can assume that 
$H$ is connected).

Therefore, we have the following claim which follows from
\autoref{lem:macro-uniqueness-gadget}:

\begin{claim}
The weight of $N$ restricted to each gadget is exactly $B=7M$. Moreover,
\begin{itemize}

\item for each $j\in [\ell], i\in N'_{H}(j)$, the horizontal secondary gadget
$\VS_{i,j}$ is represented by some $y_{i,j}\in V_j$ and is either right-oriented
or left-oriented,

\item for each $i\in [\ell], j\in N'_{H}(i)$, the vertical secondary
gadget $\HS_{i,j}$ is represented by some $x_{i,j}\in V_i$ and is either
right-oriented or left-oriented, and

\item for each $i,j$ such that $ij\in E^*_H$, the main gadget $\M_{i,j}$ is
represented by
some $(\lambda_{i,j}, \delta_{i,j})\in E_{i,j}$ and is either right-oriented or
left-oriented.
\end{itemize}
\label{thm:representing}
\end{claim}

We now show that for each $1\leq i,j\leq \ell$ such that $ij\in E^*_H$, the edge set $N$ must
also contain some \textcolor[rgb]{1.00,0.00,0.00}{red} edges which have
exactly one end-point among vertices of $\M_{i,j}$.

\begin{restatable}{claim}{lemmainatleastfourred}
For each $1\leq i,j\leq \ell$ such that $ij\in E^*_H$, the edge set $N$ must
also contain at least one
\textcolor[rgb]{1.00,0.00,0.00}{red} edge of each of the following four types:
\begin{itemize}
  \item an edge with one end-point in the set of $3$-vertices of $\M_{i,j}$ and other end-point in the set of $0$-vertices of $\HS_{i,j}$
  \item an edge with one end-point in the set of $3$-vertices of $\M_{i,j}$ and other end-point in the set of $0$-vertices of $\VS_{\nextone_{i}(j),j}$
  \item an edge with one end-point in the set of $0$-vertices of $\M_{i,j}$ and other end-point in the set of $3$-vertices of $\VS_{i,j}$
  \item  an edge with one end-point in the set of $0$-vertices of $\M_{i,j}$ and other end-point in the set of $3$-vertices of $\HS_{i,\nextone_{j}(i)}$
\end{itemize}
\label{lem:main-at-least-4-red}
\end{restatable}
\begin{proof}
We show that $N$ must use at least one red edge which has one end-point in
the set of $3$-vertices of $\M_{i,j}$ and the other end-point in the set of
$0$-vertices of $\HS_{i,j}$. Analogous arguments hold for the other three
secondary gadgets surrounding the main gadget $\M_{i,j}$ and hence we get the
lower bound of four red edges as claimed (note that the vertex sets of the
secondary gadgets are pairwise disjoint, and hence these edges are distinct).

We know by \autoref{thm:representing} that $\HS_{i,j}$ is either
right-oriented or left-oriented. Suppose $\HS_{i,j}$ is right-oriented (the case when $\HS_{i,j}$ is left-oriented is analogous). Also, by
\autoref{lem:macro-uniqueness-gadget} and ~\autoref{thm:representing}, we
know that the only base edges of $\HS_{i,j}$ picked in $N$ are
$\HS_{i,j}(0_{x_{i,j}})\rightarrow \HS_{i,j}(1_{x_{i,j}})\rightarrow
\HS_{i,j}(2_{x_{i,j}})\rightarrow \HS_{i,j}(3_{x_{i,j}})$ and the only connector
edges of $\HS_{i,j}$ picked in $E^*$ are $\HS_{i,j}(s_1)\rightarrow
\HS_{i,j}(1_{x_{i,j}}), \HS_{i,j}(s_2)\rightarrow \HS_{i,j}(2_{x_{i,j}}),
\HS_{i,j}(t_1)\leftarrow \HS_{i,j}(1_{x_{i,j}})$ and $\HS_{i,j}(t_2)\leftarrow
\HS_{i,j}(2_{x_{i,j}})$. But there is a Type II demand pair whose target is
$\HS_{i,j}(t_1)$: the path satisfying this demand pair has to enter $\HS_{i,j}$
through the vertex $\HS_{i,j}(0_{x_{i,j}})$. The only edges (which are not base
or connector edges of $\HS_{i,j}$) incident on the $0$-vertices of $\HS_{i,j}$ have
their other end-point in the set of $3$-vertices of $\M_{i,j}$. That is, $N$
contains a
red edge whose start vertex is a $3$-vertex of $\M_{i,j}$ and end vertex is a
$0$-vertex of $\HS_{i,j}$.
\end{proof}

\begin{restatable}{claim}{lemdotted}
For each $1\leq i,j\leq \ell$, the edge set $N$ must contain at least one \textcolor[rgb]{1.00,0.50,0.00}{orange} edge of each of the following types:
\begin{itemize}
\item an outgoing edge from $a_i$,
\item an incoming edge into $b_i$,
\item an outgoing edge from $c_j$, and
\item an incoming edge into $d_j$.
\end{itemize}
\label{lem:dotted}
\end{restatable}
\begin{proof}
Note that $a_i$ is the source of two Type II demand pairs, viz.\ $\big(a_i,
\HS_{i,\ell+1}(t_1)\big)$ and $\big(a_i, \HS_{i,\ell+1}(t_2)\big)$. Hence $N$ must contain at
least one outgoing edge from $a_i$. The other three statements follow similarly.
\end{proof}

We show now that we have no slack, i.e., the weight of $N$ must be exactly $B^*$.

\begin{claim}
The weight of $N$ is exactly $B^*$, and is inclusion-wise minimal.
\label{thm:exact-B^*}
\end{claim}
\begin{proof}
From \autoref{thm:representing}, we know that each gadget has a weight of $B$ in
$N$. \autoref{lem:main-at-least-4-red} says that each main gadget contributes at
least 4 red edges which have exactly one end-point in this main gadget, and
\autoref{lem:dotted} says that the \textcolor[rgb]{1.00,0.50,0.00}{orange} edges (incident on border vertices)
contribute at least $4\ell$ to weight of $N$. Since all these edges are
distinct, we have that the weight of $N$ is at least $(2k+2\ell)\cdot B +
(2k+2\ell)\cdot B + (2k+\ell)\cdot B + 4(\ell+2k+\ell) = B^*$. Hence, the weight
of $N$ is exactly $B^*$, and it is minimal (under edge deletions) since no
edges have zero weights.
\end{proof}

Consider a main gadget $\M_{i,j}$. It has four secondary gadgets surrounding it:
$\HS_{i,j}$ below it, $\HS_{i,\nextone_{j}(i)}$ above it, $\VS_{i,j}$ to the left
and $\VS_{\nextone_{i}(j),j}$ to the right. By \autoref{thm:representing}, these
gadgets are represented by $x_{i,j}, x_{i,\nextone_{j}(i)}, y_{i,j}$, and
$y_{\nextone_{i}(j),j}$ respectively. The main gadget $\M_{i,j}$ is represented
by $(\lambda_{i,j}, \delta_{i,j})$.

\begin{claim}[propagation]
For every main gadget $\M_{i, j}$, we have
$x_{i,j}=\lambda_{i,j}=x_{i,\nextone_{j}(i)}$ and
$y_{i,j}=\delta_{i,j}=y_{\nextone_{i}(j),j}$.
\label{lem:agreement-tight-dsn}
\end{claim}
\begin{proof}
Due to symmetry, it suffices to only argue that $x_{i,j}=\lambda_{i,j}$. Let us
assume for the sake of contradiction that $x_{i,j}\ne\lambda_{i,j}$. We will now
show that there is a vertex such that (1) there is exactly one edge adjacent to
it in the edge set $N$ and (2) it does not belong to any demand pair. Observe
that removing its only adjacent edge from $N$ does not effect the validity of
the solution $\big(V(G^*),N\big)$ for the instance $(G^*,\mathcal{D})$ of \bidsn. This contradicts \autoref{thm:exact-B^*} which states that $N$ is minimal.

From \autoref{lem:main-at-least-4-red} and \autoref{thm:exact-B^*}, it follows
that $N$ contains exactly one red edge, say $e_1$, which has one end-point in
the set of $3$-vertices of $\M_{i,j}$ and other end-point in set of $0$-vertices
of $\HS_{i,j}$. Also, $N$ contains exactly one red edge, say $e_2$, which has one
end-point in the set of $3$-vertices of $\M_{i,j}$ and other end-point in set of
$0$-vertices of $\VS_{\nextone_{i}(j),j}$.

Observe that if $e_1$ does not have one endpoint at $\HS_{i, j}(0_{x_{i, j}})$,
then the vertex $\HS_{i, j}(0_{x_{i, j}})$ is the desired vertex. Hence, suppose
that one endpoint of the edge $e_1$ is $\HS_{i, j}(0_{x_{i, j}})$. The other
endpoint of $e_1$ must be $\M_{i, j}(3_{x_{i, j}, y})$ for some $y \in V_j$.
Since $x_{i, j} \ne \lambda_{i, j}$, we have $\M_{i, j}(3_{x_{i, j}, y}) \ne
\M_{i, j}(3_{\lambda_{i, j}, \delta_{i, j}})$. Suppose that one of the end-points
of $e_2$ is $\M_{i, j}(3_{x', y'})$. Since $\M_{i, j}(3_{x_{i, j}, y}) \ne \M_{i,
j}(3_{\lambda_{i, j}, \delta_{i, j}})$, at least one of the following must be
true: $\M_{i, j}(3_{x_{i, j}, y}) \ne \M_{i, j}(3_{x', y'})$ or $\M_{i,
j}(3_{\lambda_{i, j}, \delta_{i, j}}) \ne \M_{i, j}(3_{x', y'})$. If $\M_{i,
j}(3_{x_{i, j}, y}) \ne \M_{i, j}(3_{x',y'})$, then
$\M_{i, j}(3_{x_{i, j}, y})$ is the desired vertex. Otherwise, if $\M_{i,
j}(3_{\lambda_{i, j}, \delta_{i, j}}) \ne \M_{i, j}(3_{x', y'})$, then $\M_{i,
j}(3_{\lambda_{i, j}, \delta_{i, j}})$ is the desired vertex.
In all cases, we have found a vertex with desired properties; hence, we have
arrived at a contradiction.
\end{proof}

By \autoref{lem:dotted} and \autoref{thm:exact-B^*}, for each $i\in [\ell]$ the
edge set $N$ has exactly one incoming edge into~$b_i$. Define the function
$\phi: V_H\rightarrow V_G$ given by $\phi(i)=z$ if the unique edge in $N$ coming
into $b_i$ is from the vertex $\HS_{i,\alpha_i}(3_{z})$. By construction of
$G^*$, the vertical secondary gadget $\HS_{i,\alpha_i}$ is a copy of $U_{|V_i|}$ and
hence $\phi(i)\in V_i$. Since $b_i$ is the target of a terminal pair of Type II, we have that $\phi(i)=x_{i,\alpha_{i}}$.

Similarly, by \autoref{lem:dotted} and \autoref{thm:exact-B^*}, for each $j\in
[\ell]$ the edge set $N$ has exactly one outgoing edge from $c_j$. Define the
function $\psi: V_H\rightarrow V_G$ given by $\psi(j)=s$ if the unique edge in
$N$ coming out from $c_j$ is into the vertex $\VS_{\alpha_j,j}(0_{s})$. By
construction of $G^*$, the horizontal secondary gadget $\VS_{\alpha_j,j}$ is a
copy
of $U_{|V_j|}$ and hence $\psi(j)\in V_j$. Since $c_j$ is the target of a terminal pair of Type II, we have that $\psi(j)=y_{\alpha_{j},j}$.

By \autoref{lem:agreement-tight-dsn}, it follows that for each
$1\leq i,j\leq \ell$ such that $ij\in E^*_H$ we have
$x_{i,j}=\lambda_{i,j}=x_{i,\nextone_{j}(i)}$ and
$y_{i,j}=\delta_{i,j}=y_{\nextone_{i}(j),j}$ in addition to $(\lambda_{i,j},
\delta_{i,j})\in E_{\{i,j\}}$ (by the definition of the main gadget).

Therefore, we have that $\phi(i)=x_{i,j}$ for each $j\in N'_{H}(i)$ and
$\psi(j)=y_{i,j}$ for each $i\in N'_{H}(j)$. For each $i\in [\ell]$ the main
gadget $\M_{i,i}$ is a copy of $U_{|E_{\{i,i\}}|}$. Hence, we have that
$\phi(i)=x_{i,i}=\lambda_{i,i}=\delta_{i,i}=y_{i,i}=\psi(i)$ for each $i\in [\ell]$. Now consider any
$1\leq i,j\leq \ell$ such that $ij\in E_H$. We know that $\VS_{i,j}, \HS_{i,j}$
are represented by $\phi(j), \phi(i)$ respectively. Since
$\lambda_{i,j}=x_{i,j}=\phi(i), \delta_{i,j}=y_{i,j}=\phi(j)$ and
$(\lambda_{i,j}, \delta_{i,j})\in E_{i,j}\subseteq E_G$ it follows that
$\phi(i)\phi(j)\in E_G$, i.e., the instance $(G,H)$ of \csi is a YES instance.
This concludes the proof of~\autoref{lem:bidsn-hard-csi}.
\end{proof}

Finally we are ready to prove~\autoref{thm:lb-biDSN} which is restated below:
\thmlbbiDSN*
\begin{proof}
Each main gadget has $O(n^2)$ vertices and $G^*$ has $2k+\ell=O(k)$ main gadgets, where $n=|V_G|$ and $k=|E_H|$. Each secondary gadget has $O(n)$ vertices and $G^*$ has $2k+2\ell=O(k)$ secondary gadgets. The number of border vertices in $G^*$ is $4k$. Hence, the total number of vertices in the graph $G^*$ is $O(n^{2}k)=\poly(n,k)$. It is easy to see that $G^*$ can be constructed in $\poly(n,k)$ time from a given instance $(G,H)$ of \csi. Combining the two directions from~\autoref{lem:bidsn-easy-csi} and~\autoref{lem:bidsn-hard-csi}, we get a parameterized reduction from \csi to an instance of \bidsn with $|\mathcal{D}|=O(k)$ terminal pairs (\autoref{eqn:D-bidsnP}).

Hence, the W[1]-hardness of \bidsn parameterized by the number $k$ of demand pairs follows from the W[1]-hardness of \csi parameterized by $|E_H|$ (observe that the W[1]-hard problem \mcc~\cite{mcc-1,mcc-2} is a special case of \csi when $H$ is a clique). Marx~\cite[Corollary 6.3]{marx-beat-treewidth} showed that assuming ETH, the \csi problem cannot be solved in time $f(|E_H|)\cdot |V_G|^{o\big(|E_H|/\log |E_H|\big)}$ for any computable function $f$.
Hence, it follows that under ETH, the \bidsn cannot be solved in $f(k)\cdot n^{o(k/\log k)}$ time for any computable function $f$.
\end{proof}

Note that \autoref{thm:lb-biDSN} also implies that there is no \emph{efficient}
parameterized approximation scheme, i.e.\ an algorithm computing a
$(1+\eps)$-approximate solution in time $f(k,\eps)\cdot n^{O(1)}$ for some
function~$f$. This is because the hardness result holds for unweighted graphs in
which the optimum solution has cost at most $B^*= 4\ell+ (2k+2\ell)\cdot B +
(2k+2\ell)\cdot B + (2k+\ell)\cdot (B+4) = O(k^5)$ if and only if the
\csi instance is a YES instance. Consequently, if an efficient parameterized
approximation scheme existed we could set $\eps$ to a value such that
$(1+\eps)\cdot B^*<B^*+1$ with $\eps$ being a function of the parameter $k$
independent of $n$. Every edge of our constructed graph $G^*$ has weight at
least $1$, and hence an $(1+\epsilon)$-approximation is in fact forced to find a
solution of cost at most $B^*$, i.e., it finds an optimum solution. Hence, we obtain the following corollary:

\begin{crl}
\label{crl:bi-DSN}
The \bidsn problem has the following two lower bounds for efficient approximation schemes:
\begin{itemize}
  \item Unless \textup{FPT=W[1]}, there is no $(1+\eps)$-approximation running in time $f(\eps,k)\cdot n^{O(1)}$ for any function~$f$.
  \item Under ETH, there is no $(1+\eps)$-approximation running in time $f(\eps,k)\cdot n^{o(k/\log k)}$ for any function~$f$.
\end{itemize}
\end{crl}

\subsection{NP-hardness and runtime lower bound for \altbiscss}
\label{sec:np-hard}

In this section we prove \autoref{thm:lb-biSCSS}, which is restated below.
\thmlbbiSCSS*
\begin{proof}
We reduce from the NP-hard \hamc problem. Given an undirected unweighted graph
$G$ on $n$ vertices as an instance to \hamc, we construct a bidirected weighted
complete graph $H$ on the same vertex set as $G$ as follows:
\begin{itemize}
\item If $\{u,v\}$ is an edge of $G$, then we set the weight of $uv$ and $vu$
in $H$ to $1$.

\item If $\{u,v\}$ is not an edge of $G$, then we set the weight of $uv$ and
$vu$ in $H$ to $2$.
\end{itemize}
Consider the \biscss instance on $H$ where every vertex is a terminal. We now
show that $G$ has a Hamiltonian cycle if and only if the \biscss instance has a
solution of cost $n$.

Suppose $G$ has a Hamiltonian cycle. It corresponds to a directed cycle in
$H$ of cost $n$ and is a feasible solution for the \biscss instance. On
the other hand, note that every \biscss solution in $H$ will have cost at least
$n$, since every vertex has out-degree at least one in the solution. Hence, if
there is a solution $N\subseteq H$ for the \biscss instance of cost exactly
$n$, then every vertex has out-degree exactly one in $N$, and each edge in
$N$ will have cost one. As $N$ is strongly connected, this means $N$ is a
directed cycle. As this cycle consists of only edges of cost one, it follows
that each edge in $N$ is also an edge in $G$, i.e., the underlying undirected
cycle $\ud{N}$ is a Hamiltonian cycle in $G$.

Finally, observe that we have shown above that \biscss with $k=|V|$ terminals
can solve the Hamiltonian cycle problem. It is known~\cite[Theorem
14.6]{pc-book} that under ETH the \hamc problem has no $2^{o(n)}\polyn$
algorithm. This immediately implies that \biscss does not have an
$2^{o(k)}\polyn$ algorithm under ETH.
\end{proof}

\newcommand{\cA}{\mathcal{A}}
\newcommand{\cB}{\mathcal{B}}
\newcommand{\cD}{\mathcal{D}}
\newcommand{\cE}{\mathcal{E}}
\newcommand{\cH}{\mathcal{H}}
\newcommand{\cZ}{\mathcal{Z}}
\newcommand{\cP}{\mathcal{P}}
\newcommand{\bA}{\mathbb{A}}
\newcommand{\bB}{\mathbb{B}}
\newcommand{\uni}{\text{unique}}

\section{FPT inapproximability of {\altscss} and \altbidsn}
\label{sec:inapprox}

The starting point of our inapproximability results are based on the
recent parameterized inapproximability of \pname{Densest $k$-Subgraph}
from~\cite{param-inapprox} (which in turn builds on a construction
from~\cite{M17}). To state the result precisely, let us first state the
underlying assumption, the Gap Exponential Time Hypothesis (Gap-ETH). Note
that the version used here rules out not only deterministic but also randomized
algorithms; this is needed for the inapproximability result
of~\cite{param-inapprox}.

\begin{hyp}[(Randomized) Gap-ETH~\cite{Dinur16,MR17}] \label{hyp:gap-eth}
There exists a constant $\delta > 0$ such that, given a \pname{3CNF}
formula $\Phi$ on $n$ variables, no (possibly
randomized) $2^{o(n)}$-time algorithm can distinguish between the following two
cases correctly with probability at least 2/3:
\begin{itemize}
\item $\Phi$ is satisfiable.
\item Every assignment to the variables violates at least a $\delta$-fraction
of the clauses of $\Phi$.
\end{itemize}
\end{hyp}

Here we do not attempt to reason why Gap-ETH is a plausible assumption; for
more detailed discussions on the topic, please refer
to~\cite{Dinur16,param-inapprox}. For now, let us move on to state the
inapproximability result from~\cite{param-inapprox} that we need. Recall that,
in the \pname{Densest $k$-Subgraph (D$k$S)} problem~\cite{KortsarzP93}, we are
given an undirected graph $G = (V, E)$ and an integer $k$ and we are asked to
find a subset $S \subseteq V$ of size $k$ that induces as many edges in $G$ as
possible. \citet{param-inapprox} showed that, even when parameterized by~$k$,
the problem is hard to approximate to within a $k^{o(1)}$-factor, as stated
more formally below.

\begin{thm}[{\cite[Lemma~5.21]{param-inapprox}}] \label{thm:inapprox-dks}
Assuming randomized Gap-ETH, for any function \mbox{$h(k) = o(1)$}, there is no
$f(k) \polyn$-time algorithm that, given a graph $G$ on $n$ vertices and an
integer~$k$, can distinguish between the following two cases:
\begin{itemize}
\item (YES) $G$ contains at least one $k$-clique as a subgraph.
\item (NO) Every $k$-subgraph of $G$ contains less than $k^{-h(k)} \cdot
\binom{k}{2}$ edges.
\end{itemize}
\end{thm}

Instead of working with \pname{D$k$S}, it will be more convenient for us to work with a closely-related problem called \mcsi, which can be defined as follows.

\begin{center}
\noindent\framebox{\begin{minipage}{0.9\textwidth}
\textbf{\mcsi} (\pname{MPSI})\\
\emph{Input}: an instance $\Gamma$ of \pname{MPSI} consists of three
components:
\begin{itemize}
\item an undirected graph $G = (V_G, E_G)$,
\item a partition of vertex set $V_G$ into disjoint subsets $V_1, \dots,
V_\ell$, and
\item an undirected graph $H = (V_H = \{1, \dots, \ell\}, E_H)$.
\end{itemize}
\emph{Goal}: find an assignment $\phi: V_H \to V_G$ where $\phi(i) \in V_i$ for
every $i \in [\ell]$ that maximizes the number of edges $ij \in E_H$ such that
$\phi(i)\phi(j) \in E_G$.
\end{minipage}}
\end{center}

It is worth noting that this problem is closely related to the so-called \pname{Label Cover} problem that appears in the the hardness of approximation literature~\cite{Moshkovitz15}\footnote{The \pname{Label Cover} problem may be viewed as a special case of \pname{MPSI} with two additional constraints on the input: the graph $G$ is bipartite, and every vertex on the left hand side of $G$ has at most one edge to each partition $V_i$.}. However, we choose
the name \mcsi as it is more compatible with the naming conventions earlier in
\autoref{sec:lb}; our new problem is simply an optimization version of \csi
defined in that section.

The graph $H$ is sometimes referred to as the \emph{supergraph} of $\Gamma$.
Similarly, the vertices and edges of $H$ are called \emph{supernodes} and
\emph{superedges} of $\Gamma$. Moreover, the size of $\Gamma$ is defined as $n =
|V_G|$, the number of vertices of $G$. Additionally, for each assignment $\phi$,
we define its value $\val(\phi)$ to be the fraction of superedges $ij \in E_H$
such that $\phi(i)\phi(j) \in E_G$; such superedges are said to be
\emph{covered} by $\phi$. The objective of \pname{MPSI} is now to find an
assignment $\phi$ with maximum value. We denote the value of the optimal
assignment by $\val(\Gamma)$, i.e., $\val(\Gamma) = \max_\phi \val(\phi)$.

For this problem, a hardness similar to that of \pname{Densest $k$-Subgraph} can be shown:

\begin{crl} \label{crl:inapprox-colored-dks}
Assuming randomized Gap-ETH, for any function $h(\ell) = o(1)$, there is no
$f(\ell)\polyn$-time algorithm that, given an \pname{MPSI} instance $\Gamma$ of
size $n$ such that the supergraph $H$ is a complete graph on $\ell$ supernodes,
can distinguish between the following two cases:
\begin{itemize}
\item (YES) $\val(\Gamma) = 1$.
\item (NO) $\val(\Gamma) < \ell^{-h(\ell)}$.
\end{itemize}
\end{crl}

The proof of \autoref{crl:inapprox-colored-dks} is rather simple, and
follows the standard technique of using splitters. Nevertheless, for
completeness, we give the full proof below. Before we proceed, we remark that our reduction is not the same as the \mcc hardness reduction from \textsc{Clique}~\cite{mcc-1,mcc-2}. Recall that the reduction in~\cite{mcc-1,mcc-2} simply makes each partition a copy of the original vertex set and add an edge between two new vertices iff there is an edge between the two corresponding vertices in the original graph. In this reduction, even in the NO case, we may take an edge $(u, v)$ of the original graph and then select $\lceil k/2 \rceil$ copies of $u$ together with $\lfloor k/2 \rfloor$ copies of $v$ in the new graph. This means that the gap between the YES and the NO cases in~\cite{mcc-1,mcc-2} can be at most two. Since we require a super constant gap in \autoref{crl:inapprox-colored-dks}, we need a different reduction.

\begin{dfn}\label{defn:splitters}
(\textbf{splitters}) Let $n \geq r \geq k$. An $(n,k,r)$-splitter is a family $\Lambda$ of functions $[n]\rightarrow [r]$ such that for every subset $S\subseteq [n]$ of size $k$ there is a function $\lambda \in \Lambda$ such that $\lambda$ is injective on $S$.
\end{dfn}

The following constructions of special families of splitters are due to
\citet{color-coding} and \citet{DBLP:conf/focs/NaorSS95}.

\begin{thm}[\cite{color-coding,DBLP:conf/focs/NaorSS95}]
\label{thm:color-coding}
There exists a $2^{O(k)} \cdot n^{O(1)}$-time algorithm that takes in $n, k \in \mathbb{N}$ such that $n \geq k$ and outputs an $(n,k,k)$-splitter family of functions $\Lambda_{n, k}$ such that $|\Lambda_{n, k}| = 2^{O(k)}\cdot \log n$.
\end{thm}

\begin{proof}[Proof of \autoref{crl:inapprox-colored-dks}]
Suppose for the sake of contradiction that there exists an algorithm
$\mathbb{B}$ that can solve the distinguishing problem stated in
\autoref{crl:inapprox-colored-dks} in $f(\ell) \cdot n^{O(1)}$ time for some
function $f$. We will use this to construct another algorithm $\mathbb{B}'$ that
can solve the distinguishing problem stated in \autoref{thm:inapprox-dks} in
time $f'(k) \cdot n^{O(1)}$ for some function $f'$, which will thereby violate
Gap-ETH.

The algorithm $\mathbb{B}'$, on input $(G, k)$, proceeds as follows. We assume
w.l.o.g. that $V = [n]$. First, $\mathbb{B}'$ runs the algorithm from
\autoref{thm:color-coding} on $(n, k)$ to produce an $(n,k,k)$-splitter
family of functions $\Lambda_{n, k}$. For each $\lambda \in \Lambda_{n, k}$, it
creates a \pname{MPSI} instance $\Gamma^\lambda = (G^\lambda, H^\lambda,
V^\lambda_1 \cup \cdots \cup V^\lambda_k)$ where
\begin{itemize}
\item the graph $G^\lambda$ is simply the input graph $G$,
\item for each $i \in [k]$, we set $V^\lambda_i = \lambda^{-1}(\{i\})$, and,
\item the supergraph $H^\lambda$ is simply the complete graph on $[k]$, i.e.,
$H^\lambda = ([k], \binom{[k]}{2})$.
\end{itemize}
Then, it runs the given algorithm $\mathbb{B}$ on $\Gamma^\lambda$. If
$\mathbb{B}$ returns YES for some $\lambda \in \Lambda$, then $\mathbb{B}'$
returns YES. Otherwise, $\mathbb{B}'$ outputs NO.

It is obvious that the running time of $\mathbb{B}'$ is at most $O(2^{O(k)}
f(k) \cdot n^{O(1)})$. Moreover, if $G$ contains a $k$-clique, say $(v_1,
\dots, v_k)$, then by the properties of splitters we are guaranteed that there
exists $\lambda \in \Lambda_{n, k}$ such that $\lambda(\{v_1, \dots, v_k\}) =
[k]$.  Hence, the assignment $i \mapsto v_i$ covers all superedges in
$E_{H^\lambda}$, implying that $\mathbb{B}$ indeed outputs YES on such
$\Gamma^\lambda$. On the other hand, if every $k$-subgraph of $G$ contains less
than $k^{-h(k)} \cdot \binom{k}{2}$, then, for any $\lambda \in \Lambda_{n, k}$
and any assignment $\phi$ of $\Gamma^\lambda$, $(\phi(1), \dots, \phi(k))$
induces less than $k^{-h(k)} \cdot \binom{k}{2}$ edges in $G$. This also upper
bounds the number of superedges covered by $\phi$, which implies that
$\Gamma^\lambda$ is a NO instance of \autoref{crl:inapprox-colored-dks}. Thus,
in this case, $\mathbb{B}$ outputs NO on all $\Gamma^\lambda$'s. In other words,
$\mathbb{B}'$ can correctly distinguish the two cases in
\autoref{thm:inapprox-dks} in $O(2^{O(k)} f(k) \cdot n^{O(1)})$ time. This
concludes our proof of \autoref{crl:inapprox-colored-dks}.
\end{proof}

With the parameterized hardness of approximating \pname{MPSI} ready, we can now
prove our hardness results for \scss and \bidsn, starting with the former.

\subsection{Strongly Connected Steiner Subgraph}

Our proof of the parameterized inapproximability of \scss is based on a
reduction from \mcsi whose properties are
described below.

\begin{lem} \label{lem:scss-red}
For every constant $\gamma > 0$, there exists a polynomial time reduction that,
given an instance $\Gamma = (G, H, V_1 \cup \cdots \cup V_\ell)$ of \pname{MPSI}
where the supergraph $H$ is a complete graph\footnote{Reductions in
\autoref{lem:scss-red} and \autoref{lem:dsn-red} can be trivially modified to
work under a weaker assumption that $H$ is regular (but not necessarily
complete). However, we choose to only state the reductions when $H$ is complete
since this suffices for our purposes and the reductions are simpler to describe
in this case.}, produces an instance $(G',\mc{T}')$ of \scss, such
that
\begin{itemize}
\item (completeness) if $\val(\Gamma) = 1$, then there exists a network $N'
\subseteq G'$ of cost $2(1 + \gamma^{1/5})$ that is a solution of the instance $(G',\mathcal{T}')$ of \scss,
\item (soundness) if $\val(\Gamma) < \gamma$, then every network $N \subseteq
G'$ that is a solution of the instance $(G',\mathcal{T}')$ of \scss has cost more than $2(2 - 2\gamma^{1/5})$, and
\item (parameter dependency) the number of terminals $|\mc{T}'|$ is $\ell^2$.
\end{itemize}
\end{lem}

\begin{proof}
Assume without loss of generality that there is no edge in $G$ between two vertices in
the same set of the partition $V_1,\ldots,V_\ell$. We use the reduction of
\citet{DBLP:journals/siamdm/GuoNS11} with only a slight modification in that we
use different edge weights. Our graph remains unchanged from the
\citet{DBLP:journals/siamdm/GuoNS11} reduction; here we copy the graph
definition verbatim from~\cite{DBLP:journals/siamdm/GuoNS11}. We refer the reader to {\cite[Figure~3.1]{DBLP:journals/siamdm/GuoNS11}} for an illustration of the reduction.
The vertex set
$V'$ of $G'$ is $B \cup C \cup C' \cup D \cup D' \cup F$ where $B, C, C', D,
D', F$ are defined as follows:
\begin{itemize}
\item $B = \{b_i \mid i \in [\ell]\}$,
\item $C = \{c_v \mid v \in V_G\}$,
\item $C' = \{c'_v \mid v \in V_G\}$,
\item $D = \{d_{u, v}, d_{v, u} \mid uv \in E_G\}$,
\item $D' = \{d'_{u, v}, d'_{v, u} \mid uv \in E_G\}$, and,
\item $F = \{f_{i,j} \mid 1 \leq i\neq j \leq \ell\}$.
\end{itemize}
We view the partition $V_G = V_1 \cup \cdots \cup V_\ell$ as a function $\lambda: V_G
\to [\ell]$. The edge set $E'$ of $G'$ is $\cA \cup \cA' \cup \cB \cup \cD \cup
\cD' \cup \cH \cup \cZ \cup \cZ'$ where
\begin{itemize}
\item $\cA = \{\alpha_v = (b_{\lambda(v)}, c_v) \mid v \in
V_G\}$,
\item $\cA' = \{\alpha'_v = (c'_v, b_{\lambda(v)}) \mid v \in
V_G\}$,
\item $\cB = \{\beta_v = (c_v, c'_v) \mid v \in V_G\}$,
\item $\cD = \{\delta_{u, v} = (c'_u, d_{u, v}), \delta_{v,
u} = (c'_v, d_{v, u}) \mid uv \in E_G\}$,
\item $\cD' = \{\delta'_{u, v} = (d'_{u, v}, c_v),
\delta'_{v, u} = (d'_{v, u}, c_u) \mid uv \in E_G\}$,
\item $\cH = \{\epsilon_{u, v} = (d_{u, v}, d'_{u, v}),
\epsilon_{v, u} = (d_{v, u}, d'_{v, u}) \mid uv \in E_G\}$,
\item $\cZ = \{\zeta_{u, v} = (f_{\lambda(u), \lambda(v)},
d_{u, v}), \zeta_{v, u} = (f_{\lambda(v), \lambda(u)}, d_{v, u}) \mid
uv \in E_G\}$, and,
\item $\cZ' = \{\zeta'_{u, v} = (d'_{u, v}, f_{\lambda(u),
\lambda(v)}), \zeta'_{v, u} = (d'_{v, u}, f_{\lambda(v), \lambda(u)})
\mid uv \in E_G\}$.
\end{itemize}
As for the weights, we give weight $2\gamma^{1/5}/\ell$ to $\beta_v$ for every
$v \in V_G$, and weight $1/\binom{\ell}{2}$ to $\epsilon_{u, v}$ and
$\epsilon_{v, u}$ for every $uv \in E_G$; the rest of the edges have weight
zero. As noted earlier, this
is different from the weights assigned by \citet{DBLP:journals/siamdm/GuoNS11};
they simply assigned the same weight to every edge.
Finally, the terminal set is defined as $\mathcal{T}':= B \cup F$. Observe that the number of terminals is $|\mathcal{T}'|=|B|+|F|=\ell + 2\binom{\ell}{2} = \ell^2$.
We next move on to prove the completeness and soundness properties of the reduction.

{\bf (Completeness)} The solution in the completeness case is exactly the same
as the solution selected in~\cite{DBLP:journals/siamdm/GuoNS11}; we will repeat
their argument here.

If $\val(\Gamma) = 1$, then there exists $(v_1, \dots, v_\ell) \in V_1 \times
\cdots \times V_\ell$ that induces an $\ell$-clique. Consider the network 
$N'=\big(V(G'),E'\big)$ where 
\[
E' =
\big\{\alpha_{v_i}, \alpha'_{v_i}, \beta_{v_i} \mid i \in [\ell]\big\} \cup 
\big\{\delta_{v_i, v_j},
\delta'_{v_i, v_j}, \epsilon_{v_i, v_j}, \zeta_{v_i, v_j}, \zeta'_{v_i, v_j}
\mid 1 \leq i, j \leq \ell, i \ne j\big\}.
\]
The total weight the network $N'$ is $\ell \cdot
\left(2\gamma^{1/5}/\ell\right) + 2\binom{\ell}{2} \cdot \left(1/\binom{\ell}{2}\right) = 2(1 + \gamma^{1/5})$ as desired.

To see that the network $N'$ is indeed a solution for the instance $(G',\mathcal{T}')$ of \scss, observe that it suffices to show
that, for every $1 \leq i \ne j \leq \ell$, $f_{i, j}$ is reachable from $b_i$ and
$b_i$ is reachable from $f_{j, i}$. The former holds due to the path in $N$ given by the following edges (in order) $\alpha_{v_i}, \beta_{v_i}, \delta_{v_i, v_j}, \epsilon_{v_i, v_j},
\zeta'_{v_i, v_j}$ whereas the latter holds due to the path in $N$ given by the following edges (in order) $\zeta_{v_j, v_i},
\epsilon_{v_j, v_i}, \delta'_{v_j, v_i}, \beta_{v_i}, \alpha'_{v_i}$.

{\bf (Soundness)} Our soundness proof will require a more subtle analysis than
that of \citet{DBLP:journals/siamdm/GuoNS11}. Again, we will prove by
contrapositive. Suppose that there exists a network $N=\big(V(G'),E^*\big)$ of cost $\rho \leq 2(2
- 2\gamma^{1/5})$ which is a solution for the instance $(G',\mathcal{T}')$ of \scss. For each $i \in [\ell]$, let $S_i \subseteq V_G$ denote the set of all vertices $v \in V_i$ such that
$\beta_{v_i}$ is included in $E^*$. Moreover, let $S =
S_1 \cup \cdots \cup S_\ell$. Observe that, since each edge in $\cB$ has weight
$2\gamma^{1/5}/\ell$, we have $$|S| \leq \frac{\rho}{(2\gamma^{1/5}/\ell)} \leq
\frac{4}{(2\gamma^{1/5}/\ell)} \leq 2\gamma^{-1/5}\ell.$$

For every $1 \leq i \ne j \leq \ell$, let $\cH_{i, j}$ denote the set of all $\epsilon_{u, v}
\in E^*$ such that $u \in V_i$ and $v \in V_j$. First, we claim that, for every
$1 \leq i \ne j \leq \ell$, $\cH_{i, j} \ne \emptyset$. To see that this holds,
consider the set $T_{i, j} = \{f_{i, j}\} \cup \{d'_{u, v} \mid u \in
V_i, v \in V_j, uv \in E\}$. The only edges from outside $T_{i, j}$ coming into
this set are those in $\cH_{i, j}$. Since $\{f_{i,j}, b_i\}\subseteq \mathcal{T}'$, the vertex
$f_{i, j}$ has to be reachable from $b_i$ in $N$. However, since $b_i \notin T_{i, j}$, we can conclude
that at least one edge in $\cH_{i, j}$ must be selected in $E^*$.

Next, recall that each edge of the form $\epsilon_{u, v}$ has weight $1/\binom{\ell}{2}$. Since $N$ has cost $\rho$, we have
\begin{align*}
\sum_{1 \leq i \ne j \leq \ell} |\cH_{i, j}| &\leq \binom{\ell}{2} \cdot \rho,
\text{and thus} \\
\sum_{1 \leq i \ne j \leq \ell} (|\cH_{i, j}| - 1) &= \sum_{1 \leq i \ne j
\leq \ell} |\cH_{i, j}| - \ell(\ell-1) \leq
\binom{\ell}{2} \cdot \left(\rho - 2\right).
\end{align*}
From $\cH_{i, j} \ne \emptyset$ for every $1 \leq i \ne j \leq \ell$, the above
inequality implies that for at least $\binom{\ell}{2} \cdot (4 - \rho)
\geq 4\gamma^{1/5}\binom{\ell}{2}$ pairs of $(i, j)$'s we have $|\cH_{i, j}| = 1$ \big(since $\rho\leq 2(2-2\gamma^{1/5})$\big). Let
$\cP_{\uni}$ be the set of all such pairs of~$(i, j)$'s.

We will argue that a
random assignment defined from picking one vertex from each $S_i$ uniformly
independently at random covers many superedges in expectation. To do this, we
need to first show that, for many $(i, j)$'s, there exist $u \in S_i$ and $v \in
S_j$ such that $uv \in E_G$. In fact, we can show this for every $(i, j) \in
\cP_{\uni}$ as stated below.

\begin{claim} \label{claim:edgeins}
For every $(i, j) \in \cP_{\uni}$, there exists $u \in S_i$ and $v \in S_j$ such
that $uv \in E_G$.
\end{claim}

\begin{proof}
Since  $(i, j) \in \cP_{\uni}$, $\cH_{i, j}$ contains only one element. Let
this element be $\epsilon_{u_i, u_j}$. We will prove that $u_i \in S_i$ and $u_j
\in S_j$; note that this implies the claimed statement since $u_iu_j \in E_G$
by definition of~$\epsilon_{u_i, u_j}$.

To see that $u_i \in S_i$, consider the subset $C_{i, j} = \{f_{i,
j}\} \cup \{d'_{u, v} \mid u \in V_i, v \in V_j, uv \in E\} \cup
\{d_{u_i, u_j}\} \cup \{c'_{u_i}\}$. There are only two types of edges coming
into $C_{i, j}$: (1) $\beta_{u_i}$ and (2) $\epsilon_{u, v}$ where $u \in V_i,
v \in V_j$ and $(u, v) \ne (u_i, u_j)$. Since $\cH_{i, j} = \{\epsilon_{u_i,
u_j}\}$, the edges of the latter types are not selected in $E^*$. Moreover, since
$f_{i, j}$ is reachable from $b_i$, there must be at least one edge coming
into~$C_{i, j}$. As a result, $\beta_{u_i}$ must be selected, which means that
$u_i \in S_i$.

An analogous argument can be applied to $u_j$. Specifically, consider the
subset $C'_{i, j} = \{f_{i, j}\} \cup \{d_{u, v} \mid u \in V_i, v \in
V_j, uv \in E\} \cup \{d'_{u_i, u_j}\} \cup \{c_{u_j}\}$. There are only
two types of edges coming out of $C'_{ij}$: (1) $\beta_{u_j}$ and (2)
$\epsilon_{u, v}$ where $u \in V_i, v \in V_j$ and $(u, v) \ne (u_i, u_j)$.
Since $\cH_{i, j} = \{\epsilon_{u_i, u_j}\}$, the edges of the latter types are
not selected in $E^*$. Moreover, since $b_j$ is reachable from $f_{i, j}$, there
must be at least one edge coming out of $C_{i, j}$. As a result, $\beta_{u_j}$
must be selected, which means that $u_j \in S_j$.
\cqed
\end{proof}

Now, let $\phi: V_H \to V_G$ be a random assignment where each $\phi(i)$ is
chosen independently uniformly at random from $S_i$. By \autoref{claim:edgeins},
for every $(i, j) \in \cP_{\uni}$, there exists $u \in S_i$ and $v \in S_j$ such
that $uv \in E_G$. This means that, for such $(i, j)$, the probability that the
superedge $ij \in E_H$ is covered is at least the probability that $\phi(i) =
u$ and $\phi(j) = v$, which is equal to~$\frac{1}{|S_i| |S_j|}$. We now want a
lower bound on the expected number of superedges covered by $\phi$. For this, we
use the following inequality which follows from a special case of H\"{o}lder's
inequality for 3 variables\footnote{$(\sum_{r=1}^{n} a_r^{3})(\sum_{r=1}^{n}
b_r^{3})(\sum_{r=1}^{n} c_r^{3})\geq (\sum_{r=1}^{n} a_r b_r c_r)^{3}$}

\begin{align*}\label{eqn1}
\Big(\sum_{(i, j) \in \cP_{\uni}} \frac{1}{|S_i||S_j|}\Big) \cdot \Big(\sum_{(i, 
j) \in \cP_{\uni}} |S_i| \Big)&\cdot \Big(\sum_{(i, j) \in \cP_{\uni}} 
|S_j|\Big)\\
&\geq
\Big(\sum_{(i, j) \in \cP_{\uni}}
\Big(\frac{1}{|S_i||S_j|}\Big)^{1/3}\cdot |S_i|^{1/3}\cdot |S_j|^{1/3} \Big)^{3} \\
&= \Big(\sum_{(i, j) \in \cP_{\uni}} 1^{1/3} \Big)^{3} \\
&= |\cP_{\uni}|^{3} \tag{1}
\end{align*}

Hence, we have that the expected number of superedges covered by $\phi$ is at
least

\begin{align*}
\frac{1}{2} \cdot \sum_{(i, j) \in \cP_{\uni}} \frac{1}{|S_i||S_j|} %
&\geq \frac{1}{2} \cdot \frac{|\cP_{\uni}|^3}{\left(\sum_{(i, j) \in
\cP_{\uni}} |S_i|\right)\left(\sum_{(i, j) \in \cP_{\uni}} |S_j|\right)}
\tag*{\big(\text{from Equation~\eqref{eqn1}}\big)}\\
&\geq \frac{1}{2} \cdot \frac{|\cP_{\uni}|^3}{((\ell - 1)|S|)^2}
\tag*{\big(\text{since $S=\bigcup_{i=1}^{\ell} S_i$, and $(i,j)\in
\cP_{\uni}\Rightarrow i\neq j$}\big)}\\
&\geq \frac{1}{2} \cdot \frac{\big( 4\gamma^{1/5}\binom{\ell}{2} \big)^3}{(\ell - 1)^2 \cdot \big(2\gamma^{-1/5}\ell\big)^{2}}
\tag*{\big(\text{since $|\cP_{\uni}| \geq 4\gamma^{1/5}\binom{\ell}{2}$ and $|S| \leq 2\gamma^{-1/5}\ell$}\big)}\\
&= \frac{1}{2} \cdot \frac{64\gamma^{3/5}\cdot \binom{\ell}{2}^3}{(\ell - 1)^2 \cdot 4\gamma^{-2/5}\cdot \ell^{2}}\\
&\geq 2\gamma\cdot \binom{\ell}{2},
\end{align*}
where
note that the factor $1/2$ comes from the fact that we may double count each edge
for both $(i, j), (j, i)$. Hence, there exists an assignment of $\Gamma$ with value at least $2\gamma$, which implies that $\val(\Gamma) \geq 2\gamma\geq \gamma$.
\end{proof}

We can now easily prove \autoref{thm:lb-approx-SCSS} by combining
\autoref{lem:scss-red} and \autoref{crl:inapprox-colored-dks}.

\begin{proof}[Proof of \autoref{thm:lb-approx-SCSS}]
We again prove by contrapositive. Suppose that, for some constant $\varepsilon
> 0$ and for some function $f(k)$ independent of $n$, there exists an $f(k)
\polyn$-time $(2 - \varepsilon)$-approximation algorithm for
\scss. Let us call this algorithm $\bA$.

It is easy to see that there exists a sufficiently small $\gamma^* =
\gamma^*(\varepsilon)$ such that $\frac{2 - 2{\gamma^*}^{1/5}}{1 +
{\gamma^*}^{1/5}} \geq (2 - \varepsilon)$. We create an algorithm $\bB$ that can
distinguish between the two cases of \autoref{crl:inapprox-colored-dks} with
$h(\ell) = \log (1/\gamma^*)/\log \ell$. Our new algorithm $\bB$ works as
follows. Given an instance $(G, H, V_1 \cup \cdots \cup V_\ell)$ of \pname{MPSI}
where $H$ is a complete graph, $\bB$ uses the reduction from
\autoref{lem:scss-red} to create an \scss instance on the graph $G'$ with $k =
\ell^2$ terminals. $\bB$ then runs $\bA$ on this instance; if $\bA$ returns a
solution $N$ of cost at most $2(2 - 2{\gamma^*}^{1/5})$, then $\bB$ returns YES.
Otherwise, $\bB$ returns NO.

To see that algorithm $\bB$ can indeed distinguish between the YES and NO
cases, first observe that, in the YES case, \autoref{lem:scss-red} guarantees
that the optimal solution has cost at most $2(1 + {\gamma^*}^{1/5})$. Since $\bA$ is a $(2 - \varepsilon)$-approximation algorithm, it
returns a solution of cost at most $2(1 + {\gamma^*}^{1/5})
\cdot (2 - \varepsilon) \leq 2(2 - 2{\gamma^*}^{1/5})$ where
the inequality comes from our choice of $\gamma^*$; this means that $\bB$
outputs YES. On the other hand, if $(G, H, V_1 \cup \cdots \cup V_\ell)$ is a NO
instance, then the soundness property of \autoref{lem:scss-red} guarantees
that the optimal solution in $G'$ has cost more than $2(2 - 2{\gamma^*}^{1/5})$, which implies that $\bB$ outputs NO.

Finally, observe that the running time of $\bB$ is $f(\ell^2) \polyn$ and
that $h(\ell) = o(1)$. Hence, from \autoref{crl:inapprox-colored-dks}, randomized
Gap-ETH breaks.
\end{proof}

\subsection{Directed Steiner Network on Bidirected Graphs}

We will next prove our inapproximability result for \bidsn. For this result, we
will need a slightly more specific hardness of approximation for \mcsi where
every supernode has bounded degree. This bounded degree version of \pname{MPSI}
is defined below.

\begin{center}
\noindent\framebox{\begin{minipage}{0.9\textwidth}
\textbf{$t$-Bounded Degree \mcsi} \big(\pname{MPSI}($t$)\big)\\
\emph{Input}: an instance $\Gamma$ of \pname{MPSI}($t$) consists of three
components:
\begin{itemize}
\item an undirected graph $G = (V_G, E_G)$,
\item a partition of vertex set $V_G$ into disjoint subsets $V_1, \dots,
V_\ell$, and
\item an undirected graph $H = (V_H = \{1, \dots, \ell\}, E_H)$ such that each
vertex of $H$ has degree at most $t$.
\end{itemize}
\emph{Goal}: find an assignment $\phi: V_H \to V_G$ where $\phi(i) \in V_i$ for
every $i \in [\ell]$ that maximizes the number of edges $ij \in E_H$ such that
$\phi(i)\phi(j) \in E_G$.
\end{minipage}}
\end{center}

\citet{LRSZ17} gave the following reduction from (unbounded degree) \pname{MPSI}
to the bounded degree version of the problem. We remark here that their
reduction uses standard technique of sparsification via expanders, and similar
reductions have been presented before in literature (see e.g.~\cite{Dinur07}).

\begin{lem}[\cite{LRSZ17}]
For every $\varepsilon > 0$, there exists $\varepsilon' > 0$ and a polynomial
time reduction that, given an instance $\Gamma = (G, H, V_1 \cup \cdots \cup
V_\ell)$ of \pname{MPSI}, produces an instance $\Gamma' = (G', H', V_1 \cup
\cdots \cup V_{\ell'})$ of \pname{MPSI}(4) such that
\begin{itemize}
\item (YES) if $\val(\Gamma) = 1$, then $\val(\Gamma') = 1$,
\item (NO) if $\val(\Gamma) < 1 - \varepsilon$, then $\val(\Gamma') < 1 -
\varepsilon'$, and
\item (parameter dependency) $\ell' = O(\ell^2)$.
\end{itemize}
\end{lem}

Combined this with the parameterized inapproximability of
\autoref{crl:inapprox-colored-dks}, we can immediately conclude that the bounded
degree version of \pname{MPSI} is also hard to approximate, even for
parameterized algorithms:

\begin{crl} \label{crl:inapprox-bounded-degree-csi}
Assuming randomized Gap-ETH, for some $\varepsilon > 0$, there is no
$f(\ell)\polyn$-time algorithm that, given a \pname{MPSI}(4) instance $\Gamma =
(G, H, V_1 \cup \cdots \cup V_\ell)$ of size $n$, can distinguish between the
following two cases:
\begin{itemize}
\item (YES) $\val(\Gamma) = 1$.
\item (NO) $\val(\Gamma) < 1 - \varepsilon$.
\end{itemize}
\end{crl}

We are now ready to state the main lemma of this subsection, which provides a
reduction from bounded degree \pname{MPSI} to \bidsn:

\begin{lem} \label{lem:inapprox-bidsn-red}
For every constant $\varepsilon > 0$ and any $d \in \mathbb{N}$, there exists a
polynomial time reduction that, given an instance $\Gamma = (G, H, V_1 \cup
\cdots \cup V_\ell)$ of \pname{MPSI}($d$), produces an instance $(G', \mc{D}')$
of \bidsn and $B^* \in \mathbb{N}$ such that
\begin{itemize}
\item (completeness) if $\val(\Gamma) = 1$, then there exists a network $N
\subseteq G'$ of cost $B^*$ that satisfies all demands,
\item (soundness) if $\val(\Gamma) < 1 - \varepsilon$, then every network $N
\subseteq G'$ that satisfies all demands has cost more than $(1 +
\frac{\varepsilon}{11840 d})B^*$, and
\item (parameter dependency) The number of demand pairs $|\mc{D}'|$ is
$O(\ell)$.
\end{itemize}
\end{lem}

Before we proceed to prove the above lemma, let us note that
\autoref{thm:lb-approx-biDSN} follows immediately from
\autoref{crl:inapprox-bounded-degree-csi} and \autoref{lem:inapprox-bidsn-red}.

The construction for \autoref{lem:inapprox-bidsn-red} is exactly the same as
that in~\autoref{sec:bidsn-w[1]} with only one exception: each gadget will now
be a copy of the uniqueness gadget from~\autoref{sec:uniqueness-gadget} with $M
=13$ (instead of $M = k^4$ used before). Again, it is clear that the number of
demand pairs is $O(k + \ell) = O(\ell)$ where $k$ is the number of superedges,
i.e., $k = |E_H|$.

Let $B = 7M = 91$ and $B^* = 4\ell + (2k + 2\ell) \cdot B + (2k + 2\ell) \cdot B
+ (2k + \ell) \cdot (B + 4)$. Note that we can simplify the value of $B^*$ as follows:
\begin{equation}\label{eqn:Beestar}
B^* = 4\ell+(6k+5\ell)\cdot B + 4(2k+\ell) = 8(k+\ell)+ 91(6k+5\ell) = 554k + 463\ell
\end{equation}
It is not hard to see that, in the completeness
case, the solution used in~\autoref{sec:bidsn-w[1]} still works, and that it has
cost exactly $B^*$ as desired. Hence, we are only left to show the soundness of
the reduction.

We proceed to prove the soundness. Again, we will prove our soundness by
contrapositive. Suppose that there exists a network $\big(V(G'),N\big)$ of cost $\rho < (1 +
\beta)B^*$ where $\beta = \frac{\varepsilon}{11840d}$ that satisfies all the
demand pairs. We will also assume without loss of generality that the edge set
$N$ is \emph{inclusion-wise minimal}, i.e., that if we remove any edge from $N$,
then at least one demand pair must be unsatisfied.

Since our underlying graph and the demand pairs are exactly the same as those
from~\autoref{sec:bidsn-w[1]}, the restriction of $N$ into each gadget must
again satisfy the in-out property (similar to
\autoref{lem:in-out-vertical-general}, \autoref{lem:in-out-horizontal-general},
and \autoref{lem:in-out-main-general}) as stated below:

\begin{lem}
For any $j\in [\ell], i\in N'_{H}(j)$ the edges of $N$ which have both
end-points in the horizontal secondary gadget $\VS_{i,j}$ satisfy the in-out
property. Hence, $\VS_{i,j}$ uses up weight of at least $B$ from the budget.
\label{lem:in-out-v}
\end{lem}

\begin{lem}
For any $i\in [\ell], j\in N'_{H}(i)$ the edges of $N$ which have both
end-points in the vertical secondary gadget $\HS_{i,j}$ satisfy the
in-out property. Hence, $\HS_{i,j}$ uses up weight of at least $B$ from the
budget.
\label{lem:in-out-h}
\end{lem}

\begin{lem}
For every $i,j$ such that $ij\in E^*_H$ the edges of $N$ which have both
end-points in the main gadget $\M_{i,j}$ satisfy the in-out property. Hence,
$\M_{i,j}$ uses up weight of at least $B$ from the budget.
\label{lem:in-out-m}
\end{lem}

We say that a gadget is \emph{tight} if $N$ restricted to the gadget has cost
exactly $B$. Recall that the first step of the proof of the reverse direction\footnote{If the instance $(G^*,\mathcal{D})$ of \bidsn has a solution of weight $\leq B^*$ then the instance $(G,H)$ of \csi has a solution} of \autoref{thm:lb-biDSN} was to observe that every gadget must be tight; this was
true because the value $M$ over there was set so large that even an excess of
$B$ was already more than the total cost of all
\textcolor[rgb]{1.00,0.00,0.00}{red} edges. However, this is not true in our
modified construction anymore as we choose $M = 13$. Fortunately for us, we will
still be able to show that all but a small fraction of the gadgets are tight.

To prove such a bound, first recall that \autoref{lem:main-at-least-4-red} (used in the proof of \autoref{thm:lb-biDSN}) exactly shows that if a main gadget and all its four surrounding secondary
gadgets are tight, then the edge set $N$ must contain at least four
\textcolor[rgb]{1.00,0.00,0.00}{red} edges with exactly one-end point in the main
gadget. We restate this formally as follows (proof is omitted since it is exactly the same as that of \autoref{lem:main-at-least-4-red})

\begin{lem}
\label{lem:main-red}
For each $1\leq i,j\leq \ell$ such that $ij\in E^*_H$, if the main gadget $\M_{i,
j}$ and the four secondary gadgets surrounding it $\big( \HS_{i,j}, \HS_{i,\nextone_{j}(i)}, \VS_{i,j}$ and $\VS_{\nextone_{i}(j),j}\big)$ are tight, then the edge set $N$ must contain at least one
\textcolor[rgb]{1.00,0.00,0.00}{red} edge of each of the following types:
\begin{enumerate}[(a)]
\item an edge with one end-point in the set of $3$-vertices of $\M_{i,j}$ and the
other end-point in the set of $0$-vertices of $\HS_{i,j}$,
\item an edge with one end-point in the set of $0$-vertices of $\M_{i,j}$ and the
other end-point in the set of $3$-vertices of $\HS_{i,\nextone_{j}(i)}$,
\item an edge with one end-point in the set of $0$-vertices of $\M_{i,j}$ and the
other end-point in the set of $3$-vertices of $\VS_{i,j}$, and
\item an edge with one end-point in the set of $3$-vertices of $\M_{i,j}$ and the
other end-point in the set of $0$-vertices of $\VS_{\nextone_{i}(j),j}$.
\end{enumerate}
\end{lem}

Next we restate \autoref{lem:dotted}, which we can use here since it only uses the fact that the edge set $N$ is such that $\big(V(G'),N\big)$ is a solution for the instance $(G',\mathcal{D}')$ of \bidsn:

\lemdotted*

We are now ready to prove a bound on the number of non-tight gadgets. The key
idea here is that, while having a non-tight gadget may help ``save'' the number
of required \textcolor[rgb]{1.00,0.00,0.00}{red} edges
from~\autoref{lem:main-at-least-4-red}, this saving is still smaller than the excess cost
of $M$. Hence, if there are too many non-tight gadgets, then the cost of $N$
must be much more than the minimum possible cost of $B^*$, which would
contradict our assumption that the cost of $N$ is at most $(1 + \beta) B^*$.

In addition to the bound on the number of non-tight gadgets, we will be able to
give an upper bound on the number of main gadgets with at least five
\textcolor[rgb]{1.00,0.00,0.00}{red} edges touching them; again, this is just
because these edges add to the minimum possible cost $B^*$.

\begin{lem}
There are at most $\beta \cdot B^*$ non-tight gadgets. Moreover, there are at
most $\beta \cdot B^*$ main gadgets $\M_{i, j}$ such that there are at least
five \textcolor[rgb]{1.00,0.00,0.00}{red} edges with at least one endpoint
in~$\M_{i, j}$.
\label{lem:most-tight}
\end{lem}

\begin{proof}
Let $X$ be the number of non-tight gadgets and $Y$ be the number of main gadget
$\M_{i, j}$ such that there are at least five
\textcolor[rgb]{1.00,0.00,0.00}{red} edges with at least one endpoint in $\M_{i,
j}$.

We can lower bound the cost of the solution $N$ as follows.
\begin{itemize}
\item From \autoref{lem:dotted}, at least $4\ell$ \textcolor[rgb]{1.00,0.50,0.00}{orange} edges must be
selected.
\item Since there is a total of $6k + 5\ell$ gadgets (including both main and
secondary gadgets), at least $(6k + 5\ell - X)$ of these gadgets are tight.
These tight gadgets use up weight of $(6k + 5\ell - X) B$ from the budget.
Further, \autoref{lem:macro-uniqueness-gadget}, together with
\autoref{lem:in-out-v}, \autoref{lem:in-out-h} and \autoref{lem:in-out-m},
implies that each of the $X$ non-tight gadgets uses up weight at least $8M=(B+M)$.
Hence, in total, the weight of edges of $N$ whose both endpoints are from the
same gadget is at least $(6k + 5\ell)B + XM$.
\item Let us divide the main gadgets $\M_{i, j}$ into three groups based on the
number of \textcolor[rgb]{1.00,0.00,0.00}{red} edges touching them: (1) there
are at most three such edges, (2) there are at least five such edges and (3)
there are exactly four such edges.

From \autoref{lem:main-red}, each gadget of type (1) must either be non-tight or
be adjacent to at least one non-tight secondary gadgets. Since there are only
$X$ non-tight gadgets and each secondary gadget is adjacent to at most two main
gadgets, the number of main gadgets of type (1) is at most $X + 2X = 3X$. Recall
also that we assume that the number of main gadgets of type (2) is $Y$. As a
result, the number of \textcolor[rgb]{1.00,0.00,0.00}{red} edges is at least $5Y
+ 4(2k + \ell - 3X - Y) = Y - 12X + 4(2k + \ell)$.
\end{itemize}
We can conclude that in total the cost of $N$ must be at least
\begin{align*}
4\ell + \Big((6k + 5\ell)B + XM\Big) &+ \Big(Y - 12X + 4(2k + \ell)\Big)\\
&= 4\ell + \Big((6k + 5\ell)B + 13X\Big) + \Big(Y - 12X + 4(2k + \ell)\Big) 
\tag*{(since $M=13$)}\\
&= \Big(4\ell + (6k + 5\ell)B + 4(2k+\ell)  \Big) + Y + X \\
&= \Big(4\ell + (2k + 2\ell)B + (2k + 2\ell)B + (2k + \ell)(B+4) \Big) + Y + X \\
&= B^* + Y + X
\end{align*}
where the last equality follows because $B^* = 4\ell + (2k + 2\ell)B + (2k + 2\ell)B + (2k + \ell)(B+4)$. Since we assume that the total cost of $N$ is at most $(1 + \beta)B^*$, we have
$X, Y \leq \beta \cdot B^*$ as desired.
\end{proof}

We will next use the above bound to help us find a solution $\phi: V_H \to V_G$
to the \pname{MPSI}($d$) instance $\Gamma$. Unlike in the proof of
\autoref{thm:lb-biDSN} where the network $N$ canonically gives $\phi(i)$ for
every $i \in [\ell]$, this will only be true for ``good'' $i$ which is defined
below.

\begin{dfn}
A main gadget $\M_{i, j}$ is \emph{good} if the gadget itself and all its
surrounding secondary gadgets ($\HS_{i,j}$, $\HS_{i,\nextone_{j}(i)}$,
$\VS_{i,j}$, and $\VS_{\nextone_{i}(j),j}$) are tight and there are exactly four
\textcolor[rgb]{1.00,0.00,0.00}{red} edges with one endpoint in $\M_{i, j}$. We
call a main gadget $\M_{i, j}$ \emph{bad} if it is not good.

Furthermore, $i \in [\ell]$ is said to be \emph{good} if $\M_{i, j}$ and $\M_{j,
i}$ are good for every $ij \in E^*_H$. Similarly, we say that $i \in [\ell]$ is
\emph{bad} if it is not good.
\end{dfn}

We now set up some notation regarding representation of each tight gadget. Note that, while in \autoref{sec:bidsn-w[1]} every gadget is tight and hence the notation there applied for all gadgets, the following notation is only well-defined for tight gadgets in our proof:
\begin{itemize}
\item For each $j\in [\ell]$ and each $i\in N'_{H}(j)$, if the horizontal secondary gadget
$\VS_{i, j}$ is tight then $\VS_{i,j}$ is represented (\autoref{defn-representation}) by some $y_{i,j}\in V_j$,
\item For each $i\in [\ell]$ and each $j\in N'_{H}(i)$, if the vertical secondary
gadget $\HS_{i,j}$ is tight then $\HS_{i, j}$ represented by some $x_{i,j}\in
V_i$,
\item For each $ij\in E^*_H$, if the main gadget $\M_{i,j}$ is tight then $\M_{i, j}$ is represented by some $(\lambda_{i,j}, \delta_{i,j})\in E_{i,j}$.
\end{itemize}

Consider the assignment $\phi: V_H \to V_G$ defined as follows: for each $i\in [\ell]$
    \[
 \phi(i) =
  \begin{cases}
   \lambda_{i,i} & \text{if } i\ \text{is good} \\
   \text{any arbitrarily chosen vertex from $V_i$}       & \text{otherwise}
  \end{cases}
\]
The remaining argument
consists of two parts. First, we will show that $\phi$ covers every superedge
$ij \in E_H$ such that both $i, j$ are good. Then, we will argue that only a
small fraction of $i \in [\ell]$ is bad. Combining these two parts completes our
proof.

To show that $\phi$ satisfies all superedges whose endpoints are both good, we first argue (similar to \autoref{lem:agreement-tight-dsn} where every main gadget was good) that good gadgets allow us to propagate equality of representations, as stated formally below.

\begin{lem}
For every good main gadget $\M_{i, j}$, we have
$x_{i,j}=\lambda_{i,j}=x_{i,\nextone_{j}(i)}$ and
$y_{i,j}=\delta_{i,j}=y_{\nextone_{j}(i),j}$.
\label{lem:agreement-tight}
\end{lem}

\begin{proof}
Due to symmetry, it suffices to only argue that $x_{i,j}=\lambda_{i,j}$. Let us
assume for the sake of contradiction that $x_{i,j}\ne\lambda_{i,j}$. We will
argue that there is a vertex such that (1) there is exactly one edge adjacent to
it from the network $N$ and (2) it does not belong to any demand pair. Observe
that removing its only adjacent edge from $N$ does not affect the validity of
the solution. This contradicts our assumption that $N$ is minimal.

From \autoref{lem:main-red} and from our assumption that there are exactly four
\textcolor[rgb]{1.00,0.00,0.00}{red} edges with one endpoint in $\M_{i, j}$,
there must be exactly one red edge from each of the following types:
\begin{enumerate}[(a)]
\item an edge with one end-point in the set of $3$-vertices of $\M_{i,j}$ and the
other end-point in the set of $0$-vertices of $\HS_{i,j}$,
\item an edge with one end-point in the set of $0$-vertices of $\M_{i,j}$ and the
other end-point in the set of $3$-vertices of $\HS_{i,\nextone_{j}(i)}$,
\item an edge with one end-point in the set of $0$-vertices of $\M_{i,j}$ and the
other end-point in the set of $3$-vertices of $\VS_{i,j}$, and
\item an edge with one end-point in the set of $3$-vertices of $\M_{i,j}$ and the
other end-point in the set of $0$-vertices of $\VS_{\nextone_{i}(j),j}$.
\end{enumerate}

Observe that if the edge of type (a) does not have one endpoint at $\HS_{i,
j}(0_{x_{i, j}})$, then the vertex $\HS_{i, j}(0_{x_{i, j}})$ is the desired
vertex.

Now, suppose that one endpoint of the edge of type (a) is $\HS_{i, j}(0_{x_{i,
j}})$. The other endpoint must be $\M_{i, j}(3_{x_{i, j}, y})$ for some $y \in
V_j$. Since $x_{i, j} \ne \lambda_{i, j}$, we have $\M_{i, j}(3_{x_{i, j}, y})
\ne \M_{i, j}(3_{\lambda_{i, j}, \delta_{i, j}})$. Consider the edge of type (d);
suppose that one of its endpoint is $\M_{i, j}(3_{x', y'})$. Since $\M_{i,
j}(3_{x_{i, j}, y}) \ne \M_{i, j}(3_{\lambda_{i, j}, \delta_{i, j}})$, at least
one of the following must be true: $\M_{i, j}(3_{x_{i, j}, y}) \ne \M_{i,
j}(3_{x', y'})$ or $\M_{i, j}(3_{\lambda_{i, j}, \delta_{i, j}}) \ne \M_{i,
j}(3_{x', y'})$.

If $\M_{i, j}(3_{x_{i, j}, y}) \ne \M_{i, j}(3_{x', y'})$,
then $\M_{i, j}(3_{x_{i, j}, y})$ is the desired vertex. Otherwise, if $\M_{i,
j}(3_{\lambda_{i, j}, \delta_{i, j}}) \ne \M_{i, j}(3_{x', y'})$, then $\M_{i,
j}(3_{\lambda_{i, j}, \delta_{i, j}})$ is the desired vertex.

In all cases, we have found a vertex with desired properties, and hence we
have arrived at a contradiction.
\end{proof}

Our main claim now follows almost immediately from
\autoref{lem:agreement-tight}.

\begin{lem}
The mapping $\phi$ covers every $ij \in E_H$ such that both $i, j$ are good.
\end{lem}

\begin{proof}
Consider any such superedge $ij \in E_H$. Let $j_1 < j_2 < \cdots < j_p$ be all
elements of $N'_H(i)$ and $i_1 < i_2 < \dots < i_q$ be all elements of
$N'_H(j)$, \autoref{lem:agreement-tight} implies that
\begin{align*}
x_{i, j_1} = \lambda_{i, j_1} = x_{i, j_2} = \cdots = x_{i, j_p},
\end{align*}
and
\begin{align*}
y_{i_1, j} = \delta_{i_1, j} = y_{i_2, j} = \cdots = y_{i_q, j}.
\end{align*}
Since $j \in N'_H(i)$ and $i \in N'_H(j)$, the above inequalities imply that
$\lambda_{i,j} = \lambda_{i, i}$ and $\delta_{i, j} = \delta_{j, j}$.
Furthermore, observe that $(\lambda_{j, j}, \delta_{j, j}) \in E_{j, j}$,
meaning that $\delta_{j, j} = \lambda_{j, j}$.

Recall that we set $\phi(i) = \lambda_{i, i}$ and $\phi(j) = \lambda_{j, j}$.
This means that $(\phi(i), \phi(j)) = (\lambda_{i, j}, \delta_{i, j})$ which
must be in~$E_{i, j}$. In other words, $ij$ is covered by $\phi$.
\end{proof}

For the second part, let us first argue an upper bound on the number of bad
main gadgets. Observe that each bad main gadget $\M_{i, j}$ must satisfy at least
one of the three following conditions: (1) $\M_{i, j}$ is not tight, (2) one of
its surrounding secondary gadgets is not tight, or (3) there are at least five
\textcolor[rgb]{1.00,0.00,0.00}{red} edges with one endpoint in $\M_{i, j}$.
\autoref{lem:most-tight} implies that there are at most $\beta \cdot B^*$, $2
\beta \cdot B^*$, and $\beta \cdot B^*$ main gadgets that satisfy (1), (2), and
(3) respectively (recall that each secondary gadget has edges to at most two main gadgets). Hence, in total, there are at most $4 \beta \cdot B^*$ bad
main gadgets. Since, for each bad $i \in [\ell]$, there must exist some $j \in N'_H(i)$
such that $\M_{i, j}$ or $\M_{j, i}$ is a bad gadget, there can be at most $8
\beta \cdot B^*$ bad $i \in [\ell]$.

Due to our bounded degree assumption on $H$, there can be at most $8 d \beta
\cdot B^*$ superedges $ij \in E_H$ such that at least one of $i, j$ is bad. As
a result, $\phi$ satisfies all but $8 \beta d \cdot B^*$ superedges. Thus, we have
\begin{align*}
\val(\Gamma) &\geq 1 - \frac{8 \beta d \cdot B^*}{k} \\
&= 1 - \frac{8 \beta d \cdot (554k+463\ell)}{k} \tag*{(since $B^*=554k+463\ell$ from~\autoref{eqn:Beestar})}\\
&\geq 1 - 8 \beta d (554+926) \tag*{(since $k\geq \ell/2$)}\\
&= 1 - \varepsilon \tag*{(since $\beta = \frac{\varepsilon}{11840d}$)}
\end{align*}
where the second inequality comes from the fact that we can assume without loss
of generality that the supergraph $H$ does not contain any isolated vertex. This
concludes the proof of \autoref{lem:inapprox-bidsn-red}.

\section{A Reduction from MPSI to DSN}
\label{app:dsn}

In a previous version~\cite{DBLP:conf/esa/ChitnisFM18} of this manuscript, we 
had provided the following $k^{o(1)}$-factor inapproximability result for~\dsn:
\begin{thm}
Under Gap-ETH, for any function $g(k) =o(1)$and any $f(k)$ independent of $n$, there is no $f(k)\cdot n^{O(1)}$ time algorithm that computes an $k^{g(k)}$-approximation for \dsn.
\label{thm:DSN-inapprox}
\end{thm}

\autoref{thm:DSN-inapprox} has since been subsumed by~\cite{DM18} which shows an improved hardness of $k^{1/4 - o(1)}$-approximation for \dsn under the same assumption of Gap-ETH.
At the heart of our proof of \autoref{thm:DSN-inapprox} is the following
lemma which provides a gap-preserving FPT reduction from \mcsi to \dsn.

\begin{lem} \label{lem:dsn-red}
There exists a polynomial time reduction that, given an instance $\Gamma = (G, H, V_1 \cup \cdots \cup V_{\ell})$ of MPSI where the supergraph $H$ is a complete graph, produces an instance of \dsn with a graph $G'$ and $k$ demand pairs, such
that
\begin{itemize}
\item (completeness) if $\val(\Gamma) = 1$, there is
a network $N \subseteq G'$ of cost $1$ that satisfies all demands,
\item (soundness) for any $\gamma > 0$ (possibly depending on $k$), if
$\val(\Gamma) < \gamma$, then every network $N \subseteq G'$ that satisfies all
demands has cost more than~$1/\sqrt{4\gamma}$, and
\item (parameter dependency) $k = \ell^2 - \ell$.
\end{itemize}
\end{lem}

The proof of \autoref{thm:DSN-inapprox} follows immediately from \autoref{lem:dsn-red} and \autoref{crl:inapprox-colored-dks}:
\begin{proof}[Proof of \autoref{thm:DSN-inapprox}]
We prove by using the contrapositive. Suppose that, for some function $g(k)=o(1)$ and for some function $f(k)$ independent of $n$, there exists an $f(k)\cdot n^{o(1)}$ time $k^{g(k)}$-approximation algorithm for \dsn. Let us call this algorithm $\mathbb{A}$.

We now design an algorithm $\mathbb{B}$ that can distinguish between the two cases of \autoref{crl:inapprox-colored-dks} with $h(\ell) =4g(\ell^{2}-\ell) + \frac{4}{\log_{2} \ell}$. The algorithm $\mathbb{B}$ works as  follows: given an instance $(G,\ell,V_1\cup V_2 \cup \ldots \cup V_{\ell})$ of \mcsi where the supergraph $H$ is the complete graph on $\ell$ nodes, $\mathbb{B}$ uses the reduction from \autoref{lem:dsn-red} to create a \dsn instance on the graph $G'$ with $k=\ell^{2}-\ell$ demands. $\mathbb{B}$ then runs $\mathbb{A}$ on this instance; if $\mathbb{A}$ returns a solution $N$ of cost at most $k^{g(k)}=(\ell^{2}-\ell)^{g(\ell^{2}-\ell)}$, then $\mathbb{B}$ returns YES. Otherwise, $\mathbb{B}$ returns NO.

To see that algorithm $\mathbb{B}$ can indeed distinguish between the YES and NO cases, first observe that,in the YES case, \autoref{lem:dsn-red} guarantees that the optimal solution is of cost at most $1$. Since $\mathbb{A}$ is an $k^{g(k)}$-approximation algorithm, it returns a solution of cost at most $ k^{g(k)} = (\ell^{2}-\ell)^{g(\ell^{2}-\ell)}$, meaning that $\mathbb{B}$ outputs YES. On the other hand, if $(G,\ell,V_1\cup V_2 \cup \ldots \cup V_{\ell})$ is a NO instance, then the soundness property of \autoref{lem:dsn-red} guarantees that the optimal solution in $G'$ has cost more than $\frac{1}{\sqrt{4 \ell^{-h(\ell)}}}$. In this case, $\mathbb{B}$ also outputs NO since we have
\begin{align*}
\frac{1}{\sqrt{4 \ell^{-h(\ell)}}} &= \frac{\sqrt{\ell^{h(\ell)}}}{2} \\
&= \frac{1}{2}\cdot \ell^{2g(\ell^{2}-\ell) + \frac{2}{\log_{2} \ell}} \tag*{\big(since $h(\ell) =4g(\ell^{2}-\ell) + \frac{4}{\log_{2} \ell}$\big)} \\
&= \frac{1}{2}\cdot (\ell^2)^{g(\ell^{2}-\ell)} \cdot \ell^{\frac{2}{\log_{2} \ell}}\\
&> \frac{1}{2}\cdot (\ell^2-\ell)^{g(\ell^{2}-\ell)} \cdot 4  \tag*{\big(since $\ell^2 > (\ell^2-\ell)$ for each $\ell\geq 1$ and $\ell^{\frac{2}{\log_{2} \ell}}=4$\big)}\\
&> (\ell^2-\ell)^{g(\ell^{2}-\ell)}
\end{align*}

Finally, observe that the running time of $\mathbb{B}$ is bounded by the running time of $\mathbb{A}$ plus the $n^{O(1)}$ time needed for the reduction of \autoref{lem:dsn-red}. Since $k=\ell^2 - \ell$ and the running time of $\mathbb{A}$ is $f(k)\cdot n^{O(1)}$, it follows that the running time of $B$ can be expressed as $f'(\ell)\cdot n^{O(1)}$ for some function $f'$. Moreover, since $g(k) =o(1)$ it also follows that $h(\ell) =o(1)$. Hence, from \autoref{crl:inapprox-colored-dks}, randomized Gap-ETH breaks. This concludes the proof of \autoref{thm:DSN-inapprox}.
\end{proof}

In~\cite{DM18}, an $\ell^{1 - o(1)}$ factor inapproximability result for \mcsi is proved, which is an improvement over the $\ell^{o(1)}$ factor hardness in \autoref{crl:inapprox-colored-dks}. The authors of~\cite{DM18} then use this improved hardness together with our reduction in \autoref{lem:dsn-red} to arrive at their $k^{1/4 - o(1)}$ factor hardness for \dsn. Since \autoref{lem:dsn-red} is used even in \cite{DM18} but does not appear in the published version of~\cite{DM18}, we have kept its proof in our paper.

\begin{proof}[Proof of \autoref{lem:dsn-red}]
The reduction is similar to that of \citet{dodis1999design}. In
particular, given $\Gamma = (G, H, V_1 \cup \cdots \cup V_\ell)$ where $H$ is the complete graph, the \dsn instance is generated as follows.
\begin{itemize}
\item The vertex set $V'$ is $(V_G \times [2]) \cup \{s_1, \dots, s_\ell\} \cup
\{t_1, \dots, t_\ell\}$ (i.e. two copies of $V$ together with $\ell$ new vertices
designated as sources and $\ell$ new vertices designated as sinks).
\item There are three types of edges in $E'$. First, for every $i \in [\ell]$,
there is an edge from $s_i$ to each vertex in $V_i \times \{1\}$. Moreover, for
every $i \in [\ell]$, there is an edge from each vertex in $V_i \times \{2\}$
to~$t_i$. Finally, there is an edge from $(u, 1)$ to $(v, 2)$ and from $(v, 1)$
to $(u, 2)$ for every edge $uv$ in the original graph $G$.
In other words, $E' = \{(s_i, (v, 1)) \mid i \in [k], v \in V_i\} \cup \{((v,
2), t_i) \mid i \in [k], v \in V_i\} \cup \{((u, 1), (v, 2)), ((v, 1), (u, 2))
\mid uv \in E_G\}$.
\item The edges of the first two types have weight $1/(2\ell)$, whereas the edges of
the last type have weight zero.
\item Finally, the demands are simply $(s_i, t_j)$ for every $i, j \in [\ell]$ such
that $i \ne j$.
\end{itemize}

Clearly, the number of demand pairs $k$ is $\ell^2 - \ell$ as desired. We now move on to show the completeness and soundness properties of the reduction.

{\bf (Completeness)} If $\val(\Gamma) = 1$, then there exists $(v_1, \dots, v_\ell) \in V_1 \times \cdots V_\ell$ that induces a clique. Thus, we can pick edges in the set $\{(s_i, (v_i,
1)) \mid i \in [\ell]\} \cup \{((v_i, 2), t_i) \mid i \in [\ell]\} \cup \{((v_i, 1),
(v_j, 2)) \mid i, j \in [\ell], i \ne j\}$. Clearly, the cost of this network is exactly one and it satisfies all the demand pairs.

{\bf (Soundness)} We will prove this by contrapositive. Suppose that there
exists a network $N\subseteq G'$ of cost $\rho \leq 1 / \sqrt{4\gamma}$.
For each $i \in [\ell]$, let $S_i \subseteq V_G$ denote the set of all vertices $v$
such that at least one of $(s_i, (v, 1))$ or $((v, 2), t_i)$ is included in
$N$. Observe that, from how our graph $G'$ is constructed, for every $i \ne j
\in [k]$, the $(s_i, t_j)$ demand implies that there exist $u \in S_i$ and $v
\in S_j$ such that $uv \in E_G$.
Let $S = \bigcup_{i \in [\ell]} S_i$.
Observe also that, since $N$ has cost $\rho$, $|S| \leq 2\ell \cdot \rho \leq \ell/\sqrt{\gamma}$.

Let $\phi: V_H \to V_G$ be a random assignment where each $\phi(i)$ is chosen
independently uniformly at random from $S_i$. For every $i \ne j \in [\ell]$,
since there exist $u \in S_i$ and $v \in S_j$ such that $uv \in E_G$, the
probability that the superedge $ij \in E_H$ is covered is at least the
probability that $\phi(i) = u$ and $\phi(j) = v$, which is equal to
$\frac{1}{|S_i| |S_j|}$. We now want a lower bound on the expected number of
superedges covered by~$\phi$. For this, we use the following inequality, which
follows from a special case of H\"{o}lder's inequality for 3
variables\footnote{$(\sum_{r=1}^{n} a_r^{3})(\sum_{r=1}^{n}
b_r^{3})(\sum_{r=1}^{n} c_r^{3})\geq (\sum_{r=1}^{n} a_r b_r c_r)^{3}$}

\begin{align*}
\Big(\sum_{1 \leq i \ne j \leq \ell} \frac{1}{|S_i||S_j|}\Big) \cdot \Big(\sum_{1 \leq i \ne j \leq \ell} |S_i| \Big)\cdot \Big(\sum_{1 \leq i \ne j \leq \ell} |S_j|\Big)
&\geq \Big(\sum_{1 \leq i \ne j \leq \ell}
\Big(\frac{1}{|S_i||S_j|}\Big)^{1/3}\cdot |S_i|^{1/3}\cdot |S_j|^{1/3}
\Big)^{3} \\
&= \Big(\sum_{1 \leq i \ne j \leq \ell} 1^{1/3} \Big)^{3} \\
&= (\ell(\ell-1))^{3} \tag{2}
\end{align*}

Hence, we have that the expected number of superedges covered
by $\phi$ is at least
\begin{align*}
\sum_{ij \in E_H} \frac{1}{|S_i||S_j|}
&= \frac{1}{2} \sum_{1 \leq i \ne j \leq \ell} \frac{1}{|S_i||S_j|} \\
&\geq \frac{1}{2} \cdot \frac{\left(\ell(\ell - 1)\right)^3}{\left(\sum_{1 \leq i \ne j \leq \ell} |S_i|\right)\left(\sum_{1 \leq i \ne j \leq \ell} |S_j|\right)} \quad(\text{From Equation (2)})\\
&\geq \binom{\ell}{2} \cdot \frac{\ell^2}{|S|^2} \quad(\text{Since $\sum_{1 \leq i \ne j \leq \ell} |S_i| \leq (\ell-1)\cdot |S|$}) \\
&\geq \binom{\ell}{2} \gamma,
\end{align*}
where the
last inequality follows from $|S| \leq \ell/\sqrt{\gamma}$. Hence, there exists an assignment of $\Gamma$ with value at least $\gamma$, which implies that $\val(\Gamma) \geq \gamma$. This concludes the proof of \autoref{lem:dsn-red}.
\end{proof}

\section{Open Questions}
\label{sec:questions}

While our work has advanced our understanding of the computational complexity of 
\scss and \dsn, there are still several interesting open questions left. We list 
some of them below:

\begin{itemize}
\item Can we get better approximation algorithms for \bidsn (without any 
restriction on the optimum) than simply getting twice the best ratio known for 
the undirected \pname{Steiner Forest} problem? This is an interesting question 
for both the parameterized and polynomial time setting.

\item We showed that for \bidsn there is both a parameterized $2$-approximation 
algorithm and a polynomial-sized $(2+\eps)$-approximate kernel for any 
$\eps>0$. However the latter is just a simple consequence of the PSAKS for 
\bidsnP. Is there a polynomial-sized $c$-approximate kernel with $c\leq 2$ for 
\bidsn? Note that this relates to the previous question as well.

\item We proved that the parameterized $2$-approximation algorithm for \scss is 
best possible, since no $(2-\eps)$-approximation can be computed in 
$f(k)\polyn$ time for any function $f$, under Gap-ETH. This implies that there 
is a $2$-approximate kernel (of large size), while no $(2-\eps)$-approximate 
kernel exists under the same assumption. However, can we obtain a 
\emph{polynomial-sized} $2$-approximate kernel for \scss? Or maybe just a 
polynomial-sized $c$-approximate kernel for some constant $c\geq 2$?

\item Can we prove any runtime lower bound under some reasonable complexity 
assumption (e.g., ETH or Gap-ETH) to compute a $2$-approximation for \scss 
using the number of terminals as a parameter? In other words, could there be a 
significantly faster $2$-approximation algorithm than the one given 
in~\cite{DBLP:conf/iwpec/ChitnisHK13}?

\item We gave a $4^{k^2+O(k)}\polyn$ time FPT algorithm for \biscss and a lower 
bound of $2^{o(k)}\polyn$. Can we obtain an FPT algorithm for \biscss with 
runtime $2^{O(k)}\polyn$?

\item What is the status of \bidsn on planar input graphs parameterized by $k$: 
FPT or W[1]\hy{}hard? Our hardness reduction in \autoref{thm:lb-scheme-biDSN} 
produces graphs that are not planar even though their optima are.

\item Can the parameterized approximation scheme for \bidsnP be generalized to 
minor-closed classes of graphs? In particular, the KPR Theorem used to prove 
\autoref{lem:full-comp} is applicable to such classes. What prevents us to 
generalize here are the vertex degree transformation of \autoref{app:degrees}, 
since applying these to some graph excluding a fixed minor can result in a 
graph containing this minor.

\item \autoref{lem:full-comp} inherently introduces a double exponential term in 
$O(1/\eps)$ to the kernel size for \bidsnP, and in the runtime of the 
approximation scheme for \bidsnP. As argued in \autoref{sec:scheme} it is known 
that the bound in \autoref{lem:full-comp} cannot be improved. Is there a 
different technique that yields a parameterized approximation scheme and/or a 
PSAKS for \bidsnP, which has better dependence on $1/\eps$? Or alternatively, is 
there some reasonable complexity assumption that can exclude such an 
improvement?
\end{itemize} 

In addition to the above questions regarding the \dsn and \scss problems, we 
also believe that studying the complexity of other problems on bidirected graphs 
is worthwhile. Many problems are substantially harder on directed graphs than 
on undirected graphs, and thus it is natural to ask about the complexity in 
bidirected graphs. For example for the \textsc{Multicut} problem an 
edge-weighted graph is given together with a list of terminal pairs, and the aim 
is to find a minimum weight cut so that no terminal pair remains connected. 
The seminal work of \citet{marx2014fixed} shows that this problem is FPT on 
undirected graphs when the parameter is the solution size, but on directed 
graphs the problem is W[1]-hard~\cite{marx2014fixed,pilipczuk2018directed}. In 
the latter case, also no $O(2^{\log^{1-\eps} n})$\hy{}approximation is possible 
in polynomial time, and there are indications that not even an 
$O(n^\delta)$-approximation is possible~\cite{chuzhoy2009polynomial} for some 
constant~$\delta>0$. It would be interesting to see whether considering 
bidirected graphs makes the problem such as \textsc{Multicut} easier than on 
directed graphs (in terms of approximation and/or parameterization).

\paragraph*{Acknowledgements.} We would like to thank Sasha Sami for pointing 
out a missing argument in the proof of \autoref{lem:bi-SCSS-struct}.

\bibliography{papers}
\bibliographystyle{plainnat}

\end{document}